\documentclass[10pt,reqno]{amsart}
\usepackage{amsfonts}
\usepackage{mathrsfs}
\usepackage{amsmath}
\usepackage{amssymb,amsfonts}
\usepackage[all]{xy}
\usepackage[mathscr]{euscript}
\usepackage{mathrsfs}
\usepackage{CJKutf8}

\pdfoutput=1

\newtheorem{thm}{Theorem}[section]
\newtheorem{lem}{Lemma}[section]
\newtheorem{defn}{Definition}[section]
\newtheorem{prop}{Proposition}[section]
\newtheorem{coro}{Corollary}[section]\numberwithin{equation}{section}
\newtheorem{rmk}{Remark}[section]
\newtheorem{example}{Example}

\newcommand{\mysection}[1]{\section{#1}\setcounter{equation}{0}}

\newfont{\bb}{msbm10 at 11pt}

\def\R{\hbox{\bb R}}

\def\C{\mathbb C}

\newcommand{\bal}{\begin{aligned}}      \newcommand{\eal}{\end{aligned}}
\newcommand{\ba}{\begin{array}}      \newcommand{\ea}{\end{array}}
\newcommand{\bc}{\begin{center}}     \newcommand{\ec}{\end{center}}
\newcommand{\be}{\begin{enumerate}}  \newcommand{\ee}{\end{enumerate}}
\newcommand{\beq}{\begin{eqnarray}}  \newcommand{\eeq}{\end{eqnarray}}
\newcommand{\beQ}{\begin{eqnarray*}} \newcommand{\eeQ}{\end{eqnarray*}}
\newcommand{\bi}{\begin{itemize}}    \newcommand{\ei}{\end{itemize}}
\newcommand{\bt}{\begin{tabular}}    \newcommand{\et}{\end{tabular}}
\newcommand{\bdm}{\begin{displaymath}} \newcommand{\edm}{\end{displaymath}}

\newcommand{\ls}{\setlength{\baselineskip}{12pt}
                 \setlength{\parskip}{3mm}}

\begin{document}
\begin{CJK}{UTF8}{gbsn}
\title[Nonexistence of Dirac Solutions]{Nonexistence of Time-periodic Solutions of the Dirac Equation in Kerr-Newman-(A)dS Spacetime}

\author[M.Z. Fan]{Mengzhang Fan$^{\dag}$}
\address[]{$^{\dag}$Institute of Mathematics, Academy of Mathematics and Systems Science, Chinese Academy of Sciences,
Beijing 100190, PR China and School of Mathematical Sciences, University of Chinese Academy of Sciences, Beijing 100049, PR China}
\email{fanmengzhang@amss.ac.cn}
\author [Y.H. Wang]{Yaohua Wang$^{\dag}$}
\address[]{$^{\dag}$School of Mathematics and Statistics, Henan University, Kaifeng 475004, PR China}
\email{wangyaohua@henu.edu.cn}
\author[X Zhang]{Xiao Zhang$^{\flat}$}
\address[]{$^{\flat}$Guangxi Center for Mathematical Research, Guangxi University, Nanning, Guangxi 530004, PR China}
\address[]{$^{\flat}$Institute of Mathematics, Academy of Mathematics and Systems Science, Chinese Academy of Sciences,
Beijing 100190, PR China and School of Mathematical Sciences, University of Chinese Academy of Sciences, Beijing 100049, PR China}
\email{xzhang@gxu.edu.cn, xzhang@amss.ac.cn}

\date{}

\begin{abstract}
In this paper, we study the nonexistence of nontrivial time-periodic solutions of the Dirac equation in Kerr-Newman-(A)dS spacetime. In the non-extreme Kerr-Newman-dS spacetime, we prove that there is no nontrivial $L^p$ integrable Dirac particle for arbitrary $(\lambda,p)\in \mathbb{R}\times[2,+\infty)$. In the extreme Kerr-Newman-dS and extreme Kerr-Newman-AdS spacetime, we show the equation relations between the energy eigenvalue $\omega$, the horizon radius, the angular momentum, the electric charge and the cosmological constant if there exists nontrivial $L^p$ integrable time-periodic solution of the Dirac equation, and further give the necessary conditions for the existence of nontrivial solutions. 
\end{abstract}

\maketitle \pagenumbering{arabic}

\mysection{Introduction}\ls

General relativity is a theory of space, time, and gravity proposed by Einstein in 1915. In general relativity, spacetime is a 4-dimensional manifold $M$ equipped with a Lorentzian metric $g$. The metric $g$ is determined by the stress-energy tensor $T$ by the following Einstein field equation \cite{10,wein}
\begin{equation}\label{efe}
    G_{\mu\nu}+\Lambda g_{\mu\nu}=\frac{8\pi G}{c^4}T_{\mu\nu},
\end{equation}
where the constant $\Lambda$ is the cosmological constant, $G_{\mu\nu}$ is the Einstein tensor with respect to the metric $g$
\begin{equation}
G_{\mu\nu}=R_{\mu\nu}-\frac{1}{2}Rg_{\mu\nu}.
\end{equation}
With the development of the times and the progress of science and technology, general relativity has passed almost all experimental tests (such as the deflection angle of light, the precession of Mercury, the gravitational redshift \cite{lcbls}, etc.). It is of great significance for the study of the structure and evolution of celestial bodies and the universe, and it is also the core theory of modern mathematical physics and other related fields. It is no exaggeration to say that one of the most exciting predictions of Einstein's theory of gravity is the existence of black holes, whose gravitational fields are so strong that no physical entity or signal can pull away from them or escape. In 2019, the first photograph of a black hole was taken, and Einstein's general relativity was once again confirmed. In addition, deep connections have been found between black hole theory and seemingly distant fields, such as thermodynamics, information theory, and quantum theory.

Quantum mechanics and general relativity are the two fundamental pillars of modern physics. In 1926, the Austrian physicist Schr\"{o}dinger proposed the following famous Schr\"{o}dinger equation in quantum mechanics \cite{hall}
\begin{equation}\label{schr}
    \frac{d}{dt}\psi=-\frac{i}{\hbar}H(\psi).
\end{equation}
However, the Schr\"{o}dinger equation (\ref{schr}) is not relativistically invariant, i.e. it is incompatible with Einstein's theory of relativity. In 1928, Dirac, a British theoretical physicist, proposed an equation that satisfies the invariance of relativity, that is, the Dirac equation \cite{faria}. Dirac constructed a first order differential operator $\mathscr{D}$, which is called the Dirac operator
\begin{equation}
    \mathscr{D}:=e^{\alpha}\partial_{\alpha},
\end{equation}
where $e^{\alpha}$ are the following $4\times 4$ matrices
\begin{equation}
    e^0=\begin{pmatrix}
0 &-\sigma^0\\
        -\sigma^0& 0
    \end{pmatrix},\quad e^{i}=\begin{pmatrix}
      0 & \sigma^{i}\\
      -\sigma^{i} & 0 \\
   \end{pmatrix},
\end{equation}
and $\sigma^{\alpha}$ are the following $2\times 2$ Pauli matrices
\begin{equation}
    \begin{split}
\sigma^0 =\begin{pmatrix}
      1 & 0\\
      0 & 1 \\
   \end{pmatrix},&\qquad  \sigma^1=\begin{pmatrix}
      -1 & 0\\
      0 & 1 \\
   \end{pmatrix},\\
        \sigma^2=\begin{pmatrix}
      0 & 1\\
      1 & 0 \\
   \end{pmatrix},&\qquad \sigma^3=\begin{pmatrix}
      0 & i\\
      -i & 0 \\
   \end{pmatrix}.
    \end{split}
\end{equation}
It is worth mentioning that the matrices $e^{\alpha}$, $\alpha=0,1,2,3$ satisfiy the following Clifford relations:
\begin{equation}
\begin{split}
    (e^0)^2=\begin{pmatrix}
1&0&0&0\\
0&1&0&0\\
0&0&1&0\\
 0&0&0&1\\       
    \end{pmatrix},\; (e^1)^2&=(e^2)^2=(e^3)^2=-\begin{pmatrix}
1&0&0&0\\
0&1&0&0\\
0&0&1&0\\
 0&0&0&1\\       
    \end{pmatrix},\\
e^{\alpha}e^{\beta}&=-e^{\beta}e^{\alpha}\,(\alpha\neq\beta).
    \end{split}
\end{equation}
By simple calculation, one can know that the square of the Dirac operator $\mathscr{D}$ is the wave operator $\square$, i.e.
\begin{equation}
    \mathscr{D}^2=\square:= \partial_t^2-\Delta.
\end{equation}
Unlike the Schr\"{o}dinger equation, the Dirac operator $\mathscr{D}$ act not on wave functions but on 4-dimensional vector-valued functions (spinors) $\Psi$. For arbitrary $\lambda\in \R$, the Dirac equation is defined as follows
\begin{equation}\label{diraceq}
    \big(\mathscr{D}+i\lambda \big)\Psi=0.
\end{equation}
It is not difficult to verify that the equation (\ref{diraceq}) is Lorentz invariant, i.e. it is compatible with special relativity. The Dirac equation describes the dynamics of particles with half-integer spin in curved spacetime, and the development of quantum field theory began.

In recent years, many scholars have studied the Dirac equation on curved spacetime, especially the separation of variables, which has been an obstacle to progress in many desired directions. In 1976, Chandrasekhar \cite{chand3,chand2} proved for the first time that under the Kerr spacetime background metric
\begin{equation}
    ds^2=-\frac{\Delta}{U}(dt-a\sin^2\theta d\varphi)^2+\frac{U}{\Delta}dr^2+Ud\theta^2+\frac{\sin^2\theta}{U}\left(a\,dt-(r^2+a^2)d\varphi\right)^2,
\end{equation}
the Dirac equation (\ref{diraceq}) can be separated into radial equations and angular equations, where
\begin{equation}
    U=r^2+a^2\cos^2\theta, \;\Delta=r^2-2mr+a^2.
\end{equation}
Subsequently, Page \cite{page} extends the conclusion of separating variables of the Dirac equation to charged, rotating black hole spacetime, namely Kerr-Newman spacetime.

Regarding the nonexistence of nontrivial time-periodic solutions of the Dirac equation (\ref{diraceq}), Finster et al. \cite{reissner} first considers the problem under the non-extreme and extreme Reissner-Nordstr\"{o}m spacetime with black holes
\begin{equation}
    ds^2=-\Big(1-\frac{2\rho}{r}+\frac{Q^2}{r^2} \Big)dt^2+\Big(1-\frac{2\rho}{r}+\frac{Q^2}{r^2} \Big)^{-1}dr^2+r^2\big(d\theta^2+\sin^2\theta d\varphi^2 \big).
\end{equation}
For the non-extreme spacetime, they define reasonable conditions to match the spinors inside and outside the black hole, and prove by reduction to absurdity that if a nontrivial solution exists, the normalization conditions are not satisfied; For the extreme spacetime, they obtained the contradiction mainly by analyzing the asymptotic behavior of the solution of the radial equation near the black hole $r=\rho$.

After that, when the spherical symmetry of the original Reissner-Nordstr\"{o}m spacetime is altered by changes in the metric and electromagnetic field, Finster el al. \cite{kn} consider the nonexistence of solutions of the Dirac equation (\ref{diraceq}) in the non-extreme Kerr-Newman spacetime (axisymmetric) with black holes
\begin{equation}
    ds^2=-\frac{\Delta}{U}\big(dt-a\sin^2\theta d\varphi\big)^2+U\Big(\frac{dr^2}{\Delta}+d\theta^2 \Big)+\frac{\sin^2\theta}{U}\big(a\,dt-(r^2+a^2)d\varphi \big)^2,
\end{equation}
where
\begin{equation}
    U=r^2+a^2\cos^2\theta,\; \Delta=r^2-2mr+a^2+Q^2
\end{equation}
and the parameters satisfy $m^2>a^2+Q^2$. By choosing the appropriate coordinate transformation, they give reasonable conditions to match the spinors inside and outside the black hole, and prove that the existence of nontrivial solutions is contradictory to the normalized conditions. Furthermore, they consider the nonexistence problem in the most general stationary axisymmetric metric 
\begin{equation}
    ds^2=T^{-2}\left[(L\,du+M\,dv)^2-(N\,du+P\,dv)^2-\frac{dw^2}{W(w)}-\frac{dx^2}{X(x)}\right]
\end{equation}
in which the Dirac equation can also be separated into radial and angular equations by Chandrasekhar's procedure, where the conformal factor $T^{-1}$ and the coefficients of the metric $L$, $M$, $N$, $P$ are functions of $x$ and $w$ only. Moreover, the following constraints are required \cite{14}
\begin{equation}
    \begin{split}
\partial_x\left(\frac{L}{LP-MN}\right)&=0,\\
\partial_x\left(\frac{M}{LP-MN}\right)&=0,\\
\partial_w\left(\frac{N}{LP-MN}\right)&=0,\\
    \partial_w\left(\frac{P}{LP-MN}\right)&=0.   
    \end{split}
\end{equation}

For the non-zero (negative) cosmological constant case, Wang and Zhang \cite{wyhzx} consider the nonexistence of $L^p$ integrable solutions of the Dirac equation (\ref{diraceq}) in the non-extreme Kerr-Newman-AdS spacetime with black holes
\begin{equation}
\begin{split}
    ds^2=&-\frac{\Delta_-(r)}{U}\Big(dt-\frac{a\sin^2\theta}{E_-}d\varphi \Big)^2+\frac{U}{\Delta_-(r)}dr^2+\frac{U}{\Delta_-({\theta})}d\theta^2\\
    &+\frac{\Delta_-({\theta})\sin^2\theta}{U}\Bigg(a\,dt-\frac{(r^2+a^2)}{E_-}d\varphi \Bigg)^2,
    \end{split}
\end{equation}
where
\begin{equation}
   U=r^2+a^2\cos^2\theta,\; \Delta_-({\theta})=1-\kappa^2a^2\cos^2\theta,
\end{equation}
$-3\kappa^2<0$ is the cosmological constant and the polynomial
\begin{equation}
    \Delta_-(r)=(r^2+a^2)(1+\kappa^2r^2)-2mr+Q^2+P^2
\end{equation}
of $r$ has two unequal positive real roots. After separating the variables, they prove that there is no $L^p$  integrable (on a slice of the spacetime where $t$ is equal to the constant and $r$ is large enough) nontrivial time-periodic solution of the Dirac equation  by analyzing the asymptotic behavior of the coefficients of the radial equation for sufficiently large $r$, this method recovers the same result of Belgiorno and Cacciatori \cite{bcads} in the case of $p=2$ by using the indirect method--spectral method. Moreover, they also study the nonexistence of $L^p$ integrable nontrivial time-periodic solutions of Dirac equation in  general stationary axisymmetric spacetime with negative cosmological constant. Similarly, for the positive cosmological constant case, Belgiorno and Cacciatori \cite{bcds} convert the nonexistence problem of nontrivial $L^2$ integrable (on the slice of spacetime between the event horizon and the cosmological horizon where $t=\text{const}$)  time-periodic solution of the Dirac equation in the Kerr-Newman-dS spacetime with black holes
\begin{equation}
\begin{split}
    ds^2=&-\frac{\Delta_+(r)}{U}\Big(dt-\frac{a\sin^2\theta}{E_+}d\varphi\Big)^2+\frac{U}{\Delta_+(r)}dr^2+\frac{U}{\Delta_+({\theta})}d\theta^2\\
    &+\frac{\Delta_+(\theta)\sin^2\theta}{U}\Bigg(a\,dt-\frac{(r^2+a^2)}{E_+}d\varphi \Bigg)^2
    \end{split}
\end{equation}
to the nonexistence problem of quantum bound states of Dirac Hamiltonian and then give the proof by using spectral methods, where
\begin{equation}
   \Delta_+({\theta})=1+\kappa^2a^2\cos^2\theta,
\end{equation}
$3\kappa^2>0$ is the cosmological constant and the polynomial 
\begin{equation}
    \Delta_+(r)=(r^2+a^2)(1-\kappa^2r^2)-2mr+Q^2+P^2
\end{equation}
of $r$ has 4 unequal real roots (3 positive and 1 negative).

For the existence of nontrivial normalizable time-periodic solutions of the Dirac equation (\ref{diraceq}), Schmid \cite{extremekerr} gave the proof under the extreme Kerr spacetime background metric where the mass of the black hole and angular momentum satisfy certain values.

Based on this research background, in this paper, we mainly study the nonexistence of nontrivial $L^p$ integrable time-periodic solutions of the Dirac equation in the non-extreme Kerr-Newman-dS spacetime, and also the necessary conditions for the existence of nontrivial $L^p$ integrable time-periodic solutions of the Dirac equation in the extreme Kerr-Newman-(A)dS spacetime. This means that with further perturbation of the spacetime background metric, that is, from the zero cosmological constant to the nonzero cosmological constant, the conclusion that the nontrivial time-periodic solution of the Dirac equation does not exist still holds, so that the Dirac particles satisfying the agreed conditions will either disappear into the black hole or escape to infinity.

The paper is organized as follows. In Section 2, we give the definition of spin structures on 4-dimensional orientable spacetime manifold $M$ and define the spinor bundle  $\Sigma M$ on $M$ by the complex spin representation. Then we show the local expression of the spinorial connection on the spinor bundle $\Sigma M$ and give the definition of Dirac operator. We also introduce the existence and uniqueness theorem for solutions of ordinary differential equations. In Section 3, we separate the Dirac equation in Kerr-Newman-dS spacetime into radial equations and angular equations by the method of Chandrasekhar \cite{chand3,chand2}. After that, by analyzing the asymptotic behaviour of the solution of the radial equation near the black hole, we show that there is no nontrivial $L^p$ integrable time-periodic solutions of the Dirac equation in the non-extreme Kerr-Newman-dS spacetime. In Section 4 and Section 5, by changing of variables of the radial equations and analyzing the corresponding solutions near the horizon, we show that if there exist nontrivial $L^p$ integrable time-periodic solutions of the Dirac equation in the extreme Kerr-Newman-dS spacetime and the extreme Kerr-Newman-AdS spacetime, then the energy eigenvalue $\omega$ must have certain equation relations with the horizon radius, the angular momentum, the electric charge and the cosmological constant. By this, we further give the necessary conditions for the existence of nontrivial solutions. In Section 6, we summarize the full paper and raise some questions to be further studied.

\mysection{Preliminaries}\ls

\subsection{Spin geometry on Lorentzian manifold}

In this subsection, we mainly introduce the spin structure and spinor vector bundle on Lorentzian manifold. Moreover, we calculate the spinorial Levi-Civita connection on spinor bundle and give the definition of Dirac operator. For basic facts about spin geometry, we refer to \cite{lawson}, \cite{hijazi}. For basic facts about the Dirac operator on Lorentzian manifold, we refer to \cite{xz}, \cite{frankel}.

\subsubsection{Spin structure}

Let $V$ be a real n-dimensional linear vector space, $g_V$ is a non-degenerate symmetric bilinear form on $V$, i.e.
\begin{equation}
g_V:V\times V \longrightarrow \R 
\end{equation}
satisfying

\noindent (1) for any $v_1,v_2\in V$, there holds $g_V(v_1,v_2)=g_V(v_2,v_1)$ (For the sake of notation, sometimes we denote $\big<v_1,v_2\big>\triangleq g_V(v_1,v_2)$;

\noindent (2) if for all $w\in V$, we have $g_V(v,w)=0$, then $v=0$.

\noindent Since $g_V(v,v)$ might be negative, the norm of $v$, i.e. $|v|$ is defined as
\begin{equation}
    |v|:=|g_V(v,v)|^{1/2}.
\end{equation}
A vector $v$ is called a unit vector if $|v|=1$, i.e. $g_V(v,v)=\pm 1$. We call a set of unit vectors orthogonal to each other orthonormal. It is easy to see that for any symmetric non-degenerate bilinear form on the nontrivial vector space $V$, there exists an orthonormal bases $e_1,\dots,e_n$ satisfying
\begin{equation}\label{orth}
    g_V(e_i,e_j)=\delta_{ij}\epsilon_i,
\end{equation}
where $\epsilon_i=\pm 1$, $i=1,\dots,n$. Besides, any orthonormal bases in $V$ satisfying (\ref{orth}) have the same sign $\{\epsilon_i\}$ (by different order), the proof can be found in \cite{neill}. We arrange the symbol $\epsilon_i$ according to the principle of minus sign before $(\epsilon_1,\dots,\epsilon_n)$, and we say that the number of negative indicators in this permutation is the indicator of the binary $(V,g_V)$.

\begin{defn}
   For $n\ge2$, the binary $(V,g_V)$ is called a Lorentzian vector space if the indicator is $1$.
\end{defn}
A vector $v$ in the Lorentzian vector space $V$ is spacelike, if $g_V(v,v)>0$ or $v=0$; a vector $v$  is lightlike if $g_V(v,v)=0$ and $v\neq 0$; a vector $v$ is timelike, if $g_V(v,v)<0$.

Let $g$ be a smooth tensor field of type (0,2) on a n-dimensional smooth manifold $M$ such that for every $p\in M$, the binary $(T_pM,g_p)$ is non-degenerate. If each binary $(T_pM,g_p)$ is a Lorentzian vector space, then we say that $(M,g)$ is a Lorentzian manifold. Naturally, the binary  $(T_p^*M,g_p)$ is also a Lorentzian vector space for any $p\in M$, where $T_p^*M$ is the cotangent space at $p$.

In this paper, we mainly consider 4-dimensional Lorentzian manifold, thus we make the following definition:

\begin{defn}
   A spacetime is a connected and time orientable (i.e. there exists smooth timelike vector field) 4-dimensional Lorentzian manifold.
\end{defn}

Next, we introduce the spin structure on spacetime. For this purpose, we need the following algebra preparations.

Let $(V,g_V)$ be a 4-dimensional Lorentzian vector space (in order to ensure the unity of symbols of the spinorial connection and the Dirac operator, we consider on the cotangent bundle). Let
\begin{equation}
    \mathcal{J}(V):=\sum_{i=0}^{\infty}\bigotimes\nolimits^iV
\end{equation}
be the tensor algebra of $V$, $\mathcal{I}_g(V)$ be the ideal in $\mathcal{J}(V)$ generated by $\big\{v\otimes v+g_V(v,v)1\,\big|v\in V\big\}$. Then the Clifford algebra with respect to the Lorentzian vector spac $(V,g_V)$ is defined as
\begin{equation}
    \text{C}\ell(V,g_V):=\mathcal{J}(V)/\mathcal{I}_g(V).
\end{equation}
By definition, it is no hard to see that there is a nature embedding from $V$ to $\text{C}\ell(V,g_V)$, i.e.
\begin{equation}    V=\bigotimes\nolimits^1V\hookrightarrow\mathcal{J}(V)\xrightarrow[]{\pi}\text{C}\ell(V,g_V),
\end{equation}
where $\pi$ is the canonical projection.
\begin{rmk}
    As vector spaces, the Clifford algebra $\text{C}\ell(V,g_V)$ is isomorphic to the exterior algebra $\Lambda^*V=\bigoplus_{k=0}^{\infty}\Lambda^kV$ of $V$. Hence, $\text{C}\ell(V,g_V)$ is a finite dimensional vector space \cite{lawson}.
\end{rmk}

The Clifford algebra $\text{C}\ell(V,g_V)$ can be generated by the vector space $V$ and the unit $1$, and the elements satisfy the following Clifford multiplication rule, i.e. 
\begin{equation}
    v\cdot w+w\cdot v=-2g_V(v,w)
\end{equation}
for arbitrary $v,w\in V$. Therefore, if $\{e^0,e^1,e^2,e^3\}$ is an orthonormal bases of $(V,g_V)$ satisfying
\begin{equation}
\big<e^0,e^0\big>=-1,\big<e^1,e^1\big>=1,\big<e^2,e^2\big>=1,\big<e^3,e^3\big>=1,
\end{equation}
then the correspondinng Clifford multiplication is
\begin{equation}\label{cliffordm}
    e^0\cdot e^0=1,e^1\cdot e^1=-1,e^2\cdot e^2=-1,e^3\cdot e^3=-1,e^{\alpha}\cdot e^{\beta}=-e^{\beta}\cdot e^{\alpha}\,(\alpha\neq \beta).
\end{equation}

The following example shows a matrix representation of $\text{C}\ell(V,g_V)$.

\begin{example}\label{exam}
    \begin{equation}\label{examp}
    \begin{split}
   e^0\longmapsto \begin{pmatrix}
      0 & 0 & -1 & 0\\
      0 & 0 & 0 & -1\\
      -1 & 0 & 0 & 0\\
      0 & -1 & 0 & 0\\
   \end{pmatrix}, & \qquad e^1\longmapsto \begin{pmatrix}
      0 & 0 & -1 & 0\\
      0 & 0 & 0 & 1\\
      1 & 0 & 0 & 0\\
      0 & -1 & 0 & 0\\
   \end{pmatrix},\\
   e^2\longmapsto \begin{pmatrix}
      0 & 0 & 0 & 1\\
      0 & 0 & 1 & 0\\
      0 & -1 & 0 & 0\\
      -1 & 0 & 0 & 0\\
   \end{pmatrix}, &\qquad e^3\longmapsto \begin{pmatrix}
      0 & 0 & 0 & i\\
      0 & 0 & -i & 0\\
      0 & -i & 0 & 0\\
      i & 0 & 0 & 0\\
   \end{pmatrix}.
   \end{split}
   \end{equation}
   It is easy to verify that the above matrix multiplications satisfy the Clifford multiplication rule (\ref{cliffordm}).
\end{example}

Next we consider the complex Clifford algebra and the corresponding representation.
\begin{defn}
\begin{equation}
\mathbb{C}\ell(V,g_V):=\text{C}\ell(V,g_V)\otimes_{\R}\mathbb{C}.
    \end{equation}
\end{defn}
By the universal property of Clifford algebra (c.f. Proposition 1.1 in \cite{lawson}), the following isomorphism holds
\begin{equation}
\mathbb{C}\ell(V,g_V)=\text{C}\ell(V,g_V)\otimes_{\R}\mathbb{C}\cong \text{C}\ell(\mathbb{C}^4,g_V\otimes\mathbb{C})\triangleq \mathbb{C}\ell_4.
\end{equation}
\begin{defn}
    Let the $\C$-algebra homomorphism
    \begin{equation}\label{rep}
    \rho:\mathbb{C}\ell_4\longrightarrow \text{Hom}_{\mathbb{C}}(W,W)
    \end{equation}
    be a complex representation of $\mathbb{C}\ell_4$, where $W$ is a finite dimensional complex vector space. We say that such a representation is reducible if and only if the vector space $W$ can be decomposed into the following nontrivial direct sum (over $\mathbb{C}$)
    \begin{equation}
        W=W_1\oplus W_2
    \end{equation}
    such that for any $u\in \C \ell_4$, there holds
    \begin{equation}
        \rho(u)(W_i)\subseteq W_i, \,i=1,2.
    \end{equation}
We say that such a complex representation is irreducible if it is not reducible.
\end{defn}
According to Theorem 5.7 of Chapter 1 in \cite{lawson}, since $n$ is even, the complex Clifford algebra $\mathbb{C}\ell_n$ has a unique (up to equivalence) irreduciable complex representation. Hence, we know that the irreduciable complex representation (\ref{rep}) is unique. The following proposition gives a concrete characterization of the unique irreducible complex representation of $\mathbb{C}\ell(V,g_V)$.

\begin{prop}\label{complexrep}
    The complex representation
    \begin{equation}\label{complexrep1}
\rho:Span\big\{e^0,e^1,e^2,e^3\big\}\otimes_{\R}\C
\longrightarrow \text{End}(\C^4) 
\end{equation}
is irreduciable, where $e^{\alpha}$ ($\alpha=0,1,2,3$) are the matrices in (\ref{examp}), and for any complex 4-dimensional vector $\psi$, there holds
\begin{equation}
    \rho(e^{\alpha})\psi:=e^{\alpha}\psi.
\end{equation}
\end{prop}
\begin{proof}
In order to show that $\rho$ is irreduciable, we only need to prove that $\rho$ is surjective \cite{hassani}. Since
\begin{equation}
     Span\big\{e^0,e^1,e^2,e^3\big\}\otimes_{\R}\C\cong \mathbb{C}\ell(V,g_V)\cong \Lambda^*V,
\end{equation}
Therefore
\begin{equation}
\begin{split}
    \text{dim}\Big(Span\big\{e^0,e^1,e^2,e^3\big\}\otimes_{\R}\C\Big) &=\text{dim} \Big(\Lambda^*V\Big)\\
&=2^4\\
&=\text{dim}\Big(\text{End}(\C^4)\Big).
    \end{split}
\end{equation}
Hence, in order to prove that $\rho$ is  surjective, it is only necessary to show that $\rho$ is  injective, i.e. $\rho(v)=\textbf{0}$ implies that $v=0$. In fact, choose a base of $Span\big\{e^0,e^1,e^2,e^3\big\}\otimes_{\R}\C$
\begin{equation}
    \big\{1,\,e^{\alpha_1}e^{\alpha_2}\dots e^{\alpha_k}\big\},
\end{equation}
where $\{\alpha_1,\dots,\alpha_k\}\subset \{1,\dots,4\}$ and $\alpha_1<\alpha_2<\cdots<\alpha_k$, then we have
\begin{equation}\label{injec}
a_0+a_{\alpha_1,\dots,\alpha_k}e^{\alpha_1}e^{\alpha_2}\dots e^{\alpha_k}=\textbf{0}.
\end{equation}
Substituting (\ref{examp}) into the above formula (\ref{injec}), by simple calculations, we haveit follows that
\begin{equation}
    a_0=a_{\alpha_1,\dots,\alpha_k}=0.
\end{equation}
This completes the proof of the proposition.
    
\end{proof}

\begin{defn}
    The Clifford multipulation is defined as the following map
    \begin{equation}\label{clim}
        \begin{split}
\mathfrak{m}_V:\mathbb{C}\ell(V,g_V)\times \C^4 &\longrightarrow\C^4 \\
            (e,\psi)& \longmapsto \rho(e)\psi.
        \end{split}
    \end{equation}
\end{defn}

Let $\mathscr{M}_{4}(\R)$ be the space consisting of $4\times4$ real matrices. The Lorentzian group $O(1,3)$ is defined as
\begin{equation}
    O(1,3):=\left\{ A\in \mathscr{M}_{4}(\R)\Bigg|A^T\begin{pmatrix}
      -1 & 0 & 0 & 0\\
      0 & 1 & 0 & 0\\
      0 & 0 & 1 & 0\\
      0 & 0 & 0 & 1\\
   \end{pmatrix}A= \begin{pmatrix}
      -1 & 0 & 0 & 0\\
      0 & 1 & 0 & 0\\
      0 & 0 & 1 & 0\\
      0 & 0 & 0 & 1\\
   \end{pmatrix}\right\}.
\end{equation}
So the Lorentzian group $O(1,3)$ is a matrix Lie group consisting of all norm preserving linear transformation in Minkowski spacetime $\R^4$. In particular, we denote
\begin{equation}
    SO(1,3):=\left\{A\in O(1,3)\Big|\, \text{det}A=1 \right\}.
\end{equation}
Since $SO(1,3)$ has two connected components \cite{helgason}, let $SO_0(1,3)$ be the component containing the identity and $\mathfrak{s}\mathfrak{o}(1,3)$ be the corresponding Lie algebra. According to the closed subgroup Theorem \cite{helgason}, we can deduce that
\begin{equation}
    \mathfrak{s}\mathfrak{o}(1,3)=\left\{X\in \mathscr{M}_4(\R):\,\exp\,tX\in SO_0(1,3),\,\forall t\in \R\right\}.
\end{equation}
Therefore, for any $X\in \mathfrak{s}\mathfrak{o}(1,3)$, we have
\begin{equation}\label{so13}
    \big(\exp\,tX^T\big)\begin{pmatrix}
      -1 & 0 & 0 & 0\\
      0 & 1 & 0 & 0\\
      0 & 0 & 1 & 0\\
      0 & 0 & 0 & 1\\
   \end{pmatrix}\big(\exp\,tX\big)=\begin{pmatrix}
      -1 & 0 & 0 & 0\\
      0 & 1 & 0 & 0\\
      0 & 0 & 1 & 0\\
      0 & 0 & 0 & 1\\
   \end{pmatrix}.
\end{equation}
Taking the derivative of $t$ on both sides of (\ref{so13}) and considering the value at $t=0$, it follows that
\begin{equation}
    \begin{pmatrix}
      -1 & 0 & 0 & 0\\
      0 & 1 & 0 & 0\\
      0 & 0 & 1 & 0\\
      0 & 0 & 0 & 1\\
   \end{pmatrix}X+X^T\begin{pmatrix}
      -1 & 0 & 0 & 0\\
      0 & 1 & 0 & 0\\
      0 & 0 & 1 & 0\\
      0 & 0 & 0 & 1\\
   \end{pmatrix}=0.
\end{equation}
By simple calculations, $X$ has the following form
\begin{equation}
    \begin{pmatrix}
      0 & x_{12} & x_{13} & x_{14}\\
      x_{12} & 0 & x_{23} & x_{24}\\
      x_{13} & -x_{23} & 0 & x_{34}\\
      x_{14} & -x_{24} & -x_{34} & 0\\
   \end{pmatrix}.
\end{equation}
Thus, the Lie algebra $\mathfrak{s}\mathfrak{o}(1,3)$ has the following $6$ basis matrices:
\begin{equation*}
    \begin{split}
&\begin{pmatrix}
      0 & 1 & 0 & 0\\
      1 & 0 & 0 & 0\\
      0 & 0 & 0 & 0\\
      0 & 0 & 0 & 0\\
   \end{pmatrix},\quad \begin{pmatrix}
      0 & 0 & 1 & 0\\
      0 & 0 & 0 & 0\\
      1 & 0 & 0 & 0\\
      0 & 0 & 0 & 0\\
   \end{pmatrix},\quad \begin{pmatrix}
      0 & 0 & 0 & 1\\
      0 & 0 & 0 & 0\\
      0 & 0 & 0 & 0\\
      1 & 0 & 0 & 0\\
   \end{pmatrix}, \\
   &\begin{pmatrix}
      0 & 0 & 0 & 0\\
      0 & 0 & 1 & 0\\
      0 & -1 & 0 & 0\\
      0 & 0 & 0 & 0\\
   \end{pmatrix},\quad \begin{pmatrix}
      0 & 0 & 0 & 0\\
      0 & 0 & 0 & 1\\
      0 & 0 & 0 & 0\\
      0 & -1 & 0 & 0\\
   \end{pmatrix},\quad\begin{pmatrix}
      0 & 0 & 0 & 0\\
      0 & 0 & 0 & 0\\
      0 & 0 & 0 & 1\\
      0 & 0 & -1 & 0\\
   \end{pmatrix}.
    \end{split}
\end{equation*}

Next, we define the group $Pin(1,3)$ and  $Spin(1,3)$ by the Clifford algebra ${C}\ell(V,g_V)$ in order to construct the 2-fold covering space of $SO_0(1,3)$. Let the multiplicative unit group of the Clifford algebra ${C}\ell(V,g_V)$ be
\begin{equation}
    {C}\ell^{\times}(V,g_V):=\left\{v\in  {C}\ell(V,g_V):\exists\, v^{-1}\; \text{satisfying}\; v^{-1}\cdot v=v\cdot v^{-1}=1\right\}.
\end{equation}
Since for all vectors $v$ in $V$ satisfying $g_V(v,v)\neq 0$, there holds
\begin{equation}
    v\cdot \frac{v}{-g_V(v,v)}=1,
\end{equation}
such vectors are contained in ${C}\ell^{\times}(V,g_V)$ and we let $P(V,g_V)$ be the subgroup  generated by these vectors. Moreover, ${C}\ell^{\times}(V,g_V)$ is a $2^4$-dimensional Lie group inheriting the topology of ${C}\ell(V,g_V)$ (as vector space), and the corresponding Lie algebra $\mathfrak{cl}^{\times}(V,g_V)$ is isomorphic to ${C}\ell(V,g_V)$.

\begin{defn}\label{pinspin}
The group $Pin(1,3)$ is generated by the elements in $P(V,g_V)$ satisfying $g_V(v,v)=\pm 1$ ($v\in V$). The group $Spin(1,3)$ is defined as
\begin{equation}
    Spin(1,3):=Pin(1,3)\cap {C}\ell^0(V,g_V),
\end{equation}
where $C\ell^0(V,g_V):=\left\{v\in C\ell(V,g_V):\varphi(v)=v \right\}$ and $\varphi$ is the following endomorphism
\begin{equation}
    \begin{split}
\varphi:C\ell(V,g_V)&\longrightarrow C\ell(V,g_V) \\
        v&\longmapsto -v \quad (v\in V).
    \end{split}
\end{equation}
\end{defn}

\begin{rmk}
    According to Definition \ref{pinspin}, we can see that
    \begin{equation}
        Spin(1,3)=\left\{v_1\cdots v_k\in Pin(1,3): k\equiv 0(mod\,2) \right\}.
    \end{equation}
    Moreover, $Spin(1,3)$ is a simple connected Lie group and the following isomorphism holds
    \begin{equation}
        Spin(1,3)\cong SL(2,\C):=\left\{A\in \text{Hom}_{\C}\left(\C^2,\C^2\right)\Big|\,\det A=1\right\}.
    \end{equation}
\end{rmk}

\begin{defn}\label{spinrep}
    The complex irreduciable representation $\rho$ in Proposition \ref{complexrep} restricting on $Spin(1,3)$, i.e.
    \begin{equation}\label{spinrepeq}
        \rho:Spin(1,3)\longrightarrow \text{End}(\C^4)
    \end{equation}
    is called the complex spinor representation. In this case, we call $\C^4$ the spinor space.
\end{defn}

The following theorem gives the 2-fold covering space of $SO_0(1,3)$, for the proof, we refer to Theorem 2.10 of Chapter 1 in \cite{lawson}.

\begin{thm}
    The following short exact sequence holds
    \begin{equation}
        0\rightarrow \mathbb{Z}_2\rightarrow Spin(1,3)\xrightarrow[]{\widetilde{Ad}} SO_0(1,3)\rightarrow 1,
    \end{equation}
    where $\widetilde{Ad}_v$ is defined as
    \begin{equation}\label{defad}
        \widetilde{Ad}_v(w):=w-2\frac{g_V(v,w)}{g_V(v,v)}v
    \end{equation}
for all $v$, $w\in V$ satisfying $g_V(v,v)\neq 0$.
\end{thm}

For the purpose of giving the local formula of the spinorial Levi-Civita connection on the spinor bundle, we calculate the concrete expression of the tangent map $\widetilde{Ad}_{*}: \mathfrak{spin(1,3)}\xrightarrow[]{\cong}\mathfrak{so}(1,3)$. The proposition is as follows:
\begin{prop}\label{ad_*}
\begin{equation}
    \widetilde{Ad}_{*}(e^0\cdot e^1)=-2\begin{pmatrix}
      0 & 1 & 0 & 0\\
      1 & 0 & 0 & 0\\
      0 & 0 & 0 & 0\\
      0 & 0 & 0 & 0\\
   \end{pmatrix},\quad  \widetilde{Ad}_{*}(e^0\cdot e^2)=-2\begin{pmatrix}
      0 & 0 & 1 & 0\\
      0 & 0 & 0 & 0\\
      1 & 0 & 0 & 0\\
      0 & 0 & 0 & 0\\
   \end{pmatrix},
   \end{equation}
   \begin{equation}
   \qquad\;\widetilde{Ad}_{*}(e^0\cdot e^3)=-2\begin{pmatrix}
      0 & 0 & 0 & 1\\
      0 & 0 & 0 & 0\\
      0 & 0 & 0 & 0\\
      1 & 0 & 0 & 0\\
   \end{pmatrix},\quad \widetilde{Ad}_{*}(e^1\cdot e^3)=2\begin{pmatrix}
      0 & 0 & 0 & 0\\
      0 & 0 & 0 & -1\\
      0 & 0 & 0 & 0\\
      0 & 1 & 0 & 0\\
   \end{pmatrix},
\end{equation}
\begin{equation}
   \widetilde{Ad}_{*}(e^1\cdot e^2)=2\begin{pmatrix}
      0 & 0 & 0 & 0\\
      0 & 0 & -1 & 0\\
      0 & 1 & 0 & 0\\
      0 & 0 & 0 & 0\\
   \end{pmatrix},\quad \widetilde{Ad}_{*}(e^2\cdot e^3)=2\begin{pmatrix}
      0 & 0 & 0 & 0\\
      0 & 0 & 0 & 0\\
      0 & 0 & 0 & -1\\
      0 & 0 & 1 & 0\\
   \end{pmatrix}.
\end{equation}   
\end{prop}

\begin{proof}
Due to the particularity of the vector $e^0$ in the Lorentzian vector space, we only need to calculate $\widetilde{Ad}_{*}(e^0\cdot e^1)$ and $\widetilde{Ad}_{*}(e^1\cdot e^3)$, and the remaining terms can be obtained by completely similar discussions.

(i) $ \widetilde{Ad}_{*}(e^0\cdot e^1)$:

For $x\in (-\infty,+\infty)$, 
the hyperbolic sine functions and the hyperbolic cosine functions are defined as follows
\begin{equation}
    \sinh{x}=\frac{e^x-e^{-x}}{2},\cosh{x}=\frac{e^x+e^{-x}}{2}.
\end{equation}
Notice that
\begin{equation}
\begin{split}
g_V&\Big((\cosh{x})e^0+(\sinh{x})e^1,(\cosh{x})e^0+(\sinh{x})e^1\Big)\\
&=-\cosh{x}^2+\sinh{x}^2\\
&=-1,
    \end{split}
\end{equation}
hence
\begin{equation}
    \begin{split}
\gamma_1: (-\infty,+\infty)&\longrightarrow Spin(1,3)\\
       x&\longmapsto e^0\cdot\left((\cosh{x})e^0+(\sinh{x})e^1\right)
    \end{split}
\end{equation}
is an one parameter subgroup of the group $Spin(1,3)$, and its tangent vector at the identity, i.e. at $x=0$ is
\begin{equation}
\begin{split}
    &\frac{d}{dx}\Big|_{x=0}\left(\cosh{x}+\sinh{x}e^0\cdot e^1\right)\\
&=\left(\sinh{x}+\cosh{x}e^0\cdot e^1\right)\Big|_{x=0}\\
&=e^0\cdot e^1.
    \end{split}
\end{equation}

Now for any $w\in V$, $w=w_{\alpha}e^{\alpha}$, by equation (\ref{defad}) we have
\begin{equation}
\begin{split}
\widetilde{Ad}_{\cosh{x}e^0+\sinh{x}e^1}w&=w+2\left(-w_0\cosh{x}+w_1\sinh{x}\right)\begin{pmatrix}
\cosh{x}\\
\sinh{x}\\
0\\
   0 
\end{pmatrix}\\
&=\begin{pmatrix}
w_0+\left(-2w_0\cosh{x}+2w_1\sinh{x}\right)\cosh{x}\\
w_1+\left(-2w_0\cosh{x}+2w_1\sinh{x}\right)\sinh{x}\\
w_2\\
   w_3 
\end{pmatrix}.
\end{split}
\end{equation}
Therefore,
\begin{equation}
    \begin{split}
\widetilde{Ad}_{\gamma_1(x)}w&=\widetilde{Ad}_{e^0\cdot\left(\cosh{x}e^0+\sinh{x}e^1\right)}w\\
   &=\widetilde{Ad}_{e^0}\begin{pmatrix}
w_0+\left(-2w_0\cosh{x}+2w_1\sinh{x}\right)\cosh{x}\\
w_1+\left(-2w_0\cosh{x}+2w_1\sinh{x}\right)\sinh{x}\\
w_2\\
   w_3 
\end{pmatrix}\\
   &=\begin{pmatrix}
w_0+\left(-2w_0\cosh{x}+2w_1\sinh{x}\right)\cosh{x}\\
w_1+\left(-2w_0\cosh{x}+2w_1\sinh{x}\right)\sinh{x}\\
w_2\\
   w_3 
\end{pmatrix}\\
   &\qquad\qquad-2\left(w_0+(-2w_0\cosh{x}+2w_1\sinh{x})\cosh{x})\right)\begin{pmatrix}
      1 \\
      0 \\
      0 \\
      0 \\
   \end{pmatrix}\\
   &=\begin{pmatrix}
    -w_0+(2w_0\cosh{x}-2w_1\sinh{x})\cosh{x}  \\
     w_1+(-2w_0\cosh{x}+2w_1\sinh{x})\sinh{x}  \\
      w_2 \\
      w_3 \\
   \end{pmatrix},
    \end{split}
\end{equation}
i.e.
\begin{equation}
    \widetilde{Ad}_{\gamma_1(x)}w=\begin{pmatrix}
    -1+2\left(\cosh{x}\right)^2  & -2\sinh{x}\cosh{x} & 0 & 0\\
     -2\sinh{x}\cosh{x}  & 1+2\left(\sinh{x}\right)^2 & 0 & 0\\
      0 & 0 & 1 & 0\\
      0 & 0 & 0 & 1\\
   \end{pmatrix}\begin{pmatrix}
      w_0 \\
      w_1 \\
      w_2 \\
      w_3 \\
   \end{pmatrix}.
\end{equation}
Thus,
\begin{equation}
    \begin{split}
\widetilde{Ad}_*(e^0\cdot e^1)&=\frac{d}{dx}\Big|_{x=0}\begin{pmatrix}
    -1+2\left(\cosh{x}\right)^2  & -2\sinh{x}\cosh{x} & 0 & 0\\
     -2\sinh{x}\cosh{x}  & 1+2\left(\sinh{x}\right)^2 & 0 & 0\\
      0 & 0 & 1 & 0\\
      0 & 0 & 0 & 1\\
   \end{pmatrix} \\
        &=-2\begin{pmatrix}
      0 & 1 & 0 & 0\\
      1 & 0 & 0 & 0\\
      0 & 0 & 0 & 0\\
      0 & 0 & 0 & 0\\
   \end{pmatrix}.
    \end{split}
\end{equation}

(ii) $ \widetilde{Ad}_{*}(e^1\cdot e^3)$:

For $\theta\in \R$, notice that
\begin{equation}
\begin{split}
    g_V\big(-\cos(\theta) e^1+\sin(\theta) e^3,-\cos(\theta) e^1+\sin(\theta) e^3\big) \\
=\cos^2\theta+\sin^2\theta=1,
    \end{split}
\end{equation}
thus
\begin{equation}
    \begin{split}
\gamma_2: (-\infty,+\infty)&\longrightarrow Spin(1,3)\\
        \theta&\longmapsto e^1\cdot\big(-\cos(\theta) e^1+\sin(\theta) e^3\big)
    \end{split}
\end{equation}
is an one parameter subgroup of the group $Spin(1,3)$, and its tangent vector at the identity $1$, i.e. at $\theta=0$ is
\begin{equation}
    \frac{d}{d\theta}\Big|_{\theta=0}\Big(\cos\theta+\sin(\theta) e^1\cdot e^3\Big)=e^1\cdot e^3.
\end{equation}

For any $w\in V$, $w=w_{\alpha}e^{\alpha}$, by equation (\ref{defad}) we can deduce that the following holds
\begin{equation}
\begin{split}
    \widetilde{Ad}_{-\cos(\theta) e^1+\sin(\theta) e^3}w&=w-2(-\cos(\theta)w_1+\sin(\theta)w_3)\begin{pmatrix}
      0 \\
      -\cos\theta \\
      0 \\
      \sin\theta \\
   \end{pmatrix}\\
&=\begin{pmatrix}
w_0 \\
w_1+2\cos\theta\left(-\cos(\theta)w_1+\sin(\theta)w_3\right) \\
w_2 \\
 w_3-2\sin\theta\left(-\cos(\theta)w_1+\sin(\theta)w_3\right)   
\end{pmatrix}.
   \end{split}
\end{equation}
Therefore,
\begin{equation}
\begin{split}
\widetilde{Ad}_{\gamma_2(\theta)}w&=\widetilde{Ad}_{e^1\cdot\big(-\cos(\theta) e^1+\sin(\theta) e^3\big)}w\\
&=\widetilde{Ad}_{e^1}\begin{pmatrix}
w_0 \\
w_1+2\cos\theta\left(-\cos(\theta)w_1+\sin(\theta)w_3\right) \\
w_2 \\
 w_3-2\sin\theta\left(-\cos(\theta)w_1+\sin(\theta)w_3\right)   
\end{pmatrix}\\
   &=\begin{pmatrix}
w_0 \\
w_1+2\cos\theta\left(-\cos(\theta)w_1+\sin(\theta)w_3\right) \\
w_2 \\
 w_3-2\sin\theta\left(-\cos(\theta)w_1+\sin(\theta)w_3\right)   
\end{pmatrix}\\
   &\qquad\qquad-2\big(w_1-2\cos^2(\theta)w_1+2w_3\sin\theta\cos\theta \big)\begin{pmatrix}
      0 \\
      1 \\
      0 \\
      0 \\
   \end{pmatrix}\\
   &=\begin{pmatrix}
      1 & 0 & 0 & 0\\
      0 & -1+2\cos^2\theta & 0 & -\sin2\theta\\
      0 & 0 & 1 & 0\\
      0 & \sin2\theta & 0 & 1-2\sin^2\theta\\
   \end{pmatrix}\begin{pmatrix}
      w_0 \\
      w_1 \\
      w_2 \\
      w_3 \\
   \end{pmatrix}.
\end{split}
\end{equation}
Hence we have
\begin{equation}
\begin{split}
    \widetilde{Ad}_{*}(e^1\cdot e^3)&=\frac{d}{d\theta}\Big|_{\theta=0}\begin{pmatrix}
      1 & 0 & 0 & 0\\
      0 & -1+2\cos^2\theta & 0 & -\sin(2\theta)\\
      0 & 0 & 1 & 0\\
      0 & 2\sin\theta\cos\theta & 0 & 1-2\sin^2\theta\\
   \end{pmatrix}\\
&=2\begin{pmatrix}
      0 & 0 & 0 & 0\\
      0 & 0 & 0 & -1\\
      0 & 0 & 0 & 0\\
      0 & 1 & 0 & 0\\
   \end{pmatrix},
   \end{split}
\end{equation}
which completes the proof of the proposition.
  
\end{proof}

Next, before defining the spin structure, we briefly introduce some basic facts of principal fibre bundle. We refer the reader to the references \cite{kobayashi,loring}.

Let $E$, $F$ and $M$ be smooth manifolds, given a smooth projection
\begin{equation}\label{proj}
    \pi: E\longrightarrow M
\end{equation}
and an open covering $\{U_{\alpha}\}$ of $M$, a local trivializationl means that for any open set $U_{\alpha}$, there exists diffeomorphism $\phi_{U_{\alpha}}$ such that the following diagram commutes
\begin{equation}
    \xymatrix{
\pi^{-1}(U_{\alpha})\ar[r]^{\phi_{U_{\alpha}}}\ar[d]^{\pi} & U_{\alpha}\times F \ar[ld]^{\eta}\\
U_{\alpha}
    }
\end{equation}
where $\eta$ is the projection to the first component. If the above conditions hold, we say that (\ref{proj}) is a fibre bundle with fibre $F$, $E$ is the total space, $M$ is the base space, and $\pi^{-1}(x)$ is the fibre at $x$.

Let $G$ be a Lie group with identity $e$, $G$ acts smoothly right on $M$ if there exists a smooth map
\begin{equation}
    R: M\times G\longrightarrow M
\end{equation}
satisfying

\noindent (1) $R(x,e)=x$;

\noindent (2) $R\left(R(x,h_1) ,h_2 \right)=R(x,h_1h_2)$.

\noindent For any $x\in M$, the stabilizer is defined as
\begin{equation}
    {\text{Stab(x)}}:=\{h\in G| R(x,h)=x \}.
\end{equation}
The right action $R$ is called free, if the stabilizer is trivial for any point.

A fibre bundle $\pi: P\longrightarrow M$ with fibre $G$ (a Lie group) is called a principal $G$-fibre bundle, if $G$ acts smoothly right on $P$ and for any $\alpha$, the local trivialization $\phi_{U_{\alpha}}$ is $G$-invariant. That is, for any $(p,h)\in \pi^{-1}(U)\times G$, the following holds
\begin{equation}
   \phi_{U_{\alpha}} \big(R(p,h)\big)=\big(\phi_{U_{\alpha}}(p)\big)h,
\end{equation}
where
\begin{equation}
    \big(\phi_{U_{\alpha}}(p)\big)h=\big(\pi(p),h_p\big)h:=\big(\pi(p),h_ph\big).
\end{equation}

Let $\{U_{\alpha},\phi_{\alpha}\}$ be a local trivialization of the principal $G$-fibre bundle $\pi: P\rightarrow M$. If $U_{\alpha\beta}\triangleq U_{\alpha}\cap U_{\beta}\neq \emptyset$, then there are two different local trivialization on $U_{\alpha\beta}$, i.e.
\begin{equation}
    \xymatrix{
U_{\alpha\beta}\times G  &\\
    \pi^{-1}(U_{\alpha\beta})\ar[r]^{\phi_{\beta}} \ar[u]^{\phi_{\beta}} &U_{\alpha\beta}\times G.
    }
\end{equation}
Hence
\begin{equation}  
\phi_{\alpha}\circ\phi_{\beta}^{-1}:U_{\alpha\beta}\times G\longrightarrow U_{\alpha\beta}\times G
\end{equation}
is a smooth, fibre preserving, $G$-invariant right action, i.e. there exists a map
\begin{equation}\label{transf}
    g_{\alpha\beta}:U_{\alpha\beta}\longrightarrow G
\end{equation}
satisfying
\begin{equation}
    \phi_{\alpha}\circ\phi_{\beta}^{-1}(x,g)=\big(x,g_{\alpha\beta}(x)g\big).
\end{equation}
We call the functions $g_{\alpha\beta}$ in (\ref{transf}) the transition functions. Obviously, for open sets $U_{\alpha}\cap U_{\beta}\cap U_{\gamma}\neq \emptyset$, the transition functions satisfy the following
\begin{equation}\label{cocycle}
    g_{\alpha\beta}\circ g_{\beta\gamma}=g_{\alpha\gamma}.
\end{equation}

\begin{example}
    Let $M$ be an oriented spacetime manifold. The frame bundle of the cotangent bundle $T^*M$ of $M$ is defined as
\begin{equation}
    \text{SO}_M^*(1,3):=\bigsqcup_{x\in M}\text{Fr}(T_x^*M),
\end{equation}
where $\text{Fr}(T_x^*M)$ consists of all orthonormal bases in the cotangent space $T_x^*M$ at $x$ that are compatible with the given orientation and the given time orientation. Define
\begin{equation}
\begin{split}
    \pi: \text{SO}_M^*(1,3)&\longrightarrow M \\
\text{Fr}(T_x^*M)&\longmapsto x.
    \end{split}
\end{equation}
It is no hard to see that the frame bundle $\pi: \text{SO}_M^*(1,3)\longrightarrow M$ is a smooth principal $SO_0(1,3)$-fibre bundle.
\end{example}

Given principal $G_1$-fibre bundle $\pi_1:P_1\rightarrow M_1$ and principal $G_2$-fibre bundle $\pi_2:P_2\rightarrow M_2$, a map $f:P_1\rightarrow P_2$ is called a bundle map if and only if there  exist group homomorphism $f_G:G_1\rightarrow G_2$ and a map $f_B:M_1\rightarrow M_2$ such that
\begin{equation}
    f(pg)=f(p)f_G(g)
\end{equation}
and the following diagram commutes
\begin{equation}
\xymatrix{
P_1\ar[r]^{f} \ar[d]_{\pi_1} & P_2 \ar[d]^{\pi_2}\\
M_1\ar[r]_{f_B} & M_2.}
\end{equation}

With the above preparations, we can now give the definition of spin structure on spacetime manifold.

\begin{defn}
    Let $M$ be a orientable spacetime manifold. A spin structure on $M$ is a binary $(\text{Spin}M,\eta)$ satisfying

\noindent (1) $\text{Spin}M$ is a principal $Spin(1,3)$-fibre bundle;

\noindent (2) $\eta:\text{Spin}M\rightarrow \text{SO}_M^*(1,3)$ is a 2-fold covering map and the following diagram commutes
\begin{equation}
    \xymatrix{
\text{Spin}M\ar[d]_{\eta} \ar[r]^{\pi} & M\\
    \text{SO}_M^*(1,3)\ar[ru]_{\pi}
    }
\end{equation}

\noindent (3) For arbitrary $(p,g)\in \text{Spin}M$, there holds
\begin{equation}
    \eta(pg)=\eta(p)\widetilde{Ad}(g).
\end{equation}    
\end{defn}

There exists a spin structure on an orientable spacetime manifold $M$ if and only if the transition functions on the frame bundle $\text{SO}_M^*(1,3)$ can be lifted to $Spin(1,3)$, i.e. there exists a map $\widetilde{g}_{\alpha\beta}$ such that the following diagram commutes
\begin{equation}
    \xymatrix{
    &Spin(1,3)\ar[d]^{\widetilde{Ad}}\\
    U_{\alpha\beta}\ar[r]_{g_{\alpha\beta}} \ar[ru]^{\widetilde{g}_{\alpha\beta}} & SO_0(1,3)
    }
\end{equation}
and $\widetilde{g}_{\alpha\beta}$ satisfies (\ref{cocycle}). This condition is equivalent to that  the second Stiefel-Whitney class $w_2(M)\in H^2(M;\mathbb{Z}_2)$ vanishes. Moreover, if $w_2(M)=0$, then there is a one-to-one correspondence between the number of spin structures on $M$ and the  elements in $H^1(M;\mathbb{Z}_2)$ \cite{bicht,lawson}. For more facts about spin structures on manifolds, we refer the reader to \cite{lawson}. In particular, we recall the following facts:

\begin{lem}
    (1) Any orientable manifold with dimension less than or equal to 3 is spin;

    (2) The Cartesian product of two spin manifolds is also a spin manifold.
\end{lem}

\subsubsection{Dirac operator}

First of all, we introduce the connection on principal fibre bundle and the induced covariant derivatives on the associated vector bundle, we refer the reader to \cite{hijazi} for details. Let $\pi:P\rightarrow M$ be a principal $G$-fibre bundle, $\mathfrak{g}$ is the Lie algebra of $G$. For any $u\in P$, let $G_u$ be the subspace of $T_uP$ which consisting of the tangent vectors that are tangent to the fibre at $u$. We say that $\Gamma$ is a connection on $P$, if for every $u\in P$, there is a subspace $Q_u$ (depends smoothly on $u$) of $T_uP$ satisfying

(1) $T_uP=G_u\oplus Q_u$;

(2) $Q_{ug}=(R_g)_{*}Q_u$ for every $g\in G$.

\noindent $Q_u$ is called the horizontal distribution and $G_u$ is the vertical distribution. A tangent vector $X\in T_uP$ is said to be horizontal or vertical if and only if it belongs to the subspace $Q_u$ or $G_u$. In particular, the corresponding components of the vector $X$ with respect to the direct sum $G_u\oplus Q_u$ are called the vertical and horizontal components of $X$, respectively.

For any vector $A\in \mathfrak{g}$, $A$ induces a vector field on $P$, i.e. for any $u\in P$, 
\begin{equation}
    {A^{*}}_u:=\frac{d}{dt}\Big|_{t=0}(u\cdot \exp tA).
\end{equation}
We say that $A^*$ is the fundamental vector field generated by $A$. According to the definition, it is easy to see that the map
\begin{equation}
    \begin{split}
\mathfrak{g}&\longrightarrow G_u \\
        A&\longmapsto {A^*}_u
    \end{split}
\end{equation}
is a linear isomorphism. Given a connection $\Gamma$ on $P$, we define the $\mathfrak{g}$-valued 1-form $\omega$ on $P$ as follows: for any $X\in T_uP$, $\omega(X)$ is defined as the unique $A\in \mathfrak{g}$ such that ${A^*}_u$  is equal to the vertical component of $X$. We call
\begin{equation}
    \omega: TP\longrightarrow \mathfrak{g}
\end{equation}
the connection 1-form on the principal fibre bundle corresponding to $\Gamma$. Thus, $X$ is horizontal if and only if $\omega(X)=0$. By definition, it is no hard to see that
\begin{equation}\label{cp1}
    \omega({A^*}_u)=A.
\end{equation}
Moreover, for any $g\in G$, the following holds \cite{loring}
\begin{equation}\label{cp2}
    (R_g)^*\omega=\text{ad}(g^{-1})\omega,
\end{equation}
where $\text{ad}$ is the adjoint representation of $G$. On the contrary, for any $C^{\infty}$ $\mathfrak{g}$-valued 1-form $\omega$ on $P$ satisfying (\ref{cp1}) and (\ref{cp2}), there always exists a unique connection $\Gamma$ on $P$ such that the connection 1-form corresponding to $\Gamma$ is exactly $\omega$. In fact, the corresponding horizontal distribution can be defined as follows
\begin{equation}
    Q_u:=\left\{X\in T_uP\Big|\,\omega (X)=0 \right\}.
\end{equation}

After we define the connection on the principal fibre bundle, we then discuss how the connection 1-form on the principal fibre bundle induces a covariant derivative on the associated vector bundle. We first consider the following representation on the group $G$
\begin{equation}\label{rep}
    \rho: G\longrightarrow \text{End}(\Sigma_n),
\end{equation}
where $\Sigma_n$ is a n-dimensional vecotr space. And for any vector $v$ in $\Sigma_n$, we denote by $gv\triangleq \rho(g)v$. The group $G$ naturally induces the following action on the space $P\times \Sigma$
\begin{equation}
    (p,v)\longmapsto (pg,g^{-1}v).
\end{equation}
The associated vector bundle $E:=P\times_{\rho}\Sigma_n$ of the principal $G$-fibre bundle $\pi: P\rightarrow M$ under the representation (\ref{rep}) is defined as the following quotient space
\begin{equation}
    E:=P\times_{\rho}\Sigma_n=\big(P\times \Sigma_n\big)\big/\sim,
\end{equation}
where the equivalence relation is
\begin{equation}
    (p,v)\sim (pg,g^{-1}v).
\end{equation}
It is no hard to see that
\begin{equation}
    \xymatrix{
E\ar[d]_{\pi^{\prime}}:=P\times_{\rho}\Sigma_n\\
    M
    }
\end{equation}
is a vector bundle, where $\pi^{\prime}([p,v]):=\pi(p)$. Each fibre of $E$ is isomorphic to the vector space $\Sigma_n$, and the transition functions are
\begin{equation}
\begin{split}
U_{\alpha\beta}&\longrightarrow GL(\Sigma_n)\\
x&\longmapsto \rho\circ\varphi_{\alpha\beta}(x),
    \end{split}
\end{equation}
where $\varphi_{\alpha\beta}$ is the transition function of the principal $G$-fibre bundle $\pi: P\rightarrow M$.

\begin{example}
    For the cotangent bundle $T^*M$ of orientable spacetime $M$, the following isomorphism between vector bundles holds
\begin{equation}
    T^*M\simeq \text{SO}_{M}^*(1,3)\times_{\rho_0}\R^4,
\end{equation}
    where $\rho_0: SO_0(1,3)\rightarrow \text{End}(\R^4)$ is the usual matrix representation.
\end{example}

If there exists a spin struction on a spacetime manifold, the spinor vector bundle is defined as the associated vector bundle of the principal fibre bundle $\text{Spin}M$ with respect to the complex spin representation (\ref{spinrepeq}), i.e.
\begin{equation}
    \Sigma M:=\text{Spin}M\times_{\rho}\C^4.
\end{equation}
A section $\Psi\in\Gamma(\Sigma M)$ is called the spinor. Moreover, according to the definition of the associated vector bundle, $\Psi$ can be expressed as
\begin{equation}
    \Psi\big|_{U}=[\widetilde{s},\psi]
\end{equation}
 on any open set $U$ of $M$, where $\widetilde{s}\in \Gamma_U(\text{Spin}M)$  and $\psi:U\rightarrow \C^4$ is a smooth vector-valued function.

On the spinor vector bundle, the Clifford multiplication can be defined as follows:
\begin{defn}
    \begin{equation}\label{climsb}
        \begin{split}
\mathfrak{m}:T^*M\otimes \Sigma M&\longrightarrow \Sigma M\\
            X^*\otimes \Psi=[s,e]\otimes [\widetilde{s},\psi]&\longmapsto [\widetilde{s},e\cdot \psi]\triangleq X^*\cdot \Psi,
        \end{split}
    \end{equation}
    where $e\cdot \psi$ is exactly the Clifford multiplication defined in (\ref{clim}).
\end{defn}
\begin{rmk}
    It is not difficult to verify that the (\ref{climsb}) does not depend on the choice of the equivalence classes.
\end{rmk}

Next we consider the covariant derivative on vector bundle. Let $\Gamma(E)$ be a section of $E$. The covariant derivative of the vector bundle $E$ is defined as the following map
\begin{equation}
    \nabla: \Gamma(E)\longrightarrow \Gamma(E)\otimes \Gamma(T^*M)
\end{equation}
which satisfies:

\noindent (1) For any $\psi\in \Gamma(E)$, $X,Y\in T_pM$
\begin{equation}\label{pcd1}
    \nabla_{X+Y}\psi=\nabla_X\psi+\nabla_Y\psi.
\end{equation}
For any smooth vector field $Z$ and smooth function $f$ on $M$
\begin{equation}\label{pcd2}
\begin{split}
    \nabla_{fZ}\psi&=f\nabla_Z\psi,\\
\nabla_X(f\psi)&=X(f)\psi+f\nabla_X\psi.
    \end{split}
\end{equation}

\noindent (2) For $\psi_1$, $\psi_2\in\Gamma(E)$
\begin{equation}\label{pcd3}
    \nabla_X(\psi_1+\psi_2)=\nabla_X\psi_1+\nabla_X\psi_2.
\end{equation}
According to property (1), $\nabla$ is also a map from  $\Gamma(TM)\otimes \Gamma(E)$ to $\Gamma(E)$, let
\begin{equation}
    \nabla_X\psi:=\nabla\psi(X).
\end{equation}
Let $\omega$ be a connection 1-form on the principal $G$-fibre bundle $\pi:P\rightarrow M$, and $E:=P\times_{\rho}\Sigma_n$ is the associated vector bundle.  For a local section $\Psi=[s,\sigma]$ on $E$ and a smooth vector field$X$, the covariant derivative is defined as
\begin{equation}\label{cdavb}
\nabla_X\Psi:=\big[s,X(\sigma)+\rho_{*}\big(\omega\circ s_{*})(X)\sigma \big)\big].
\end{equation}
It is no hard to see that $\nabla_X\Psi$ is well-defined (i.e. does not depend on the choice of equivalence classes) and satisfies the property (\ref{pcd1}), (\ref{pcd2}) and (\ref{pcd3}). Conversely, taking the cotangent bundle $T^*M$ of the spacetime manifold $M$ as an example, there exists a Levi-Civita connection $\nabla$ under the Lorentz metric $g$. For any open set $U\subset M$, let $\{e^{\alpha}\}_{\alpha=0}^3$ be an orthonormal basis on $T^*U$ which is compatible with the orientation and the time orientation, for any smooth vector field $X$, we have
\begin{equation}\label{wij}
    \nabla_Xe^{\alpha}={\omega^*}^{\alpha}_{\beta}(X)e^{\beta},
\end{equation}
then the connection 1-form $\omega$ on the frame bundle $\text{SO}_M^*(1,3)$ satisfies
\begin{equation}\label{*wij}
\begin{split}
    \omega(s_*X)&=\begin{pmatrix}
0 & {\omega^*}_0^1(X) &{\omega^*}_0^2(X) &{\omega^*}_0^3(X) \\
{\omega^*}_1^0(X)&0&{\omega^*}_1^2(X) &{\omega^*}_1^3(X) \\
{\omega^*}_2^0(X)&{\omega^*}_2^1(X) &0 & {\omega^*}_2^3(X)\\
 {\omega^*}_3^0(X) &{\omega^*}_3^1(X) & {\omega^*}_3^2(X)& 0     
    \end{pmatrix}\\
    &=\begin{pmatrix}
      0 & \big<\nabla_Xe^0,e^1\big> & \big<\nabla_Xe^0,e^2\big> & \big<\nabla_Xe^0,e^3\big>\\
      \big<\nabla_Xe^0,e^1\big> & 0 & \big<\nabla_Xe^2,e^1\big> & \big<\nabla_Xe^3,e^1\big>\\
      \big<\nabla_Xe^0,e^2\big> & \big<\nabla_Xe^1,e^2\big> & 0 & \big<\nabla_Xe^3,e^2\big>\\
      \big<\nabla_Xe^0,e^3\big> & \big<\nabla_Xe^1,e^3\big> & \big<\nabla_Xe^2,e^3\big> & 0\\
   \end{pmatrix}\in \mathfrak{so}(1,3),
   \end{split}
\end{equation}
where
\begin{equation}
    \begin{split}
s: U&\longrightarrow \text{SO}_M^*(1,3)\\
        x&\longmapsto (e^0,e^1,e^2,e^3)\big|_{x}
    \end{split}
\end{equation}
is a section of the frame bundle. In order to calculate the spinorial Levi-Civita connection on the spinor bundle, we rewrite the above formula (\ref{*wij}) as a linear combination of the 6 basis matrices of the Lie algebra $\mathfrak{so}(1,3)$, i.e.
\begin{equation}\label{*wijcb}
\begin{split}
 \omega(s_*X)&=\big<\nabla_Xe^0,e^1\big>\begin{pmatrix}
      0 & 1 & 0 & 0\\
      1 & 0 & 0 & 0\\
      0 & 0 & 0 & 0\\
      0 & 0 & 0 & 0\\  \end{pmatrix}+\big<\nabla_Xe^0,e^2\big>\begin{pmatrix}
      0 & 0 & 1 & 0\\
      0 & 0 & 0 & 0\\
      1 & 0 & 0 & 0\\
      0 & 0 & 0 & 0\\
   \end{pmatrix}\\
&+\big<\nabla_Xe^0,e^3\big>\begin{pmatrix}
      0 & 0 & 0 & 1\\
      0 & 0 & 0 & 0\\
      0 & 0 & 0 & 0\\
      1 & 0 & 0 & 0\\
\end{pmatrix}+\big<\nabla_Xe^2,e^1\big>\begin{pmatrix}
      0 & 0 & 0 & 0\\
      0 & 0 & 1 & 0\\
      0 & -1 & 0 & 0\\
      0 & 0 & 0 & 0\\
   \end{pmatrix}\\
   &+\big<\nabla_Xe^3,e^1\big>\begin{pmatrix}
      0 & 0 & 0 & 0\\
      0 & 0 & 0 & 1\\
      0 & 0 & 0 & 0\\
      0 & -1 & 0 & 0\\
\end{pmatrix}+\big<\nabla_Xe^3,e^2\big>\begin{pmatrix}
      0 & 0 & 0 & 0\\
      0 & 0 & 0 & 0\\
      0 & 0 & 0 & 1\\
      0 & 0 & -1 & 0\\
   \end{pmatrix}.
   \end{split}
\end{equation}

With the above preparations, we now define the connection on the spinor vector bundle. Let $U\subset M$ be a simple connected open set, then any local section $s$ of the frame bundle $\text{SO}_M^*(1,3)$ can be lifted to a local section $\widetilde{s}$ of the principal fibre bundle $\text{Spin}M$, i.e. the following diagram commutes
\begin{equation}
\xymatrix{
 &\text{Spin}M\ar[d]^{\eta} \\
U\ar[r]_{s} \ar[ru]^{\widetilde{s}} &\text{SO}_M^*(1,3)
}
\end{equation}
In order to define the connection on $\Sigma M$, we only need to define a connection 1-form $\widetilde{\omega}$ on the principal fibre bundle $\text{Spin}M$. To do this, we define $\widetilde{\omega}$ to be the only 1-form such the diagram
\begin{equation}
    \xymatrix{
 & T(\text{Spin}M)\ar[d]^{\eta_*}\ar[r]^{\widetilde{\omega}}&\mathfrak{spin}(1,3)\ar[d]^{\widetilde{Ad}_*}\\
    TU\subset TM\ar[ru]^{\widetilde{s}_*} \ar[r]^{s_*}&T(\text{SO}_M^*(1,3))\ar[r]^{\omega}&\mathfrak{so}(1,3)
    }
\end{equation}
commutes, where $\omega\circ s_*$ is exactly (\ref{*wijcb}). Therefore, for any $X\in TU$, let $\Psi_{\alpha}=[\widetilde{s},\sigma_{\alpha}]\in \Gamma_U(\Sigma M)$ be a local section ($\{\sigma_{\alpha}\}_{\alpha=0}^3$ is an orthonormal basis of $\C^4$), by the definition of the covariant derivative on the associated vector bundle (\ref{cdavb}) and Proposition \ref{ad_*}, we can deduce that
\begin{equation}\label{spinorcd}
\begin{split}
    \nabla_X\Psi_{\alpha}&=\big[\widetilde{s},\rho_*(\widetilde{\omega}\circ \widetilde{s}_*(X))\sigma_{\alpha}\big]\\
&=\big[\widetilde{s},\rho_{*}\big( \widetilde{Ad}_*^{-1}(\omega \circ s_*(X))\big)\sigma_{\alpha}\big]\\
&=-\frac{1}{2}\big<\nabla_Xe^0,e^1\big>e^0\cdot e^1\cdot \Psi_{\alpha}-\frac{1}{2}\big<\nabla_Xe^0,e^2\big>e^0\cdot e^2\cdot \Psi_{\alpha}\\
&\quad-\frac{1}{2}\big<\nabla_Xe^0,e^3\big>e^0\cdot e^3\cdot \Psi_{\alpha}+\frac{1}{2}\big<\nabla_Xe^1,e^2\big>e^1\cdot e^2\cdot \Psi_{\alpha}\\
&\quad+\frac{1}{2}\big<\nabla_Xe^1,e^3\big>e^1\cdot e^3\cdot \Psi_{\alpha}+\frac{1}{2}\big<\nabla_Xe^2,e^3\big>e^2\cdot e^3\cdot \Psi_{\alpha}\\
&=-\frac{1}{2}\Big[\omega_{01}(X)e^0\cdot e^1+\omega_{02}(X)e^0\cdot e^2+\omega_{03}(X)e^0\cdot e^3\\
&\quad+\omega_{12}(X)e^1\cdot e^2+\omega_{13}(X)e^1\cdot e^3+\omega_{23}(X)e^2\cdot e^3\Big]\cdot \Psi_{\alpha}.
    \end{split}
\end{equation}

Now we give the definition of the Dirac operator on the spinor vector bundle $\Sigma M$:
\begin{defn}
Let $\nabla$ be the connection on $\Sigma M$ given by  (\ref{spinorcd}), the Dirac operator $\mathscr{D}$ is defined as
\begin{equation}
\begin{split}
     \mathscr{D}:=\mathfrak{m}\circ \nabla: \Gamma(\Sigma M)&\longrightarrow \Gamma(\Sigma M)\\
     \Psi &\longmapsto e^{\alpha}\cdot \nabla_{{\alpha}}\Psi.
     \end{split}
\end{equation}
\end{defn}

\begin{rmk}
    Since the signature of the metric in Lorentzian manifold is different from the Riemannian case, the Dirac operator $\mathscr{D}$ is not elliptic. 
\end{rmk}

In the 4-dimensional Minkowski spacetime, the Dirac operator is
\begin{equation}
    \mathscr{D}=e^0\cdot \partial_t+\sum_{i=1}^3e^i\cdot\partial_{x_i}.
\end{equation}
Its square is as follows
\begin{equation}
\begin{split}
\mathscr{D}^2&=\Big(e^0\cdot \partial_t+\sum_{i=1}^3e^i\cdot\partial_{x_i} \Big)^2\\
    &=e^0\cdot e^0\cdot\partial^2_t+e^0\cdot\partial_t\Big(\sum_{i=1}^3e^i\cdot\partial_{x_i}\Big)+\Big(\sum_{i=1}^3e^i\cdot\partial_{x_i} \Big)(e^0\cdot\partial_t)+\Big(\sum_{i=1}^3e^i\cdot\partial_{x_i}\Big)^2\\
    &=\partial^2_t+\sum_{i=1}^3(e^0\cdot e^i+e^i\cdot e^0)\partial_{x_i}\partial_t+\Big(\sum_{i=1}^3e^i\cdot\partial_{x_i}\Big)^2\\
&=\partial^2_t+\Big(\sum_{i=1}^3e^i\cdot\partial_{x_i}\Big)^2\\
&=\partial^2_t-\sum_{i=1}^3\partial^2_{x_i}\\
&=\square.
\end{split}
\end{equation}
Thus, the Dirac operator $\mathscr{D}$ is now the square root of the wave operator $\square$.

\subsection{Existence and uniqueness}
In this paper, the separation method proposed by Chandrasekhar \cite{chand3,chand2} is used to prove the nonexistence of the nontrivial time-periodic solution of the Dirac equation, that is, the original Dirac equation is separated into two first order ordinary differential equations. In this subsection, we introduce the existence and uniqueness theorem of the solution of a system of first-order ordinary differential equations, the proof of which can be found in, for example, \cite{teschl}, \cite{walter}.

Let $D$ be an open domain in the Euclidean space $\R^{n+1}$, the points on it are denoted as $(t,\textbf{y})$, where $\textbf{y}\in \R^n$.

\begin{defn}\label{dlocallip}
    A vector-valued function $\textbf{f}=\textbf{f}(t,\textbf{y})$ satisfies the local Lipschitz condition on $D$ with respect to the variable $\textbf{y}$, if for any $(t_0,\textbf{y}_0)\in D$, there exists an open neighborhood $U$ of $(t_0,\textbf{y}_0)$ such that for any $(t,\textbf{y}_1)$, $(t,\textbf{y}_2)\in U\cap D$, there exists a positive constant $L_U$ such that the following inequality holds
    \begin{equation}\label{locallip}\tag{2-104}
        |\textbf{f}(t,\textbf{y}_1)-\textbf{f}(t,\textbf{y}_2)|\le L_U|\textbf{y}_1-\textbf{y}_2|.
    \end{equation}
\end{defn}

\begin{rmk}
    In general, the Lipschitz constant $L_U$ in (\ref{locallip}) can vary depending on the choice of $U$.
\end{rmk}

The existence and uniqueness theorem of the solution is as follows:

\begin{thm}\label{odeeu}
    Let $\textbf{f}(t,\textbf{y})$ be a continuous function on the domain $D$ which satisfies the local Lipschitz condition with respect to the variable $\textbf{y}$. Then for any fixed $(t_0,\textbf{y}_0)\in D$, the following initial value problem of first order ordinary differential equation has a unique solution
\begin{equation}
    \frac{d\textbf{y}}{dt}=\textbf{f}(t,\textbf{y}),\quad \textbf{y}(t_0)=\textbf{y}_0.
\end{equation}
Moreover, the solution can be extended to the left and right up to the boundary of the domain $D$. 
\end{thm}

\begin{rmk}
If the vecter-valued function $\textbf{f}(t,\textbf{y})$ has continuous partial derivative with respect to the variable $\textbf{y}$ in domain $D$, i.e.
    \begin{equation}
        \frac{\partial\textbf{f}}{\partial\textbf{y}}\in C(D),
    \end{equation}
    then $\textbf{f}$ satisfies the local Lipschitz condition in Definition \ref{dlocallip}.
\end{rmk}

\mysection{Non-extreme Kerr-Newman-dS spacetime}\ls

In this section, we manily consider the nonexistence of time-periodic solutions of the following Dirac equation
\begin{equation}\label{diraceq3}
\big(\mathscr{D}+ie^{\alpha}A(e_{\alpha})+i\lambda\big)\Psi=0
\end{equation}
in the non-extreme Kerr-Newman-dS spacetime, where $\lambda\in\mathbb{R}$, $A$ is the electromagnetic vector potential.

The Kerr-Newman-dS spacetime is an exact solution of the Einstein-Maxwell equation, which describes a charged rotating black hole with positive cosmological constant. Kerr-Newman-dS spacetime is the following manifold
\begin{equation}
    M_{KNdS}=\R_{t}\times \R^+_{r}\times S^2,
\end{equation}
equipped with a Lorentzian metric (in the Boyer-Lindquist coordinate)
\begin{equation}\label{kndsmetric1}
\begin{split}
    g_{KNdS}=&-\Big(1-\frac{2mr-Q^2-P^2}{U}-\kappa^2(r^2+a^2\sin^2\theta) \Big)dt^2\\
    &+\frac{V_+}{U{E_+}^2}\sin^2\theta d\varphi^2
+\frac{U}{\Delta_+(r)}dr^2+\frac{U}{\Delta_+(\theta)}d\theta^2\\
&-\frac{a\sin^2\theta}{E_+}\Big(\frac{2mr-Q^2-P^2}{U}+\kappa^2(r^2+a^2)\Big)\big(dt\,d\varphi+d\varphi\,dt \big),
    \end{split}
\end{equation}
where the constants $\kappa>0$, $m>0$, and
\begin{equation}
    \begin{split}
U&=r^2+a^2\cos^2\theta,\\
E_+&=1+\kappa^2a^2,\\
\Delta_+(r)&=(r^2+a^2)(1-\kappa^2r^2)-2mr+Q^2+P^2,\\
\Delta_+(\theta)&=1+\kappa^2a^2\cos^2\theta,\\
V_+&=(2mr-Q^2-P^2)a^2\sin^2\theta+U(r^2+a^2)(1+\kappa^2a^2).
    \end{split}
\end{equation}
Moreover, the electromagnetic field is $F=dA$, and $A$ is the following electromagnetic 1-form
\begin{equation}
    A=-\frac{Qr}{U}\Big(dt-\frac{a\sin^2\theta}{E_+}d\varphi\Big)-\frac{P\cos\theta}{U}\Big(a\,dt-\frac{r^2+a^2}{E_+}d\varphi\Big).
\end{equation}
Let
\begin{equation}
\begin{split}
    m^{\pm}\triangleq \frac{1}{\sqrt{54}}&\Bigg((1-\kappa^2a^2)\pm\sqrt{\big(1-\kappa^2a^2\big)^2-12\kappa^2(a^2+Q^2+P^2)} \Bigg)^{\frac{1}{2}}\\
&\times\Bigg(2(1-\kappa^2a^2)^2\mp \sqrt{\big(1-\kappa^2a^2\big)^2-12\kappa^2(a^2+Q^2+P^2)} \Bigg).
    \end{split}
\end{equation}
According to the discussions in \cite{bcds}, if the parameters $\kappa$, $a$, and $m$ satisfy the following constraints, i.e.
\begin{equation}
    \kappa^2a^2\le 7-4\sqrt{3},\; m^{-}<m<m^{+},
\end{equation}
then the polynomial $\Delta_+(r)$ of order 4 with respect to $r$ has exactly 4 different real roots, 3 positive $0<r_c<r_{-}<r_{+}$ and 1 negative $r_n=-(r_c+r_{-}+r_{+})$. At this time, we call $\big(M_{KNdS},g_{KNdS}\big)$ the non-extreme Kerr-Newman-dS spacetime.
The hypersurfaces corresponding to the 3 roots $\{r=r_c\}$, $\{r=r_{-}\}$ and $\{r=r_{+}\}$ are called the Cauchy horizon, the event horizon and the cosmological horizon, respectively. In particular, if $m=m^{-}$  then $r_c=r_{-}$, i.e. the Cauchy horizon coincides with the event horizon, at this point we call $\big(M_{KNdS},g_{KNdS}\big)$ the extreme Kerr-Newman-dS spacetime. In this section, we only consider the non-extreme case.
\begin{rmk}
    $r=r_c$ and $r=r_{\pm}$ are just the coordinate singularities of the metric $g_{KNdS}$, and the corresponding hypersurfaces are regular lightlike hypersurfaces \cite{akcay}.
\end{rmk}

For convenience, we rerepresent the Kerr-Newman-dS metric in the following form
\begin{equation}
\begin{split}
    g_{KNdS}=&-\frac{\Delta_+(r)}{U}\Big(dt-\frac{a\sin^2\theta}{E_+}d\varphi \Big)^2+\frac{U}{\Delta_+(r)}dr^2+\frac{U}{\Delta_+(\theta)}d\theta^2\\
    &\qquad\qquad\qquad+\frac{\Delta_{+}(\theta)\sin^2\theta}{U}\Big(a\,dt-\frac{r^2+a^2}{E_+}d\varphi\Big)^2.
    \end{split}
\end{equation}

We require that the solution $\Psi$ of the Dirac equation (\ref{diraceq3}) is of the form
\begin{equation}\label{3-8}
    \Psi={S_+}^{-1}\Phi,
\end{equation}
where
\begin{equation}\label{Phi}
    \Phi=e^{-i\left(\omega t+(k+\frac{1}{2})\varphi\right)}\begin{pmatrix}
X_{-}(r)Y_-(\theta)\\
X_{+}(r)Y_+(\theta)\\
X_{+}(r)Y_-(\theta)\\
X_{-}(r)Y_+(\theta)\\
    \end{pmatrix},
\end{equation}
$k\in\mathbb{Z}$ and $S_+$ is the following diagonal matrix
\begin{equation}
    S_+=\left|\Delta_+(r)\right|^{\frac{1}{4}}\begin{pmatrix}
 (r+ia\cos\theta)^{\frac{1}{2}}& 0&0 & 0\\
0& (r+ia\cos\theta)^{\frac{1}{2}}&0 & 0\\
0& 0&(r-ia\cos\theta)^{\frac{1}{2}} & 0\\
        0& 0&0 & (r-ia\cos\theta)^{\frac{1}{2}}\\
    \end{pmatrix}.
\end{equation}
We can see that $S_+$ vanishes on the event horizon $\{r=r_-\}$. By the definition in \cite{reissner,kn}, a wave function $\Psi$ is called time-periodic with period $T$, if there exists a real number $\Omega$ such that
\begin{equation}
    \Psi(t+T,r,\theta,\varphi)=e^{-i\Omega T}\Psi(t,r,\theta,\varphi).
\end{equation}
Hence the $\Psi$ in (\ref{3-8}) satisfies the above definition.

\subsection{Spinorial联络}
In this subsection, we calculate the spinorial connection on the spinor bundle $\Sigma M$ in Kerr-Newman-dS spacetime when  $\Delta_+(r)>0$.

Denote the frame of the Kerr-Newman-dS metric
\begin{equation}
    \begin{split}
e_0&=\frac{r^2+a^2}{\sqrt{U\Delta_+(r)}}\left(\partial_t+\frac{aE_+}{r^2+a^2}\partial_{\varphi}\right),\\
e_1&=\sqrt{\frac{\Delta_+(r)}{U}}\partial_r, \\
e_2&=\sqrt{\frac{\Delta_+(\theta)}{U}}\partial_{\theta}, \\
   e_3&=\frac{-1}{\sqrt{U\Delta_+(\theta)}}\left(a\sin\theta\partial_t+\frac{E_+}{\sin\theta}\partial_{\varphi}\right),     
    \end{split}
\end{equation}
and the corresponding 1-form
\begin{equation}
    \begin{split}
e^0&=\sqrt{\frac{\Delta_+(r)}{U}}\left(dt-\frac{a\sin^2\theta}{E_+}d\varphi\right) ,\\
e^1&=\sqrt{\frac{U}{\Delta_+(r)}}dr ,\\
e^2&=\sqrt{\frac{U}{\Delta_+(\theta)}}d\theta ,\\
        e^3&=\sqrt{\frac{\Delta_+(\theta)}{U}}\sin\theta\left(a\,dt-\frac{r^2+a^2}{E_+}d\varphi\right) 
    \end{split}
\end{equation}
which satisfy
\begin{equation}
e^{\alpha}\left(e_{\beta}\right)=\delta^{\alpha}_{\beta}.
\end{equation}
Therefore, the metric $g_{KNdS}$ can be expressed as
\begin{equation}
    g_{KNdS}=-e^0\otimes e^0+e^1\otimes e^1+e^2\otimes e^2+e^3\otimes e^3.
\end{equation}
By Cartan's structure equations \cite{jose}
\begin{equation}
    \begin{split}
de^0&=-\omega_1^0\wedge e^1-\omega_2^0\wedge e^2-\omega_3^0\wedge e^3 ,\\
de^1&=-\omega_0^1\wedge e^0-\omega_2^1\wedge e^2-\omega_3^1\wedge e^3 ,\\
de^2&=-\omega_0^2\wedge e^0-\omega_1^2\wedge e^1-\omega_3^2\wedge e^3 ,\\
de^3&=-\omega_0^3\wedge e^0-\omega_1^3\wedge e^1-\omega_2^3\wedge e^2,\\
    \end{split}
\end{equation}
the connection 1-forms are as follows:
\begin{equation}
    \begin{split}
\omega_1^0&=C^0_{10}e^0-\frac{1}{2}C^3_{10}e^3,\quad \omega_2^0= C_{20}^0e^0+\frac{1}{2}C^0_{23}e^3,\\
\omega_3^0&=-\frac{1}{2}C^3_{10}e^1-\frac{1}{2}C^0_{23}e^2,\quad \omega_2^1=-C^1_{12}e^1-C^2_{12}e^2,\\
        \omega_3^1&=-\frac{1}{2}C^3_{10}e^0-C^3_{13}e^3,\quad \omega_3^2=\frac{1}{2}C^0_{23}e^0-C^3_{23}e^3,
    \end{split}
\end{equation}
and
\begin{equation}
    \begin{split}
\omega_1^0=-\omega_{01},&\quad \omega_2^0=-\omega_{02},\\
\omega_3^0=-\omega_{03},&\quad \omega_2^1=\omega_{12},\\
        \omega_3^1=\omega_{13},&\quad \omega_3^2=\omega_{23},
    \end{split}
\end{equation}
where
\begin{equation}
    \begin{split}
C_{10}^0&=\partial_{r}\sqrt{\frac{\Delta_+(r)}{U}} ,\quad C^0_{20}=-\sqrt{\Delta_+(\theta)}\partial_{\theta}\frac{1}{\sqrt{U}},\\
C_{12}^1&=-\frac{\sqrt{\Delta_+(\theta)}}{U}\partial_{\theta}\sqrt{U} ,\quad C^2_{12}=\frac{\sqrt{\Delta_+(r)}}{U}\partial_r\sqrt{U},\\
C_{10}^3&=-2ar\sqrt{\Delta_+(\theta)}U^{-\frac{3}{2}}\sin\theta ,\quad C^0_{23}=2a\sqrt{\Delta_+(r)}U^{-\frac{3}{2}}\cos\theta,\\
        C_{13}^3&=-\sqrt{\Delta_+(r)}\partial_r\frac{1}{\sqrt{U}} ,\quad C^3_{23}=\frac{1}{\sin\theta}\partial_{\theta}\left(\sqrt{\frac{\Delta_+(\theta)}{U}}\sin\theta \right).
    \end{split}
\end{equation}
According to (\ref{spinorcd}), when $\Delta_+(r)>0$, the spinorial connections take the following form:
\begin{equation}
    \begin{split}
\nabla_{e_0}\Psi&=e_0\left(\Psi\right)-\frac{1}{2}\omega_{01}(e_0)e^0\cdot e^1\cdot\Psi-\frac{1}{2}\omega_{02}(e_0)e^0\cdot e^2\cdot\Psi\\
        &\quad\qquad\;\;\,-\frac{1}{2}\omega_{13}(e_0)e^1\cdot e^3\cdot\Psi-\frac{1}{2}\omega_{23}(e_0)e^2\cdot e^3\cdot \Psi,\\
    \nabla_{e_1}\Psi&=e_1\left(\Psi\right)-\frac{1}{2}\omega_{03}(e_1)e^0\cdot e^3\cdot\Psi-\frac{1}{2}\omega_{12}(e_1)e^1\cdot e^2\cdot\Psi,\\
    \nabla_{e_2}\Psi&=e_2\left(\Psi\right)-\frac{1}{2}\omega_{03}(e_2)e^0\cdot e^3\cdot\Psi-\frac{1}{2}\omega_{12}(e_2)e^1\cdot e^2\cdot\Psi,\\
    \nabla_{e_3}\Psi&=e_3\left(\Psi\right)-\frac{1}{2}\omega_{01}(e_3)e^0\cdot e^1\cdot\Psi-\frac{1}{2}\omega_{02}(e_3)e^0\cdot e^2\cdot\Psi\\
        &\quad\qquad\;\;\,-\frac{1}{2}\omega_{13}(e_3)e^1\cdot e^3\cdot\Psi-\frac{1}{2}\omega_{23}(e_3)e^2\cdot e^3\cdot \Psi.
    \end{split}
\end{equation}
We fix the following Clifford representation:
\begin{equation}\label{cliffordbs}
    \begin{split}
   e^0\longmapsto \begin{pmatrix}
      0 & 0 & -1 & 0\\
      0 & 0 & 0 & -1\\
      -1 & 0 & 0 & 0\\
      0 & -1 & 0 & 0\\
   \end{pmatrix}, & \qquad e^1\longmapsto \begin{pmatrix}
      0 & 0 & -1 & 0\\
      0 & 0 & 0 & 1\\
      1 & 0 & 0 & 0\\
      0 & -1 & 0 & 0\\
   \end{pmatrix},\\
   e^2\longmapsto \begin{pmatrix}
      0 & 0 & 0 & 1\\
      0 & 0 & 1 & 0\\
      0 & -1 & 0 & 0\\
      -1 & 0 & 0 & 0\\
   \end{pmatrix}, &\qquad e^3\longmapsto \begin{pmatrix}
      0 & 0 & 0 & i\\
      0 & 0 & -i & 0\\
      0 & -i & 0 & 0\\
      i & 0 & 0 & 0\\
   \end{pmatrix}.
   \end{split}
   \end{equation}
then for $\alpha=0,1,2,3$, we have
\begin{equation}
\nabla_{e_{\alpha}}\Psi=e_{\alpha}\left(\Psi\right)+E_{\alpha}\cdot\Psi,
\end{equation}
where
\begin{equation}
    E_0=-\frac{1}{2}\begin{pmatrix}
 C^0_{10}+\frac{i}{2}C^0_{23}& -C^0_{20}-\frac{i}{2}C^3_{10}&0 & 0\\
-C^0_{20}-\frac{i}{2}C^3_{10}& -C^0_{10}-\frac{i}{2}C^0_{23}&0 & 0\\
0& 0&-C^0_{10}+\frac{i}{2}C^0_{23} & C^0_{20}-\frac{i}{2}C^3_{10}\\
        0& 0&C^0_{20}-\frac{i}{2}C^3_{10} & C^0_{10}-\frac{i}{2}C^0_{23}\\
    \end{pmatrix},
\end{equation}
\begin{equation}
E_1=-\frac{1}{2}\begin{pmatrix}
 0& -C^1_{12}+\frac{i}{2}C^3_{10}&0 & 0\\
C^1_{12}-\frac{i}{2}C^3_{10}& 0&0 & 0\\
0& 0&0 & -C^1_{12}-\frac{i}{2}C^3_{10}\\
        0& 0&C^1_{12}+\frac{i}{2}C^3_{10} & 0\\
    \end{pmatrix},
\end{equation}
\begin{equation}
E_2=-\frac{1}{2}\begin{pmatrix}
 0& -C^2_{12}+\frac{i}{2}C^0_{23}&0 & 0\\
C^2_{12}-\frac{i}{2}C^0_{23}& 0&0 & 0\\
0& 0&0 & -C^2_{12}-\frac{i}{2}C^0_{23}\\
        0& 0&C^2_{12}+\frac{i}{2}C^0_{23} & 0\\
    \end{pmatrix},
\end{equation}
\begin{equation}
E_3=-\frac{1}{2}\begin{pmatrix}
 -\frac{1}{2}C^3_{10}-iC^3_{23}& -\frac{1}{2}C^0_{23}-iC^3_{13}&0 & 0\\
-\frac{1}{2}C^0_{23}-iC^3_{13}&\frac{1}{2}C^3_{10}+iC^3_{23} &0 & 0\\
0& 0& \frac{1}{2}C^3_{10}-iC^3_{23}&\frac{1}{2}C^0_{23}-iC^3_{13} \\
        0& 0& \frac{1}{2}C^0_{23}-iC^3_{13}  & -\frac{1}{2}C^3_{10}+iC^3_{23} \\
    \end{pmatrix}.
\end{equation}

\subsection{Nonexistence}

With the above preparations, we now separate the variables for the Dirac equation (\ref{diraceq3}) when $\Delta_+(r)>0$. Let
\begin{equation}
    S_0\triangleq\begin{pmatrix}
    (r+ia\cos\theta)^{-\frac{1}{2}}& 0&0 & 0\\
0& (r+ia\cos\theta)^{-\frac{1}{2}}&0 & 0\\
0& 0&(r-ia\cos\theta)^{-\frac{1}{2}} & 0\\
        0& 0&0 & (r-ia\cos\theta)^{-\frac{1}{2}}\\
    \end{pmatrix},
\end{equation}
\begin{equation}
    \phi\triangleq\begin{pmatrix}
        X_-(r)Y_-(\theta)\\
        X_+(r)Y_+(\theta)\\
        X_+(r)Y_-(\theta)\\
        X_-(r)Y_+(\theta)
    \end{pmatrix},
\end{equation}
\begin{equation}
    \widetilde{E}_0\triangleq -\frac{1}{2}\begin{pmatrix}
0&0&-C_{10}^0+\frac{i}{2}C_{23}^0& C_{20}^0-\frac{i}{2}C_{10}^3\\
0&0& C_{20}-\frac{i}{2}C_{10}^3& C_{10}^0-\frac{i}{2}C_{23}^0\\
C_{10}^0+\frac{i}{2}C_{23}^0& -C_{20}^0-\frac{i}{2}C_{10}^3 &0&0\\
        -C_{20}^0-\frac{i}{2}C_{10}^3 & -C_{10}^0-\frac{i}{2}C_{23}^0 &0&0
    \end{pmatrix},
\end{equation}
\begin{equation}
    \widetilde{E}_1\triangleq -\frac{1}{2}\begin{pmatrix}
0&0&0&C_{12}^1+\frac{i}{2}C_{10}^3\\
0&0&C_{12}^1+\frac{i}{2}C_{10}^3&0\\
0&-C_{12}^1+\frac{i}{2}C_{10}^3&0&0\\
        -C_{12}^1+\frac{i}{2}C_{10}^3&0&0&0\\
    \end{pmatrix},
\end{equation}
\begin{equation}
    \widetilde{E}_2\triangleq -\frac{1}{2}\begin{pmatrix}
0&0&C_{12}^2+\frac{i}{2}C_{23}^0 &0\\
0&0&0&-C_{12}^2-\frac{i}{2}C_{23}^0\\
-C_{12}^2+\frac{i}{2}C_{23}^0 &0&0&0\\
        0&C_{12}^2-\frac{i}{2}C_{23}^0 &0&0
    \end{pmatrix},
\end{equation}
\begin{equation}
    \widetilde{E}_3\triangleq -\frac{1}{2}\begin{pmatrix}
0&0&\frac{i}{2}C_{23}^0+C_{13}^3&-\frac{i}{2}C_{10}^3-C_{23}^3\\
0&0&-\frac{i}{2}C_{10}^3-C_{23}^3& -\frac{i}{2}C_{23}^0-C_{13}^3\\
\frac{i}{2}C_{23}^0-C_{13}^3 & -\frac{i}{2}C_{10}^3+C_{23}^3&0&0\\
        -\frac{i}{2}C_{10}^3+C_{23}^3& -\frac{i}{2}C_{23}^0+C_{13}^3&0&0\\
    \end{pmatrix}.
\end{equation}
Then the Dirac equation (\ref{diraceq3}) is equivalent to the following equation
\begin{equation}\label{pop}
    \mathscr{L}\phi=-i\lambda S_0\phi,
\end{equation}
where the operator $\mathscr{L}$ is defined as
\begin{equation}
    \begin{split}
\mathscr{L}&:=-e^0\cdot\frac{r^2+a^2}{\sqrt{u\Delta_+(r)}}\left(i\omega+\frac{aE_+}{r^2+a^2}(k+\frac{1}{2})i\right)S_0 \\
        &\quad+e^1\cdot\sqrt{\frac{\Delta_+(r)}{U}}\left(-\frac{1}{4}\Delta_+(r)^{-1}\partial_r\big(\Delta_+(r)\big)S_0+S_0\partial_r+\partial_r{S_0}\right)\\
        &\quad+e^2\cdot\sqrt{\frac{\Delta_+(\theta)}{U}}\left(\partial_{\theta}S_0+S_0\partial_{\theta}\right)\\
        &\quad+e^3\cdot \frac{1}{\sqrt{U\Delta_+(\theta)}}\left(ia\omega \sin\theta S_0+\frac{E_+}{\sin\theta}(k+\frac{1}{2})iS_0\right)\\
        &\quad+\left(-\widetilde{E}_0+\sum_{j=1}^3\widetilde{E}_j\right)S_0-e^0\cdot \frac{Qri}{\sqrt{U\Delta_+(r)}}S_0-e^3\cdot \frac{i}{\sqrt{U\Delta_+(\theta)}}P\cot\theta S_0.
    \end{split}
\end{equation}
Substituting (\ref{cliffordbs}) and $\widetilde{E}_{\alpha}$ into (\ref{pop}), we can deduce that $\phi$ satisfies the following equation
\begin{equation}\label{manzu}
\begin{split}
    &\begin{pmatrix}
i\lambda(r+ia\cos\theta)^{-\frac{1}{2}} & 0& D_{13}& L_{14} \\
0&i\lambda(r+ia\cos\theta)^{-\frac{1}{2}}& L_{23} & D_{24}\\
D_{31} & L_{32}& i\lambda(r-ia\cos\theta)^{-\frac{1}{2}}& 0\\
    L_{41}& D_{42} &0 &i\lambda(r-ia\cos\theta)^{-\frac{1}{2}}\\
    \end{pmatrix}\phi\\
&=0,
\end{split}
\end{equation}
where
\begin{equation}
    \begin{split}
D_{13}&=\sqrt{\frac{\Delta_+(r)}{U}}(r-ia\cos\theta)^{-\frac{1}{2}}\left[-\partial_r+\frac{i}{\Delta_+(r)}\left(\omega(r^2+a^2)+(k+\frac{1}{2})E_+a+Qr\right)\right],\\
D_{24}&=\sqrt{\frac{\Delta_+(r)}{U}}(r-ia\cos\theta)^{-\frac{1}{2}}\left[\partial_r+\frac{i}{\Delta_+(r)}\left(\omega(r^2+a^2)+(k+\frac{1}{2})E_+a+Qr\right)\right], \\
D_{31}&=\sqrt{\frac{\Delta_+(r)}{U}}(r+ia\cos\theta)^{-\frac{1}{2}}\left[\partial_r+\frac{i}{\Delta_+(r)}\left(\omega(r^2+a^2)+(k+\frac{1}{2})E_+a+Qr\right)\right], \\
        D_{42}&=\sqrt{\frac{\Delta_+(r)}{U}}(r+ia\cos\theta)^{-\frac{1}{2}}\left[-\partial_r+\frac{i}{\Delta_+(r)}\left(\omega(r^2+a^2)+(k+\frac{1}{2})E_+a+Qr\right)\right],\\
L_{14}&=\sqrt{\frac{\Delta_+(\theta)}{U}}(r-ia\cos\theta)^{-\frac{1}{2}}\Bigg[\partial_{\theta}-\frac{1}{\Delta_+(\theta)}\left(a\omega\sin\theta+\frac{E_+}{\sin\theta}(k+\frac{1}{2})-P\cot\theta\right)\\
&\qquad\qquad\qquad\qquad\qquad\qquad\qquad +\frac{1}{2}\left(\cot\theta-\frac{\kappa^2a^2\sin\theta\cos\theta}{\Delta_+(\theta)}\right)\Bigg], \\
L_{23}&=\sqrt{\frac{\Delta_+(\theta)}{U}}(r-ia\cos\theta)^{-\frac{1}{2}}\Bigg[\partial_{\theta}+\frac{1}{\Delta_+(\theta)}\left(a\omega\sin\theta+\frac{E_+}{\sin\theta}(k+\frac{1}{2})-P\cot\theta\right)\\
&\qquad\qquad\qquad\qquad\qquad\qquad\qquad +\frac{1}{2}\left(\cot\theta-\frac{\kappa^2a^2\sin\theta\cos\theta}{\Delta_+(\theta)}\right)\Bigg], \\
L_{32}&=\sqrt{\frac{\Delta_+(\theta)}{U}}(r+ia\cos\theta)^{-\frac{1}{2}}\Bigg[-\partial_{\theta}+\frac{1}{\Delta_+(\theta)}\left(a\omega\sin\theta+\frac{E_+}{\sin\theta}(k+\frac{1}{2})-P\cot\theta\right)\\
&\qquad\qquad\qquad\qquad\qquad\qquad\qquad -\frac{1}{2}\left(\cot\theta-\frac{\kappa^2a^2\sin\theta\cos\theta}{\Delta_+(\theta)}\right)\Bigg], \\
        L_{41}&=\sqrt{\frac{\Delta_+(\theta)}{U}}(r+ia\cos\theta)^{-\frac{1}{2}}\Bigg[-\partial_{\theta}-\frac{1}{\Delta_+(\theta)}\left(a\omega\sin\theta+\frac{E_+}{\sin\theta}(k+\frac{1}{2})-P\cot\theta\right)\\
&\qquad\qquad\qquad\qquad\qquad\qquad\qquad -\frac{1}{2}\left(\cot\theta-\frac{\kappa^2a^2\sin\theta\cos\theta}{\Delta_+(\theta)}\right)\Bigg]. 
    \end{split}
\end{equation}
Let
\begin{equation}
    D_{r\pm}=\partial_r\mp \frac{i}{\Delta_+(r)}\left(\omega(r^2+a^2)+Qr+\big(k+\frac{1}{2}\big)E_+a \right),
\end{equation}
\begin{equation}
\begin{split}
    L_{\theta\pm}=\partial_{\theta}\mp \frac{1}{\Delta_+(\theta)}\Bigg(\omega a \sin\theta&+\frac{(k+\frac{1}{2})E_+}{\sin\theta} -P\cot\theta\Bigg)\\
    &+\frac{1}{2}\left(\cot\theta-\frac{\kappa^2a^2\sin\theta\cos\theta}{\Delta_+(\theta)}\right).
    \end{split}
\end{equation}
Then by equation (\ref{manzu}) we know that $\phi$ satisfies
\begin{equation}
    \begin{pmatrix}
i\lambda(r-ia\cos\theta)& 0&-\sqrt{\Delta_+(r)}D_{r+}& \sqrt{\Delta_{+}(\theta)}L_{\theta+}\\
0& i\lambda(r-ia\cos\theta)& \sqrt{\Delta_{+}(\theta)}L_{\theta-}& \sqrt{\Delta_+(r)}D_{r-}\\
\sqrt{\Delta_+(r)}D_{r-}& -\sqrt{\Delta_{+}(\theta)}L_{\theta+}&i\lambda(r+ia\cos\theta) &0 \\
    -\sqrt{\Delta_{+}(\theta)}L_{\theta-}&     -\sqrt{\Delta_+(r)}D_{r+}&0 &i\lambda(r+ia\cos\theta)
    \end{pmatrix}\phi=0.
\end{equation}
By moving the angular term $\theta$ from the above equation to the right hand side, we can separate the Dirac equation (\ref{diraceq3}) for $\Delta_+(r)>0$ into the following equation:
\begin{equation}\label{reduced+}
    D\phi=L\phi,
\end{equation}
where the matrix operators $D$ and $L$ are
\begin{equation}
    D=\begin{pmatrix}
-i\lambda r &0&\sqrt{\Delta_+(r)}D_{r+}&0 \\
0&i\lambda r&0&\sqrt{\Delta_+(r)}D_{r-}\\
\sqrt{\Delta_+(r)}D_{r-}&0&i\lambda r&0\\
        0&\sqrt{\Delta_+(r)}D_{r+}&0&-i\lambda r
    \end{pmatrix},
\end{equation}
and
\begin{equation}
    L=\begin{pmatrix}
a\lambda\cos\theta& 0& 0&\sqrt{\Delta_+(\theta)}L_{\theta+}\\
0&-a\lambda \cos\theta &-\sqrt{\Delta_+(\theta)}L_{\theta-} &0\\
0&\sqrt{\Delta_+(\theta)}L_{\theta+} &a\lambda\cos\theta &0\\
 -\sqrt{\Delta_+(\theta)}L_{\theta-}    & 0&0 &-a\lambda\cos\theta   
    \end{pmatrix},
\end{equation}
respectively.

Next, we discuss how to obtain the radial equations from  (\ref{reduced+}).

Since we are considering the nontrivial solution, there exists $\theta_0\in(0,\pi)$, such that $Y_+(\theta_0)$ or $Y_-(\theta_0)$ is non-zero. Without loss of generality, we assume that $Y_-(\theta_0)\neq 0$. According to (\ref{reduced+}), we have
\begin{equation}\label{3-40}
    \begin{split}
-i\lambda rX_-+\sqrt{\Delta_+(r)}D_{r+}X_+&=\frac{\left(a\lambda\cos\theta_0Y_-(\theta_0)+\sqrt{\Delta_+(\theta_0)}\Big(L_{\theta_+}Y_+\Big)\Big|_{\theta_0}\right)}{Y_-(\theta_0)}X_-,\\
        \sqrt{\Delta_+(r)}D_{r-}X_-+i\lambda rX_+&=\frac{\left(a\lambda\cos\theta_0Y_-(\theta_0)+\sqrt{\Delta_+(\theta_0)}\Big(L_{\theta_+}Y_+\Big)\Big|_{\theta_0}\right)}{Y_-(\theta_0)}X_+.
    \end{split}
\end{equation}
Let
\begin{equation}\label{epsilon+}
    \epsilon_+\triangleq \frac{\left(a\lambda\cos\theta_0Y_-(\theta_0)+\sqrt{\Delta_+(\theta_0)}\Big(L_{\theta_+}Y_+\Big)\Big|_{\theta_0}\right)}{Y_-(\theta_0)}\in \C.
\end{equation}
Hence, substituting (\ref{3-40}) into  (\ref{reduced+}), it follows that
\begin{equation}\label{3-42}
    D\phi=L\phi=\epsilon_{+}\phi.
\end{equation}
The following lemma states that the constant $\epsilon_+$ in (\ref{epsilon+}) is actually a real number.

\begin{lem}\label{lemma3.1}
    $\epsilon_+\in \R$.
\end{lem}

\begin{proof}
Since we are considering the nontrivial solution, then there exists $r_0>0$ such that $X_-(r_0)$ or $X_+(r_0)$ is non-zero. Without loss of generality, we assume that $X_-(r_0)\neq 0$.

Since $L\phi=\epsilon_+\phi$, then
\begin{equation}
    \begin{split}
a\lambda\cos\theta X_-(r)Y_-(\theta)+\sqrt{\Delta_+(\theta)}L_{\theta+}Y_+(\theta)X_-(r)&=\epsilon_+X_-(r)Y_-(\theta),\\
       -\sqrt{\Delta_+(\theta)}L_{\theta-}Y_-(\theta)X_-(r)-a\lambda\cos\theta X_-(r)Y_+&=\epsilon_+X_-(r)Y_+(\theta).
    \end{split}
\end{equation}
By taking $r=r_0$ on both sides of the above equation, it follows that
\begin{equation}\label{3-44}
    \begin{pmatrix}
a\lambda\cos\theta & \sqrt{\Delta_+(\theta)}L_{\theta+}\\
        -\sqrt{\Delta_+(\theta)}L_{\theta-}& -a\lambda\cos\theta
    \end{pmatrix}\begin{pmatrix}
Y_-\\
     Y_+   
    \end{pmatrix}=\epsilon_+\begin{pmatrix}
Y_-\\
     Y_+   
    \end{pmatrix}.
\end{equation}
Let
\begin{equation}
    \mathcal{L}:=\begin{pmatrix}
a\lambda\cos\theta & \sqrt{\Delta_+(\theta)}L_{\theta+}\\
        -\sqrt{\Delta_+(\theta)}L_{\theta-}& -a\lambda\cos\theta
    \end{pmatrix},\quad Y=\begin{pmatrix}
Y_-\\
     Y_+   
    \end{pmatrix},
\end{equation}
then the equation (\ref{3-44}) becomes
\begin{equation}
    \mathcal{L}Y=\epsilon_+Y.
\end{equation}

Since $\Psi=S_+^{-1}\Phi$, i.e. in the Boyer-Lindquist coordinate we have
\begin{equation}
\begin{split}
    \Psi=\begin{pmatrix}
\Psi_1\\
\Psi_2\\
\Psi_3\\
        \Psi_4
    \end{pmatrix}&=e^{-i\left(\omega t+(k+\frac{1}{2})\varphi\right)}S_+^{-1}\begin{pmatrix}
X_-(r)Y_-(\theta)\\
X_+(r)Y_+(\theta)\\
X_+(r)Y_-(\theta)\\
X_-(r)Y_+(\theta)
    \end{pmatrix}\\
    &=e^{-i\left(\omega t+(k+\frac{1}{2})\varphi\right)}\Delta_+(r)^{-\frac{1}{4}}\begin{pmatrix}
(r+ia\cos\theta)^{-\frac{1}{2}}X_-(r)Y_-(\theta)\\
(r+ia\cos\theta)^{-\frac{1}{2}}X_+(r)Y_+(\theta)\\
(r-ia\cos\theta)^{-\frac{1}{2}}X_+(r)Y_-(\theta)\\
     (r-ia\cos\theta)^{-\frac{1}{2}}X_-(r)Y_+(\theta)
    \end{pmatrix}.
    \end{split}
\end{equation}
Fixing $t_0\in\R$, since $X_-(r_0)\neq 0$, then it follows that
\begin{equation}
\begin{split}
    Y_-(\theta)&=\frac{e^{i\omega t_0}}{X_-(r_0)}\Delta_+(r_0)^{\frac{1}{4}}(r+ia\cos\theta)^{\frac{1}{2}}\Psi_1(t_0,r_0,\theta,\varphi),\\
Y_+(\theta)&=\frac{e^{i\omega t_0}}{X_-(r_0)}\Delta_+(r_0)^{\frac{1}{4}}(r-ia\cos\theta)^{\frac{1}{2}}\Psi_4(t_0,r_0,\theta,\varphi),
    \end{split}
\end{equation}
i.e. $Y_+$, $Y_-$, $\partial_{\theta}Y_+$ and $\partial_{\theta}Y_-$ are uniformly continuous on the closed interval $[0,\pi]$.

For the n-dimensional complex vector space $\C^n$, let $\left<\cdot,\cdot\right>$ denotes the Hermitian inner product on $\C^n$, i.e. for any $x,y\in {\C}^n$ we have
\begin{equation}
\left<x,y\right>:=\sum_{i=1}^nx_i\bar{y_i}.
\end{equation}
Hence it follows that
\begin{equation}
\begin{split}
    &\int_{S^2}\left<\mathcal{L}Y,Y\right>-\left<Y,\mathcal{L}Y\right>dS^2\\
    &=\int_{S^2}\overline{Y_-}\Bigg[a\lambda\cos\theta Y_-+\sqrt{\Delta_+(\theta)}\Bigg(\partial_{\theta}-\frac{1}{\Delta_+(\theta)}\left(\omega a \sin\theta+\frac{(k+\frac{1}{2})E_+}{\sin\theta}-P\cot\theta\right)\\
    &\qquad+\frac{1}{2}\left(\cot\theta-\frac{\kappa^2a^2\sin\theta\cos\theta}{\Delta_+(\theta)}\right)\Bigg)Y_+\Bigg]dS^2 \\
    &\qquad  -\int_{S^2}\overline{Y_+}\Bigg[a\lambda\cos\theta Y_++\sqrt{\Delta_+(\theta)}\Bigg(\partial_{\theta}+\frac{1}{\Delta_+(\theta)}\left(\omega a \sin\theta+\frac{(k+\frac{1}{2})E_+}{\sin\theta}-P\cot\theta\right)\\
    &\qquad+\frac{1}{2}\left(\cot\theta-\frac{\kappa^2a^2\sin\theta\cos\theta}{\Delta_+(\theta)}\right)\Bigg)Y_-\Bigg]dS^2\\
    &\qquad -\int_{S^2}{Y_-}\Bigg[a\lambda\cos\theta \overline{Y_-}+\sqrt{\Delta_+(\theta)}\Bigg(\partial_{\theta}-\frac{1}{\Delta_+(\theta)}\left(\omega a \sin\theta+\frac{(k+\frac{1}{2})E_+}{\sin\theta}-P\cot\theta\right)\\
    &\qquad+\frac{1}{2}\left(\cot\theta-\frac{\kappa^2a^2\sin\theta\cos\theta}{\Delta_+(\theta)}\right)\Bigg)\overline{Y_+}\Bigg]dS^2\\
    &\qquad +\int_{S^2}{Y_+}\Bigg[a\lambda\cos\theta \overline{Y_+}+\sqrt{\Delta_+(\theta)}\Bigg(\partial_{\theta}+\frac{1}{\Delta_+(\theta)}\left(\omega a \sin\theta+\frac{(k+\frac{1}{2})E_+}{\sin\theta}-P\cot\theta\right)\\
    &\qquad+\frac{1}{2}\left(\cot\theta-\frac{\kappa^2a^2\sin\theta\cos\theta}{\Delta_+(\theta)}\right)\Bigg)\overline{Y_-}\Bigg]dS^2,
    \end{split}
\end{equation}
i.e.
\begin{equation}\label{3-47}
    \begin{split}
    &\int_{S^2}\left<\mathcal{L}Y,Y\right>-\left<Y,\mathcal{L}Y\right>dS^2\\
        &=\int_{S^2}\sqrt{\Delta_+(\theta)}\left(\partial_{\theta}\big(Y_+\overline{Y_-}\big)+\Big(\cot\theta-\frac{\kappa^2a^2\sin\theta\cos\theta}{\Delta_+(\theta)} \Big)Y_+\overline{Y_-} \right)dS^2\\
&-\int_{S^2}\sqrt{\Delta_+(\theta)}\left(\partial_{\theta}\big(Y_-\overline{Y_+}\big)+\Big(\cot\theta-\frac{\kappa^2a^2\sin\theta\cos\theta}{\Delta_+(\theta)} \Big)Y_-\overline{Y_+} \right)dS^2.
    \end{split}
\end{equation}
Next we show that the first integral after the equal sign of (\ref{3-47}) is zero, and the second integral is also equal to zero by a quite similar discussion. In fact, let $f$ be the function and $V$ be the vector field defined as follows
\begin{equation}
    f\triangleq \sqrt{\Delta_+(\theta)},\quad V\triangleq Y_+\overline{Y_-}\partial_{\theta},
\end{equation}
then we have
\begin{equation}
    \begin{split}
\text{div}\left(fV\right)&=\nabla f\cdot V+f\text{div}V\\
&=(\partial_{\theta}f)Y_+\overline{Y_-}+f\left(\frac{1}{\sin\theta}\partial_{\theta}(Y_+\overline{Y_-}\sin\theta)\right)\\
&=\left(\partial_{\theta}f+\frac{\cos\theta}{\sin\theta}f\right)Y_+\overline{Y_-}+f\partial_{\theta}\left(Y_+\overline{Y_-}\right)\\
&=\left(-{\frac{\kappa^2a^2\sin\theta\cos\theta}{\sqrt{\Delta_+}(\theta)}}+\cot\theta\sqrt{\Delta_+(\theta)}\right)Y_+\overline{Y_-}+\sqrt{\Delta_+(\theta)}\partial_{\theta}\left(Y_+\overline{Y_-}\right),
    \end{split}
\end{equation}
i.e.
\begin{equation}
    \text{div}\left(fV\right)=\sqrt{\Delta_+(\theta)}\left(\partial_{\theta}\big(Y_+\overline{Y_-}\big)+\Big(\cot\theta-\frac{\kappa^2a^2\sin\theta\cos\theta}{\Delta_+(\theta)} \Big)Y_+\overline{Y_-} \right).
\end{equation}
Therefore, by the divergence Theorem, it follows that
\begin{equation}
\begin{split}
    &\quad\int_{S^2}\sqrt{\Delta_+(\theta)}\left(\partial_{\theta}\big(Y_+\overline{Y_-}\big)+\Big(\cot\theta-\frac{\kappa^2a^2\sin\theta\cos\theta}{\Delta_+(\theta)} \Big)Y_+\overline{Y_-} \right)dS^2\\
    &=\int_{S^2}\text{div}\left(fV\right)dS^2\\
&=0.
    \end{split}
\end{equation}
Similarly,
\begin{equation}
\begin{split}
    &\quad\int_{S^2}\sqrt{\Delta_+(\theta)}\left(\partial_{\theta}\big(Y_-\overline{Y_+}\big)+\Big(\cot\theta-\frac{\kappa^2a^2\sin\theta\cos\theta}{\Delta_+(\theta)} \Big)Y_-\overline{Y_+} \right)dS^2\\
    &=\int_{S^2}\text{div}\left(fY_-\overline{Y_+}\partial_{\theta}\right)dS^2 \\
&=0.
    \end{split}
\end{equation}
Thus, according to (\ref{3-47}) we have
\begin{equation}
\int_{S^2}\left<\mathcal{L}Y,Y\right>dS^2=\int_{S^2}\left<Y,\mathcal{L}Y\right>dS^2,
\end{equation}
i.e. 
\begin{equation}
\epsilon_+\int_{S^2}\left|Y\right|^2dS^2=\overline{\epsilon_+}\int_{S^2}\left|Y\right|^2dS^2.
\end{equation}
Since the solution is nontrivial, we have
\begin{equation}
    \epsilon_+=\overline{\epsilon_+},
\end{equation}
i.e. $\epsilon_+\in\R$.

\end{proof}

By (\ref{3-42}), it follows immediately that the radial equations when $\Delta_+(r)>0$ are
\begin{equation}\label{3-58}
    \begin{split}
&\frac{dX_+}{dr}-\frac{i}{\Delta_+(r)}\left(\omega(r^2+a^2)+Qr+\big(k+\frac{1}{2}\big)E_+a\right)X_+-\frac{i\lambda r+\epsilon_+}{\sqrt{\Delta_+(r)}}X_-=0,\\
        &\frac{dX_-}{dr}+\frac{i}{\Delta_+(r)}\left(\omega(r^2+a^2)+Qr+\big(k+\frac{1}{2}\big)E_+a\right)X_-+\frac{i\lambda r-\epsilon_+}{\sqrt{\Delta_+(r)}}X_+=0.
    \end{split}
\end{equation}

According to the radial equation (\ref{3-58}), we have the following Lemma:
\begin{lem}\label{lemma3.2}
If $\Delta_+(r)>0$, then
    \begin{equation}
        \frac{d}{dr}\left(\left|X_+\right|^2-\left|X_-\right|^2 \right)=0.
    \end{equation}
\end{lem}

\begin{proof}
For a complex number $z$, let $\mathscr{Re}(z)$ be the real part of $z$. If $\Delta_+(r)>0$, then
\begin{equation}
    \begin{split}
&\frac{1}{2}\frac{d}{dr}\left(\left|X_+\right|^2-\left|X_-\right|^2 \right)=\mathscr{Re}\left(\frac{dX_+}{dr}\overline{X_+}-\frac{dX_-}{dr}\overline{X_-} \right)\\
        &=\mathscr{Re}\left(\Big(\frac{i}{\Delta_+(r)}\big(\omega(r^2+a^2)+Qr+(k+\frac{1}{2})E_+a\big)X_++\frac{i\lambda r+\epsilon_+}{\sqrt{\Delta_+(r)}}X_-\Big)\overline{X_+} \right)\\
        &\quad-\mathscr{Re}\left(\Big(\frac{-i}{\Delta_+(r)}\big(\omega(r^2+a^2)+Qr+(k+\frac{1}{2})E_+a\big)X_-+\frac{-i\lambda r+\epsilon_+}{\sqrt{\Delta_+(r)}}X_+\Big)\overline{X_-} \right)\\
        &=\mathscr{Re}\left(\frac{i}{\Delta_+(r)}\Big(\omega(r^2+a^2)+Qr+(k+\frac{1}{2})E_+a\Big)\big(X_+\overline{X_+}+X_-\overline{X_-}\big) \right)\\
        &\quad+\mathscr{Re}\left(\frac{i\lambda r}{\sqrt{\Delta_+(r)}}\big(X_-\overline{X_+}+X_+\overline{X_-}\big)\right)\\
        &\quad+\mathscr{Re}\left(\frac{\epsilon_+}{\sqrt{\Delta_+(r)}}\big(X_-\overline{X_+}-X_+\overline{X_-}\big)\right)\\
        &=0.
    \end{split}
\end{equation}
Therefore, 
\begin{equation}
        \frac{d}{dr}\left(\left|X_+\right|^2-\left|X_-\right|^2 \right)=0.
\end{equation}

\end{proof}

Let $r_1\in(r_-,r_+)$ be some fixed positive constant and $M_{(r_1,r_+)}$ be the time slice of non-extreme Kerr-Newman-dS spacetime satisfying $\left\{t=\text{constant}\right\}$ and $r_1<r<r_+$. Moreover, by the similar assumptions as in \cite{kn} and \cite{wyhzx}, we assume that $X_+=0$ or $X_-=0$ on the horizons can match the solution inside and outside the horizons.

The following nonexistence theorem is the main result of this section:
\begin{thm}
Let $\Psi$ be the solution of the Dirac equation
\begin{equation}
\big(\mathscr{D}+ie^{\alpha}A(e_{\alpha})+i\lambda\big)\Psi=0
\end{equation}
in the exterior region $r_-<r<r_+$ of the non-extreme Kerr-Newman-dS spacetime, and it is of the form
\begin{equation}
    \Psi={S_+}^{-1}\Phi,
\end{equation}
where
\begin{equation}
    \Phi=e^{-i\left(\omega t+(k+\frac{1}{2})\varphi\right)}\begin{pmatrix}
X_{-}(r)Y_-(\theta)\\
X_{+}(r)Y_+(\theta)\\
X_{+}(r)Y_-(\theta)\\
X_{-}(r)Y_+(\theta)\\
    \end{pmatrix},
\end{equation}
$k\in\mathbb{Z}$, and $S_+$ is the following diagonal matrix
\begin{equation}
    S_+=\Delta_+(r)^{\frac{1}{4}}\begin{pmatrix}
 (r+ia\cos\theta)^{\frac{1}{2}}& 0&0 & 0\\
0& (r+ia\cos\theta)^{\frac{1}{2}}&0 & 0\\
0& 0&(r-ia\cos\theta)^{\frac{1}{2}} & 0\\
        0& 0&0 & (r-ia\cos\theta)^{\frac{1}{2}}\\
    \end{pmatrix}.
\end{equation}
Then for arbitrary $(\lambda,p)\in \R\times \big[2,+\infty\big)$, if 
\begin{equation}
\Psi\in L^p\left(M_{(r_1,r_+)}\right),
\end{equation}
 then $\Psi\equiv 0$.
\end{thm}

\begin{proof}
Since $\Delta_+(r)>0$ on $r_-<r<r_+$, by the radial equation (\ref{3-58}) we have
\begin{equation}
    \partial_{r}\Phi=E\cdot\Phi,
\end{equation}
where
\begin{equation}
    E=\begin{pmatrix}
-i\alpha_1 &0&-i\beta_1+\gamma_1&0\\
0&i\alpha_1&0&i\beta_1+\gamma_1\\
i\beta_1+\gamma_1&0&i\alpha_1&0\\
        0&-i\beta_1+\gamma_1&0&-i\alpha_1
    \end{pmatrix}
\end{equation}
and
\begin{equation}
    \begin{split}
\alpha_1&=\frac{1}{\Delta_+(r)}\left(\omega(r^2+a^2)+Qr+(k+\frac{1}{2})E_+a \right),\\
\beta_1&=\frac{\lambda r}{\sqrt{\Delta_+(r)}},\\
        \gamma_1&=\frac{\epsilon_+}{\sqrt{\Delta_+(r)}}.
    \end{split}
\end{equation}
Thus, we have
\begin{equation}
    \begin{split}
\partial_r\left(|\Phi|^2\right)&=\partial_r\left(\Phi^T\cdot\overline{\Phi}\right)\\
        &=2\Phi^T\begin{pmatrix}
0&0&i\beta_1+\gamma_1&0\\
0&0&0&-i\beta_1+\gamma_1\\
-i\beta_1+\gamma_1&0&0&0\\
            0&i\beta_1+\gamma_1&0&0
\end{pmatrix}\overline{\Phi}\triangleq2\Phi^TA\overline{\Phi}.
    \end{split}
\end{equation}
Notice that
\begin{equation}
    \overline{A}^TA=\begin{pmatrix}
\beta_1^2+\gamma_1^2&0&0&0\\
0&\beta_1^2+\gamma_1^2&0&0\\
0&0&\beta_1^2+\gamma_1^2&0\\
            0&0&0&\beta_1^2+\gamma_1^2
\end{pmatrix},
\end{equation}
Hence, according to the Cauchy-Schwarz inequality and the compatibility of the matrix norm, the following estimates can be obtained
\begin{equation}\label{3-73}
\left|\partial_r\left(|\Phi|^2\right)\right|=2\left|\left<\Phi,\overline{A}\Phi\right>\right|\le 2|\Phi|\cdot |A\overline{\Phi}|\le 2||A||_{2}\cdot|\Phi|^2\le 2\sqrt{\beta_1^2+\gamma_1^2}|\Phi|^2.
\end{equation}
Next, we claim that there exists a constant $C>0$ and $r_1\in(r_-,r_+)$ such that for any $r\in(r_-,r_1]$, the following inequality holds
\begin{equation}
    \sqrt{\beta_1^2+\gamma_1^2}\le C(r-r_-)^{-\frac{1}{2}}.
\end{equation}
In fact, taking $r_1=\frac{r_-+r_+}{2}$, it follows that
\begin{equation}
\begin{split}
    (r-r_-)^{\frac{1}{2}}\sqrt{\beta_1^2+\gamma_1^2}&=(r-r_-)^{\frac{1}{2}}\sqrt{\frac{\lambda^2 r^2+\epsilon_+^2}{\Delta_+(r)}}\\
&=\sqrt{\frac{\lambda^2 r^2+\epsilon_+^2}{\kappa^2(r_+-r)(r-r_c)(r+r_++r_-+r_c)}}\\
&\le \sqrt{\frac{\lambda^2r_+^2+\epsilon_+^2}{\kappa^2(\frac{r_+-r_-}{2})(r_--r_c)r_+}}\triangleq C.
    \end{split}
\end{equation}
Combining with (\ref{3-73}), for any $r_-<r\le r_1$, we have
\begin{equation}\label{3-76}
    \left|\partial_r\left(|\Phi|^2\right)\right|\le C(r-r_-)^{-\frac{1}{2}}|\Phi|^2.
\end{equation}
Therefore, by the Gronwall Lemma \cite{bahouri}, for any $r_-<s<r\le r_1$, we can deduce that
\begin{equation}\label{gron}
    |\Phi(r)|\le |\Phi(s)|\exp\left(C\int_s^r(\overline{r}-r_-)^{-\frac{1}{2}}d\overline{r}\right).
\end{equation}
If $|\Phi|^2$ has a zero on $r\in(r_-,r_+)$, then by the existence and uniqueness theorem for solutions of ordinary differential equations, i.e. Theorem \ref{odeeu}, we have $\Phi\equiv 0$. Hence we can assume that $|\Phi|^2>0$ for $r\in(r_-,r_+)$. Therefore, deviding by $|\Phi|^2$ on both sides of (\ref{3-76}) and integrate, it follows that
\begin{equation}
    -\int_{s}^r(\overline{r}-r_-)^{-\frac{1}{2}}d\overline{r}\le \log(|\Phi|^2)\Big|_{s}^r \le \int_{s}^r(\overline{r}-r_-)^{-\frac{1}{2}}d\overline{r}
\end{equation}
for any $r_-<s<r\le r_1$. Therefore, there exists a constant $C_1>0$ such that for arbitrary $r_-<s<r\le r_1$,
\begin{equation}
    \left|\log(|\Phi(r)|^2)-\log(|\Phi(s)|^2)\right|\le C_1|r-s|,
\end{equation}
i.e. $\log(|\Phi(r)|)$ is uniformly continuous on $(r_-,r_1)$, which implies that $|\Phi(r)|$ is uniformly continuous on $(r_-,r_1)$. Hence we have $|\Phi|<\infty$ at $r=r_-$. Moreover, according to (\ref{gron}), if $|\Phi|=0$ at $r=r_-$, then $|\Phi|$ is identically equal to zero on the interval $(r_-,r_1]$. On the other hand, by Lemma \ref{lemma3.2}, there exists a constant $C_0$ such that
\begin{equation}
    |X_+|^2=|X_-|^2+C_0
\end{equation}
on $r_-<r<r_+$. Substituting the expression (\ref{3-8}) of $\Psi_1$, now we have
\begin{equation}
    |\Psi_1|^2=\Psi_1\cdot\overline{\Psi_1}=\frac{1}{\sqrt{U}}\Delta_+(r)^{-\frac{1}{2}}\big(|Y_-|^2+|Y_+|^2\big)\big(C_0+2|X_-|^2\big).
\end{equation}
If $C_0\neq 0$, without loss of generality, we can assume that $C_0>0$. Hnece we have
\begin{equation}
    \frac{1}{\sqrt{U}}\Delta_+(r)^{-\frac{1}{2}}\left(|Y_+|^2+|Y_-|^2\right)\in L^{\frac{p}{2}}\left(M_{(r_1,r_+)}\right),
\end{equation}
i.e.
\begin{equation}\label{3-83}
    \int_{M_{(r_1,r_+)}}\left(\frac{1}{\sqrt{U}}\Delta_+(r)^{-\frac{1}{2}}\right)^{\frac{p}{2}}|Y|^{p}\sqrt{\frac{UV_+\sin^2\theta}{E_+^2\Delta_+(r)\Delta_+(\theta)}}dr\,d\theta\,d\varphi<\infty,
\end{equation}
where
\begin{equation}
    V_+=(2mr-Q^2-P^2)a^2\sin^2\theta+U(r^2+a^2)(1+\kappa^2a^2).
\end{equation}
By the following relationships between roots and coefficients
\begin{equation}
    \begin{split}
m&=\frac{1}{2}\kappa^2(r_++r_-)(r_++r_c)(r_c+r_-),\\
a^2&=\frac{1}{\kappa^2}-(r_c^2+r_-^2+r_+^2+r_cr_-+r_cr_++r_-r_+),\\
        Q^2+P^2&=\kappa^2r_cr_+r_-(r_c+r_++r_-)-a^2,
    \end{split}
\end{equation}
it follows that
\begin{equation}
    \begin{split}
2mr-(Q^2+P^2)&=\kappa^2r(r_++r_-)(r_++r_c)(r_c+r_-) \\
        &\qquad -\kappa^2r_cr_+r_-(r_c+r_++r_-)+a^2\\
        &=\kappa^2\left(r(r_++r_-)(r_++r_c)(r_c+r_-)-r_cr_+r_-(r_c+r_++r_-) \right)+a^2\\
        &>\kappa^2\left(rr_+^2r_c+rr_-^2r_++rr_c^2r_+-r_-r_c^2r_+-r_-r_+^2r_c-r_cr_-^2r_+\right)\\
        &=\kappa^2\left((r-r_-)r_+^2r_c+(r-r_c)r_-^2r_++(r-r_-)r_c^2r_+\right) \\
        &>0
    \end{split}
\end{equation}
on $r_-<r_1<r<r_+$. Thus, for $r\in(r_1,r_+)$, we have
\begin{equation}
\begin{split}
    \sqrt{V_+}&\ge \sqrt{U(r^2+a^2)(1+\kappa^2a^2)}\\
&=\sqrt{(r^2+a^2\cos^2\theta)(r^2+a^2)(1+\kappa^2a^2)}\\
&\ge \sqrt{r^4}=r^2.
    \end{split}
\end{equation}
Since there exists a constant $C_2>0$ such that $|Y|^2=|Y_+|^2+|Y_-|^2>C_2$ on $\left[\frac{\pi}{4},\frac{\pi}{2}\right]$ (otherwise $\Psi\equiv 0$), by (\ref{3-83}), it follows that
\begin{equation}
    \int_{r_1}^{r_+}\int_{\frac{\pi}{4}}^{\frac{\pi}{2}}\int_{0}^{2\pi}\left(\frac{1}{\sqrt{U}}\Delta_+(r)^{-\frac{1}{2}}\right)^{\frac{p}{2}}r^3\sqrt{\frac{\sin^2\theta}{\Delta_+(r)\Delta_+(\theta)}}dr\,d\theta\,d\varphi<\infty,
\end{equation}
i.e.
\begin{equation}\label{3-89}
    \int_{r_1}^{r_+}\int_{\frac{\pi}{4}}^{\frac{\pi}{2}}\int_{0}^{2\pi}r^3U^{-\frac{p}{4}}\Big(\Delta_+(r)\Big)^{-\frac{1}{2}-\frac{p}{4}}\sqrt{\frac{\sin^2\theta}{\Delta_+(\theta)}}dr\,d\theta\,d\varphi<\infty.
\end{equation}
Notice that
\begin{equation}
    r^3U^{-\frac{p}{4}}=\frac{r^3}{(r^2+a^2\cos^2\theta)^{\frac{p}{4}}}\ge \frac{r_-^3}{(r_+^2+a^2)^{\frac{p}{4}}}>0
\end{equation}
on $r\in(r_1,r_+)$. Therefore, combining with (\ref{3-89}) we have
\begin{equation}
    \int_{r_1}^{r_+}\frac{1}{(r_+-r)^{\frac{p}{4}+\frac{1}{2}}}dr<\infty.
\end{equation}
However, since $p\ge 2$, i.e. $\frac{p}{4}+\frac{1}{2}\ge 1$, it is a contraction!  Therefore, we have $C_0=0$. Hence
\begin{equation}
    |X_+|=|X_-|
\end{equation}
on $r\in(r_-,r_+)$. Since the limit of $|\Phi|$ exists at $r=r_-$ and
\begin{equation}
    |\Phi|^2=2\left(|Y_+|^2+|Y_-|^2\right)\cdot |X_+|^2,
\end{equation}
we can deduce that the limits of $|X_+|$ and $|X_-|$ also exist at $r=r_-$. Therefore, according to the matching conditions at $r=r_-$, it follows that
\begin{equation}
    |X_+|=|X_-|=\left|\lim_{r_-<r\rightarrow r_-}X_-\right|=0.
\end{equation}
Hence, $|\Phi|$ vanishes at $r=r_-$. Then by (\ref{gron}), we have $\Phi\equiv 0$ on $r\in[r_-,r_+)$ which implies that
\begin{equation}
    \Psi\equiv 0.
\end{equation}

\end{proof}

\mysection{Extreme Kerr-Newman-dS spacetime}\ls

In this section, we consider the necessary conditions for the existence of nontrivial $L^p$ integrable time-periodic solutions of the Dirac equation
\begin{equation}\label{4-1}
\big(\mathscr{D}+ie^{\alpha}A(e_{\alpha})+i\lambda\big)\Psi=0
\end{equation}
in the extreme Kerr-Newman-dS spacetime. More specifically, we give the equation relationship between $\omega$, the radius of the event horizon, the angular momentum, the charge, and the cosmological constant, which generalize the conclusion obtained by \cite{extremekerr} in the extreme Kerr-Newman spacetime (zero cosmological constant).

\begin{defn}
The Kerr-Newman-dS spacetime is called extreme, if the polynomial of order 4 in $r$
\begin{equation}
    \Delta_+(r)=(r^2+a^2)(1-\kappa^2r^2)-2mr+Q^2+P^2
\end{equation}
has exactly 4 real roots, i.e. a double positive root $r=r_-$, a simple positive root $r=r_+>r_-$ and a negative root $r_n=-(2r_-+r_+)$. Moreover, $m$ satisfies the following equality
\begin{equation}
\begin{split}
    m=\frac{1}{\sqrt{54}}&\Bigg((1-\kappa^2a^2)-\sqrt{\big(1-\kappa^2a^2\big)^2-12\kappa^2(a^2+Q^2+P^2)} \Bigg)^{\frac{1}{2}}\\
&\times\Bigg(2(1-\kappa^2a^2)^2+ \sqrt{\big(1-\kappa^2a^2\big)^2-12\kappa^2(a^2+Q^2+P^2)} \Bigg).
    \end{split}
\end{equation}

\end{defn}

In the extreme circumstances, since $r = r_-$ is a double root of $\Delta_+ (r) $, $\Delta_+(r)^{-\frac{1}{2}}$ is not integrable near $r = r_- $. Thus, the method in Chapter 3 when dealing with the non-extreme case is not applicable at this time. In this section, we mainly refer to the method of \cite{extremekerr} when dealing with the extreme Kerr-Newman spacetime (zero cosmological constant) and the necessary conditions for the existence of nontrivial $L^p$ integrable  time-periodic solutions of the Dirac equation (\ref{4-1}) in the extreme Kerr-Newman-dS spacetime are given (considered in the exterior region  $r_-<r<r_+$ of the spacetime).

When $\Delta_+(r)>0$, the radial equations are as follows
\begin{equation}\label{4-6}
    \begin{split}
&\frac{dX_+}{dr}-\frac{i}{\Delta_+(r)}\left(\omega(r^2+a^2)+Qr+\big(k+\frac{1}{2}\big)E_+a\right)X_+-\frac{i\lambda r+\epsilon_+}{\sqrt{\Delta_+(r)}}X_-=0,\\
        &\frac{dX_-}{dr}+\frac{i}{\Delta_+(r)}\left(\omega(r^2+a^2)+Qr+\big(k+\frac{1}{2}\big)E_+a\right)X_-+\frac{i\lambda r-\epsilon_+}{\sqrt{\Delta_+(r)}}X_+=0,
    \end{split}
\end{equation}
where $\epsilon_+\in\R$. Let
\begin{equation}
    X(r)=\begin{pmatrix}
X_+(r)\\
     X_-(r)   
    \end{pmatrix},
\end{equation}
then the equations (\ref{4-6}) are
\begin{equation}\label{4-3}
    \frac{d}{dr}X=\begin{pmatrix}
i\alpha_1 & i\beta_1+\gamma_1\\
-i\beta_1+\gamma_1& -i\alpha_1
    \end{pmatrix}X,
\end{equation}
where
\begin{equation}\label{666}
    \begin{split}
\alpha_1&=\frac{\omega(r^2+a^2)+Qr+\big(k+\frac{1}{2}\big)E_+a}{\Delta_+(r)},\\
\beta_1&=\frac{\lambda r}{\sqrt{\Delta_+(r)}},\\
        \gamma_1&=\frac{\epsilon_+}{\sqrt{\Delta_+(r)}}.
    \end{split}
\end{equation}

Next, we derive the necessary condition for $\omega$ in (\ref{666}) if the nontrivial time-periodic solution
\begin{equation}
    \Psi\in L^p\left(M_{(r_-,r_+)}\right)
\end{equation}
exists.

Since $r_-$ is a double root of $\Delta_+(r)$, we have
\begin{equation}
    \Delta_+(r)=(r-r_-)^2\kappa^2(r_+-r)(r+r_++2r_-).
\end{equation}
Let
\begin{equation}
\begin{split}
B(r)&\triangleq \kappa^2(r_+-r)(r+r_++2r_-),\\
\tau&\triangleq\omega(r_-^2+a^2)+\left(k+\frac{1}{2}\right)E_+a+r_eQ,\\
    \mu&\triangleq 2r_-\omega+Q.
\end{split}    
\end{equation}
Therefore,
\begin{equation}\label{411}
\begin{split}
\omega\left((x+r_-)^2+a^2\right)+Q(x+r_-)+\big(k+\frac{1}{2}\big)E_+a&=\omega x^2+\tau+2\omega r_- x+Qx \\
&=\tau+\mu x+\omega x^2.
\end{split}
\end{equation}

We consider in the exterior region, i.e. $r\in(r_-,r_+)$. For convenience, we define a new variable $x:=r-r_-$ and the function
\begin{equation}
    F(x):=X(x+r_-),\;x\in \left(0,r_+-r_-\right).
\end{equation}
According to  (\ref{4-3}) and (\ref{411}), we have
\begin{equation}
\begin{split}
    \partial_xF_1(x)&=\partial_rX_1\Big|_{x+r_-} \\
&=i\alpha_1(x+r_-)X_1(x+r_-)+\left(i\beta_1(x+r_-)+\gamma_1(x+r_-)\right)X_2(x+r_-)\\
&=i\frac{(\tau+\mu x+\omega x^2)}{x^2B(x+r_-)}F_1(x)+\left(\frac{i\lambda (x+r_-)}{\sqrt{x^2B(x+r_-)}}+\frac{\epsilon_+}{\sqrt{x^2B(x+r_-)}}\right)F_2(x),
    \end{split}
\end{equation}
i.e.
\begin{equation}
\begin{split}
    \partial_xF_1(x)&=\left(\frac{i\tau}{x^2B(x+r_-)}+\frac{i\mu}{xB(x+r_-)}+\frac{i\omega}{B(x+r_-)}\right)F_1(x)\\
&\qquad+\left(\frac{\epsilon_++i\lambda r_-}{x\sqrt{B(x+r_-)}}+\frac{i\lambda}{\sqrt{B(x+r_-)}}\right)F_2(x).
    \end{split}
\end{equation}
Similarly,
\begin{equation}
\begin{split}
    \partial_xF_2(x)&=\partial_rX_2\Big|_{x+r_-}\\
&=\left(-i\beta_1(x+r_-)+\gamma_1(x+r_-)\right)X_1(x+r_-)-i\alpha_1(x+r_-)X_2(x+r_-)\\
&=\left(\frac{-i\lambda (x+r_-)}{\sqrt{x^2B(x+r_-)}}+\frac{\epsilon_+}{\sqrt{x^2B(x+r_-)}}\right)F_1(x)-i\frac{(\tau+\mu x+\omega x^2)}{x^2B(x+r_-)}F_2(x)\\
&=\left(\frac{\epsilon_+-i\lambda r_-}{x\sqrt{B(x+r_-)}}-\frac{i\lambda}{\sqrt{B(x+r_-)}}\right)F_1(x)\\
&\qquad+\left(\frac{-i\tau}{x^2B(x+r_-)}+\frac{-i\mu}{xB(x+r_-)}+\frac{-i\omega}{B(x+r_-)}\right)F_2(x).
    \end{split}
\end{equation}
In the matrix form, we have that $F(x)$, $x\in\left(0,r_+-r_-\right)$ satisfying the following equation
\begin{equation}\label{416}
    \partial_xF=\begin{pmatrix}
\frac{i\tau}{x^2B(x+r_-)}+\frac{i\mu}{xB(x+r_-)}+\frac{i\omega}{B(x+r_-)} & \frac{\epsilon_++i\lambda r_-}{x\sqrt{B(x+r_-)}}+\frac{i\lambda}{\sqrt{B(x+r_-)}} \\
 \frac{\epsilon_+-i\lambda r_-}{x\sqrt{B(x+r_-)}}-\frac{i\lambda}{\sqrt{B(x+r_-)}} &  \frac{-i\tau}{x^2B(x+r_-)}+\frac{-i\mu}{xB(x+r_-)}+\frac{-i\omega}{B(x+r_-)}    
    \end{pmatrix}F.
\end{equation}

Now we define the function
\begin{equation}
    W(x):=T\cdot F(x),\;x\in\left(0,r_+-r_-\right),
\end{equation}
where $T$ is the unitary matrix
\begin{equation}
    T=\begin{pmatrix}
\frac{-1}{\sqrt{2}}& \frac{-1}{\sqrt{2}}\\
        \frac{-i}{\sqrt{2}}& \frac{i}{\sqrt{2}}
    \end{pmatrix}.
\end{equation}
According to the properties of the unitary matrix, for any $x\in\left(0,r_+-r_-\right)$, we have
\begin{equation}
    \left|W(x)\right|=\left|F(x)\right|.
\end{equation}
Moreover, by (\ref{416}), we can deduce that
\begin{equation}
\begin{split}
    \partial_xW_1(x)&=-\frac{1}{\sqrt{2}}\left(\partial_xF_1(x)+\partial_xF_2(x)\right)\\
&=-\frac{1}{\sqrt{2}}\Bigg[\left(\frac{i\tau}{x^2B}+\frac{i\mu}{xB}+\frac{i\omega}{B}\right)F_1(x)+\left(\frac{\epsilon_++i\lambda r_-}{x\sqrt{B}}+\frac{i\lambda}{\sqrt{B}}\right)F_2(x)\\
&\qquad+\left(\frac{\epsilon_+-i\lambda r_-}{x\sqrt{B}}-\frac{i\lambda}{\sqrt{B}}\right)F_1(x)-\left(\frac{i\tau}{x^2B}+\frac{i\mu}{xB}+\frac{i\omega}{B}\right)F_2(x)\Bigg].
    \end{split}
\end{equation}
Hence, by the definition of $W(x)$ we have
\begin{equation}
\begin{split}
    \partial_xW_1(x)&=\frac{\epsilon_+}{x\sqrt{B}}\cdot \frac{-1}{\sqrt{2}}\left(F_1(x)+F_2(x)\right)-\frac{i}{\sqrt{2}}\left(\frac{\tau}{x^2B}+\frac{\mu}{xB}+\frac{\omega}{B}\right)\left(F_1(x)-F_2(x)\right)\\
&\qquad -\frac{i}{\sqrt{2}}\left(\frac{\lambda r_-}{x\sqrt{B}}+\frac{\lambda}{\sqrt{B}}\right)\left(F_2(x)-F_1(x)\right)\\
&=\frac{\epsilon_+}{x\sqrt{B}}W_1(x)+\left(\frac{\tau}{x^2B}+\frac{\mu}{xB}+\frac{\omega}{B}-\frac{\lambda r_-}{x\sqrt{B}}-\frac{\lambda}{\sqrt{B}}\right)W_2(x).
\end{split}
\end{equation}
Similarly, 
\begin{equation}
\begin{split}
    \partial_xW_2(x)&=\frac{-i}{\sqrt{2}}\left(\partial_xF_1(x)-\partial_xF_2(x)\right)\\
&=\frac{-i}{\sqrt{2}}\Bigg[\left(\frac{i\tau}{x^2B}+\frac{i\mu}{xB}+\frac{i\omega}{B}\right)F_1(x)+\left(\frac{\epsilon_++i\lambda r_-}{x\sqrt{B}}+\frac{i\lambda}{\sqrt{B}}\right)F_2(x)\\
&\qquad-\left(\frac{\epsilon_+-i\lambda r_-}{x\sqrt{B}}-\frac{i\lambda}{\sqrt{B}}\right)F_1(x)+\left(\frac{i\tau}{x^2B}+\frac{i\mu}{xB}+\frac{i\omega}{B}\right)F_2(x)\Bigg]\\
&=\frac{-i}{\sqrt{2}}\cdot\frac{-\epsilon_+}{x\sqrt{B}}\left(F_1(x)-F_2(x)\right)+\frac{1}{\sqrt{2}}\left(\frac{\tau}{x^2B}+\frac{\mu}{xB}+\frac{\omega}{B}\right)\left(F_1(x)+F_2(x)\right)\\
&\qquad +\frac{1}{\sqrt{2}}\left(\frac{\lambda r_-}{x\sqrt{B}}+\frac{\lambda}{\sqrt{B}}\right)\left(F_1(x)+F_2(x)\right)\\
&=\left(\frac{-\tau}{x^2B}+\frac{-\mu}{xB}+\frac{-\omega}{B}-\frac{\lambda r_-}{x\sqrt{B}}-\frac{\lambda}{\sqrt{B}}\right)W_1(x)+\frac{-\epsilon_+}{x\sqrt{B}}W_2(x).
    \end{split}
\end{equation}
Therefore, the following equation 
\begin{equation}\label{423}
    \partial_xW(x)=\begin{pmatrix}
\frac{\epsilon_+}{x\sqrt{B}} & \frac{\tau}{x^2B}+\frac{\mu}{xB}+\frac{\omega}{B}-\frac{\lambda r_-}{x\sqrt{B}}-\frac{\lambda}{\sqrt{B}}\\
   -\frac{\tau}{x^2B}-\frac{\mu}{xB}-\frac{\omega}{B}-\frac{\lambda r_-}{x\sqrt{B}}-\frac{\lambda}{\sqrt{B}} & \frac{-\epsilon_+}{x\sqrt{B}}
    \end{pmatrix}W(x)
\end{equation}
holds for $x\in\left(0,r_+-r_-\right)$, where $B\triangleq B(x+r_-)$.

For $x\in\big(0,\frac{r_+-r_-}{2}\big]$, we define a new variable $z:=\frac{1}{x}$, $z\in\big[\frac{2}{r_+-r_-},+\infty\big)$. Define the function
\begin{equation}
    V(z):=W\left(\frac{1}{z}\right),\;z\in\big[\frac{2}{r_+-r_-},+\infty\big).
\end{equation}
Thus, 
\begin{equation}
    \partial_zV(z)=-\frac{1}{z^2}W^{\prime}\Big|_{\frac{1}{z}}.
\end{equation}
Substituting (\ref{423}), we have
\begin{equation}\label{4-23}
    \partial_zV(z)=-\frac{1}{z^2}\begin{pmatrix}
\widetilde{E}_{11}(z) & \widetilde{E}_{12}(z) \\
        \widetilde{E}_{21}(z) & \widetilde{E}_{22}(z)
    \end{pmatrix}W\left(\frac{1}{z}\right),
\end{equation}
where
\begin{equation}
    \begin{split}
\widetilde{E}_{11}(z)&=\frac{\epsilon_+z}{\sqrt{B\left(\frac{1}{z}+r_-\right)}}, \\
\widetilde{E}_{12}(z)&=\frac{\tau z^2}{B\left(\frac{1}{z}+r_-\right)}+\frac{\mu z}{B\left(\frac{1}{z}+r_-\right)}+\frac{\omega}{B\left(\frac{1}{z}+r_-\right)}-\frac{\lambda r_- z}{\sqrt{B\left(\frac{1}{z}+r_-\right)}}\\
&\qquad\qquad\qquad\qquad\qquad\qquad\qquad\qquad\qquad-\frac{\lambda}{\sqrt{B\left(\frac{1}{z}+r_-\right)}}, \\
\widetilde{E}_{21}(z)&=-\frac{\tau z^2}{B\left(\frac{1}{z}+r_-\right)}-\frac{\mu z}{B\left(\frac{1}{z}+r_-\right)}-\frac{\omega}{B\left(\frac{1}{z}+r_-\right)}-\frac{\lambda r_- z}{\sqrt{B\left(\frac{1}{z}+r_-\right)}}\\
&\qquad\quad\quad\qquad\qquad\qquad\;\;\qquad\qquad\qquad\qquad-\frac{\lambda}{\sqrt{B\left(\frac{1}{z}+r_-\right)}}, \\
\widetilde{E}_{22}(z)&=-\frac{\epsilon_+z}{\sqrt{B\left(\frac{1}{z}+r_-\right)}}.
    \end{split}
\end{equation}
Sorting (\ref{4-23}) , $V(z)$ should satisfy the following equation
\begin{equation}\label{428}
    \partial_zV(z)=\begin{pmatrix}
E_{11}(z) & E_{12}(z) \\
        E_{21}(z) & E_{22}(z)
    \end{pmatrix}V(z),\;z\in \big[\frac{2}{r_+-r_-},+\infty\big),
\end{equation}
where
\begin{equation}
    \begin{split}
E_{11}(z)&=\frac{-\epsilon_+}{z\sqrt{B\left(\frac{1}{z}+r_-\right)}}, \\
E_{12}(z)&=-\frac{\tau}{B\left(\frac{1}{z}+r_-\right)}-\frac{\mu }{zB\left(\frac{1}{z}+r_-\right)}-\frac{\omega}{z^2 B\left(\frac{1}{z}+r_-\right)}+\frac{\lambda r_- }{z\sqrt{B\left(\frac{1}{z}+r_-\right)}}\\
&\qquad\qquad\qquad\qquad\qquad\qquad\qquad\qquad\qquad\qquad +\frac{\lambda}{z^2\sqrt{B\left(\frac{1}{z}+r_-\right)}}, \\
E_{21}(z)&=\frac{\tau}{B\left(\frac{1}{z}+r_-\right)}+\frac{\mu }{zB\left(\frac{1}{z}+r_-\right)}+\frac{\omega}{z^2 B\left(\frac{1}{z}+r_-\right)}+\frac{\lambda r_- }{z\sqrt{B\left(\frac{1}{z}+r_-\right)}}\\
&\qquad\quad\;\qquad\qquad\qquad\qquad\qquad\qquad\qquad\qquad +\frac{\lambda}{z^2\sqrt{B\left(\frac{1}{z}+r_-\right)}}, \\
E_{22}(z)&=\frac{\epsilon_+}{z\sqrt{B\left(\frac{1}{z}+r_-\right)}}.
    \end{split}
\end{equation}

Before proving the main theorem in this section, we need the following lemma in \cite{extremekerr}:
\begin{lem}\label{lemma41}
Let $a>0$ be a fixed constant. Let $Y(z)$ be the nontrivial solution of the following ordinary differential equation
\begin{equation}
    \frac{d}{dz}Y(z)=\big(C+R(z)\big)Y(z),\;z\in[a,+\infty),
\end{equation}
where $C$ and $R(z)$ are $2\times 2$ matrices satisfying

\noindent (i) $\det C>0$;

\noindent (ii) $\text{tr}\big(C+R(z)\big)\equiv0$;

\noindent (iii) $R(z)\rightarrow 0$ when $z\rightarrow +\infty$ and $R^{\prime}(z)$ is integrable on $[a,+\infty)$.

\noindent Then there exists constant $\delta>0$, such that
\begin{equation}
    |Y(z)|\ge \delta
\end{equation}
on $[a,+\infty)$.
\end{lem}

Let $M_{(r_-,r_+)}$ be the time slice in the extreme Kerr-Newman-dS spacetime satisfying $\left\{t=\text{constant}\right\}$ and $r_-<r<r_+$. With the above preparations, we can prove the following theorem:

\begin{thm}\label{thm4.1}
Let $\Psi$ be a nontrivial solution on the exterior region $r_-<r<r_+$ in the extreme Kerr-Newman-dS spacetime of the Dirac equation
\begin{equation}
\big(\mathscr{D}+ie^{\alpha}A(e_{\alpha})+i\lambda\big)\Psi=0
\end{equation}
which is of the form 
\begin{equation}\label{dsppp}
    \Psi={S_+}^{-1}\Phi,
\end{equation}
where
\begin{equation}
    \Phi=e^{-i\left(\omega t+(k+\frac{1}{2})\varphi\right)}\begin{pmatrix}
X_{-}(r)Y_-(\theta)\\
X_{+}(r)Y_+(\theta)\\
X_{+}(r)Y_-(\theta)\\
X_{-}(r)Y_+(\theta)\\
    \end{pmatrix},
\end{equation}
$k\in\mathbb{Z}$ and $S_+$ is the following diagonal matrix
\begin{equation}
    S_+=\Delta_+(r)^{\frac{1}{4}}\begin{pmatrix}
 (r+ia\cos\theta)^{\frac{1}{2}}& 0&0 & 0\\
0& (r+ia\cos\theta)^{\frac{1}{2}}&0 & 0\\
0& 0&(r-ia\cos\theta)^{\frac{1}{2}} & 0\\
        0& 0&0 & (r-ia\cos\theta)^{\frac{1}{2}}\\
    \end{pmatrix}.
\end{equation}
If there exists some $p\in[1,+\infty)$ such that
\begin{equation}
    \Psi\in L^p\left(M_{(r_-,r_+)}\right),
\end{equation}
then $\omega$ satisfies the following equality
\begin{equation}\label{4-50}
    \omega\left(r_-^2+a^2\right)+\left(k+\frac{1}{2}\right)E_+a+r_-Q=0.
\end{equation}
\end{thm}

\begin{proof}
We adopt the method of proof by contradiction. Assume that
\begin{equation}
    \tau=\omega\left(r_-^2+a^2\right)+\left(k+\frac{1}{2}\right)E_+a+r_-Q\neq 0.
\end{equation}
We rewrite the equation (\ref{428}) as follows
\begin{equation}
    \partial_zV(z)=\big(C+R(z)\big)V(z),\;z\in \big[\frac{2}{r_+-r_-},+\infty\big),
\end{equation}
where
\begin{equation}
    C=\begin{pmatrix}
0 & -\frac{\tau}{B(r_-)}\\
        \frac{\tau}{B(r_-)} &0
    \end{pmatrix}
\end{equation}
is a constant matrix and the 4 componets of the $2\times 2$ square matrix $R(z)$ are
\begin{equation}\label{435}
    \begin{split}
R_{11}(z)&=\frac{-\epsilon_+}{z\sqrt{B\left(\frac{1}{z}+r_-\right)}}, \\
R_{12}(z)&=\left(\frac{\tau}{B(r_-)}-\frac{\tau}{B\left(\frac{1}{z}+r_-\right)}\right)-\frac{\mu }{zB\left(\frac{1}{z}+r_-\right)}-\frac{\omega}{z^2 B\left(\frac{1}{z}+r_-\right)}\\
&\qquad\qquad+\frac{\lambda r_- }{z\sqrt{B\left(\frac{1}{z}+r_-\right)}}+\frac{\lambda}{z^2\sqrt{B\left(\frac{1}{z}+r_-\right)}}, \\
R_{21}(z)&=\left(\frac{\tau}{B\left(\frac{1}{z}+r_-\right)}-\frac{\tau}{B(r_-)}\right)+\frac{\mu }{zB\left(\frac{1}{z}+r_-\right)}+\frac{\omega}{z^2 B\left(\frac{1}{z}+r_-\right)}\\
&\qquad\qquad+\frac{\lambda r_- }{z\sqrt{B\left(\frac{1}{z}+r_-\right)}}+\frac{\lambda}{z^2\sqrt{B\left(\frac{1}{z}+r_-\right)}}, \\
R_{22}(z)&=\frac{\epsilon_+}{z\sqrt{B\left(\frac{1}{z}+r_-\right)}}. 
    \end{split}
\end{equation}
Since $\tau\neq 0$, thus
\begin{equation}
    \det C=0-\frac{-\tau^2}{B(r_-)^2}>0.
\end{equation}
By the formula (\ref{435}) of $R_{ij}(z)$, it is no hard to see that the following holds

\noindent (i) $R^{\prime}_{ij}(z)$ is integrable on $\big[\frac{2}{r_+-r_-},+\infty\big)$;

\noindent (ii) \begin{equation}
    R_{ij}(z)\longrightarrow 0,\;1\le i,j\le 2
\end{equation}
as $z\rightarrow +\infty$;

\noindent (iii) $\text{tr}\,(C+R)=0$.

\noindent Since $\Psi$ is nontrivial, hence $V\neq0$ (otherwise we can derive $X=0$ and thus $\Psi=0$). Therefore, by Lemma \ref{lemma41} (or c.f. Lemma 3.1 in \cite{extremekerr}) we know that there exists a constant $\delta>0$ such that for all $z\in\big[\frac{2}{r_+-r_-},+\infty\big)$, we have
\begin{equation}
    |V(z)|\ge \delta>0,
\end{equation}
i.e. for all $x\in\big(0,\frac{r_+-r_-}{2}\big]$,
\begin{equation}
    |W(x)|\ge \delta>0.
\end{equation}
Since
\begin{equation}
    |W(x)|=|F(x)|,
\end{equation}
which means that for any $r\in \big(r_-,r_-+\frac{r_+-r_-}{2}\big]$, we have
\begin{equation}\label{4-43}
    |X(r)|\ge \delta>0.
\end{equation}

The integrability condition
\begin{equation}
    \Psi\in L^p\left(M_{(r_-,r_+)}\right)
\end{equation}
means that
\begin{equation}
    \int_{M_{(r_-,r_+)}}\left(\frac{1}{\sqrt{U}}\Delta_+(r)^{-\frac{1}{2}}\right)^{\frac{p}{2}}|Y|^{p}|X|^{p}\sqrt{\frac{UV_+\sin^2\theta}{E_+^2\Delta_+(r)\Delta_+(\theta)}}dr\,d\theta\,d\varphi<\infty.
\end{equation}
Combining (\ref{4-43}) with the fact that there exists constant $C_1>0$ such that $|Y|^2=|Y_+|^2+|Y_-|^2>C_1$ on $\left[\frac{\pi}{4},\frac{\pi}{2}\right]$ (otherwise $\Psi\equiv 0$), we can deduce that there exists constant $C_2>0$ such that
\begin{equation}
\int_{M_{(r_-,r_+)}}|\Psi|^pdV>C_2\int_{r_-}^{r_-+\frac{r_+-r_-}{2}}\frac{1}{(r-r_-)^{\frac{p}{2}}}\cdot \frac{1}{r-r_-}dr=+\infty,
\end{equation}
which is a contraction! Therefore, 
\begin{equation}
    \tau=\omega\left(r_-^2+a^2\right)+\left(k+\frac{1}{2}\right)E_+a+r_-Q=0.
\end{equation}

\end{proof}

\begin{rmk}
    $\omega$ is called the energy eigenvalue of the Dirac equation (\ref{4-1}).
\end{rmk}

Next, we use the equality (\ref{4-50}) derived above to further study the necessary conditions for the existence of nontrivial $L^p$ integrable time-periodic solutions of the Dirac equation (\ref{4-1}) in the extreme Kerr-Newman-dS spacetime. To do this, we quote the following lemma in \cite{reissner}:
\begin{lem}\label{lemma4.2}
    For $x>0$, let $Y(x)$ be a nontrivial solution of the following equation
    \begin{equation}
        \frac{d}{dx}Y(x)=\left[a(x)\begin{pmatrix}
0&-1\\
      1&0      
        \end{pmatrix} +b(x)\begin{pmatrix}
1&0\\
     0&-1       
        \end{pmatrix}+c(x)\begin{pmatrix}
0&1\\
     1&0       
        \end{pmatrix}\right]Y(x)
    \end{equation}
where $a(x)$, $b(x)$ and $c(x)$ are smooth real functions and $a\neq0$. If near the origin,
\begin{equation}
    b(x)^2+c(x)^2<a(x)^2
\end{equation}
   and the functions $\frac{b(x)}{a(x)}$ and $\frac{c(x)}{a(x)}$ are monotone, then there exists constant $\delta>0$ such that
    \begin{equation}
        |Y(x)|\ge \delta
    \end{equation}
near the origin.
\end{lem}

Let $\sigma_1$, $\sigma_2$, $\sigma_3$ be the following constants
\begin{equation}
    \begin{split}
\sigma_1&:= B(r_-)>0,\\
\sigma_2&:= B^{\prime}(r_-),\\
\sigma_3&:= \frac{\omega a^2+\left(k+\frac{1}{2}\right)E_+a}{-r_-}+\omega r_-,
    \end{split}
\end{equation}
where
\begin{equation}
    B(r)=\kappa^2(r_+-r)(r+r_++2r_-).
\end{equation}

\begin{coro}\label{coro4.2}
Let $\Psi$ be a nontrivial time-periodic solution of the Dirac equation (\ref{4-1}) in the extreme  Kerr-Newman-dS spacetime which is of the form (\ref{dsppp}).  If there exists $p\in[1,+\infty)$ such that
\begin{equation}
    \Psi\in L^p\left(M_{(r_-,r_+)}\right),
\end{equation}
then at least one of the following three conditions holds:

\noindent (i) $\left(\epsilon_+^2+\lambda^2r_-^2\right)\sigma_1-\sigma_3^2\ge 0$;

\noindent (ii) $\sigma_2\sigma_3-2\omega\sigma_1=0$;

\noindent (iii) $r_-\sigma_2\sigma_3+2\sigma_1\sigma_3=0$.

\noindent Moreover, if $\epsilon_+=0$, then at least one of the conditions (i) and (iii) holds; if $\lambda=0$, then at least one of the conditions (i) and (ii) holds. In particular, if $\lambda=\epsilon_+=0$, then $Q=-2\omega r_-$.
\end{coro}

\begin{proof}
For $r\in(r_-,r_+)$, the radial function $X(r)$ satisfies the following equation
\begin{equation}
    \frac{d}{dr}X=\begin{pmatrix}
i\alpha_1 & i\beta_1+\gamma_1\\
-i\beta_1+\gamma_1& -i\alpha_1
    \end{pmatrix}X,
\end{equation}
where
\begin{equation}
    \begin{split}
\alpha_1&=\frac{\left(\omega(r^2+a^2)+Qr+\big(k+\frac{1}{2}\big)E_+a\right)}{\Delta_+(r)},\\
\beta_1&=\frac{\lambda r}{\sqrt{\Delta_+(r)}},\\
        \gamma_1&=\frac{\epsilon_+}{\sqrt{\Delta_+(r)}}.
    \end{split}
\end{equation}
Define the function
\begin{equation}
    H(r):=\frac{\sqrt{2}}{2}\begin{pmatrix}
-1 &-1\\
        -i& i
    \end{pmatrix}X(r),
\end{equation}
then for any $r\in (r_-,r_+)$, we have
\begin{equation}
    \left|H(r)\right|=\left|X(r)\right|
\end{equation}
and
\begin{equation}
\begin{split}
    \partial_rH_1(r)&=-\frac{\sqrt{2}}{2}\left(\partial_rX_1(r)+\partial_rX_2(r)\right)\\
&=-\frac{\sqrt{2}}{2}\left(i\alpha_1X_1+(i\beta_1+\gamma_1)X_2+(-i\beta_1+\gamma_1)X_1-i\alpha_1X_2\right)\\
&=\gamma_1H_1(r)+(\alpha_1-\beta_1)H_2(r),
    \end{split}
\end{equation}
\begin{equation}
    \begin{split}
\partial_rH_2(r)&=\frac{\sqrt{2}}{2}\left(-i\partial_rX_1(r)+i\partial_rX_2(r)\right) \\
&=\frac{\sqrt{2}}{2}\left(\alpha_1X_1+\beta_1X_2-i\gamma_1X_2+\beta_1X_1+i\gamma_1X_1+\alpha_1X_2\right)\\
&=-\gamma_1H_2(r)+(-\alpha_1-\beta_1)H_1(r),  
    \end{split}
\end{equation}
i.e.
\begin{equation}
    \partial_rH(r)=\Bigg[-\alpha_1\begin{pmatrix}
0&-1\\
      1&0  
    \end{pmatrix}+\gamma_1\begin{pmatrix}
1&0\\
     0&-1   
    \end{pmatrix}-\beta_1\begin{pmatrix}
0&1\\
     1&0   
    \end{pmatrix} \Bigg]H(r).
\end{equation}
Next, we adopt the method of proof by contradiction.  Suppose the conclusion is not true when $\lambda\epsilon_+\neq0$, that is, the conditions (i), (ii) and (iii) are not valid. First, by (\ref{4-50}) we know that
\begin{equation}
    \left(\epsilon_+^2+\lambda^2r_-^2\right)\sigma_1-\sigma_3^2<0
\end{equation}
i.e.,
\begin{equation}
    \epsilon_+^2+\lambda^2r_-^2<\frac{\left(\omega r_-+\frac{\omega a^2+\left(k+\frac{1}{2}\right)E_+a}{-r_-}\right)^2}{B(r_-)},
\end{equation}
 which means that there exists $\epsilon>0$ (small enough) such that for $r\in(r_-,r_-+\epsilon)$
\begin{equation}
\begin{split}
    \epsilon_+^2+\lambda^2r^2<\frac{\left(\omega r+\frac{\omega a^2+\left(k+\frac{1}{2}\right)E_+a}{-r_-}\right)^2}{B(r)}&=\frac{(r-r_-)^2\left(\omega r+\frac{\omega a^2+\left(k+\frac{1}{2}\right)E_+a}{-r_-}\right)^2}{(r-r_-)^2B(r)}\\
&=\frac{\left(\omega(r^2+a^2)+Qr+\big(k+\frac{1}{2}\big)E_+a\right)^2}{\Delta_+(r)}.
    \end{split}
\end{equation}
Thus we have
\begin{equation}
    (\gamma_1)^2+(\beta_1)^2<(\alpha_1)^2
\end{equation}
for $r\in(r_-,r_-+\epsilon)$. 

When condition (ii) is not satisfied, we have
\begin{equation}
    \left(\omega r_-+\frac{\omega a^2+\left(k+\frac{1}{2}\right)E_+a}{-r_-}\right)\cdot \frac{\sigma_2}{2\sqrt{\sigma_1}}-\sqrt{\sigma_1}\omega\neq 0,
\end{equation}
i.e., 
\begin{equation}
    \left(\omega r_-+\frac{\omega a^2+\left(k+\frac{1}{2}\right)E_+a}{-r_-}\right)\cdot \left(\sqrt{B(r)}\right)^{\prime}\Big|_{r=r_-}-\sqrt{B(r_-)}\omega \neq 0.
\end{equation}
Therefore, by continuity we can deduce that
\begin{equation}
\begin{split}
    \frac{\sqrt{B(r)}}{\omega r+\frac{\omega a^2+\left(k+\frac{1}{2}\right)E_+a}{-r_-}}&=\frac{(r-r_-)\sqrt{B(r)}}{(r-r_-)\left(\omega r+\frac{\omega a^2+\left(k+\frac{1}{2}\right)E_+a}{-r_-}\right)}\\
    &=\frac{\sqrt{\Delta_+(r)}}{\omega(r^2+a^2)+Qr+\big(k+\frac{1}{2}\big)E_+a}
    \end{split}
\end{equation}
is monotone for $r\in(r_-,r_-+\epsilon)$, i.e. $\frac{\gamma_1}{-\alpha_1}$ is monotone for $r\in(r_-,r_-+\epsilon)$. Similarly, we have $\frac{-\beta_1}{-\alpha_1}$ is also monotone for $r\in(r_-,r_-+\epsilon)$. Therefore, by Lemma  \ref{lemma4.2} (or c.f. Lemma 5.1 in \cite{reissner}), there exists constant $\delta>0$, such that for $r\in(r_-,r_-+\epsilon)$, 
\begin{equation}
    \left|X(r)\right|=\left|H(r)\right|\ge \delta>0.
\end{equation}
Hence, the integrability condition
\begin{equation}
    \Psi\in L^p\left(M_{(r_-,r_+)}\right)
\end{equation}
implies that there exists constant $C>0$ such that
\begin{equation}
\int_{M_{(r_-,r_+)}}|\Psi|^pdV>C\int_{r_-}^{r_-+\epsilon}\frac{1}{(r-r_-)^{\frac{p}{2}}}\cdot \frac{1}{r-r_-}dr=+\infty,
\end{equation}
which is a contraction!

If $\lambda \epsilon_+=0$, by repeating the above discussions we can still deduce that such contradiction exists, hence completing the proof of the corollary.

\end{proof}

\mysection{Extreme Kerr-Newman-AdS spacetime}\ls
In this section, we consider the necessary conditions for the existence of nontrivial $L^p$ integrable time-periodic solutions of the Dirac equation
\begin{equation}\label{5-1}
\big(\mathscr{D}+ie^{\alpha}A(e_{\alpha})+i\lambda\big)\Psi=0
\end{equation}
in the extreme Kerr-Newman-AdS spacetime. More specifically, we give the equality between $\omega$, the radius of the event horizon, the angular momentum, the charge, and the cosmological constant, which generalize the conclusion obtained by \cite{extremekerr} in the extreme Kerr-Newman spacetime, i.e. from zero cosmological constant to negative cosmological constant.

The Kerr-Newman-AdS spacetime is an exact solution of the Einstein-Maxwell equation, which describes a charged rotating black hole with a negative cosmological constant. Although it contradicts the recent cosmological observations that our real universe should have a positive cosmological constant, the negative cosmological constant case and the results in this section may have some physical implications for the strongly coupled superconductor theory based on the AdS-CFT  correspondence, see \cite{adscft}. The Kerr-Newman-AdS spacetime is the following manifold
\begin{equation}
    M_{KNdS}=\R_{t}\times \R^+_{r}\times S^2,
\end{equation}
equipped with the Lorentzian metric  (in Boyer-Lindquist coordinate)
\begin{equation}\label{knadsmetric}
\begin{split}
    g_{KNAdS}=&-\Big(1-\frac{2mr-Q^2-P^2}{U}+\kappa^2(r^2+a^2\sin^2\theta) \Big)dt^2\\
    &+\frac{V_-}{U{E_-}^2}\sin^2\theta d\varphi^2
+\frac{U}{\Delta_-(r)}dr^2+\frac{U}{\Delta_-(\theta)}d\theta^2\\
&-\frac{a\sin^2\theta}{E_-}\Big(\frac{2mr-Q^2-P^2}{U}-\kappa^2(r^2+a^2)\Big)\big(dt\,d\varphi+d\varphi\,dt \big)\\
=&-\frac{\Delta_-(r)}{U}\Big(dt-\frac{a\sin^2\theta}{E_-}d\varphi \Big)^2+\frac{U}{\Delta_-(r)}dr^2+\frac{U}{\Delta_-(\theta)}d\theta^2\\
    &\qquad\qquad\qquad+\frac{\Delta_{-}(\theta)\sin^2\theta}{U}\Big(a\,dt-\frac{r^2+a^2}{E_-}d\varphi\Big)^2,
    \end{split}
\end{equation}
where the constants $\kappa>0$, $m>0$, and
\begin{equation}
    \begin{split}
U&=r^2+a^2\cos^2\theta,\\
E_-&=1-\kappa^2a^2>0,\\
\Delta_-(r)&=(r^2+a^2)(1+\kappa^2r^2)-2mr+Q^2+P^2,\\
\Delta_-(\theta)&=1-\kappa^2a^2\cos^2\theta,\\
V_-&=(2mr-Q^2-P^2)a^2\sin^2\theta+U(r^2+a^2)(1-\kappa^2a^2).
    \end{split}
\end{equation}
Moreover, the electromagnetic field is $F=dA$, where $A$ is the following 1-form
\begin{equation}
    A=-\frac{Qr}{U}\Big(dt-\frac{a\sin^2\theta}{E_-}d\varphi\Big)-\frac{P\cos\theta}{U}\Big(a\,dt-\frac{r^2+a^2}{E_-}d\varphi\Big).
\end{equation}

\begin{defn}
    The Kerr-Newman-AdS spacetime is called extreme, if the polynomial of order 4 with respect to $r$
    \begin{equation}
        \Delta_-(r)=(r^2+a^2)(1+\kappa^2r^2)-2mr+Q^2+P^2
    \end{equation}
has a double real root $r=r_e>0$ and 2 imaginary roots. Moreover, $m$ satisfies the following
\begin{equation}
\begin{split}
    m=\frac{1}{\sqrt{54}}&\left(\sqrt{\big(1+a^2\kappa^2\big)^2+12\kappa^2(a^2+Q^2+P^2)}+2a^2\kappa^2+2 \right)\\
&\times\left(\sqrt{\big(1+a^2\kappa^2\big)^2+12\kappa^2(a^2+Q^2+P^2)}-a^2\kappa^2-1\right)^{\frac{1}{2}}.
    \end{split}
\end{equation}
\end{defn}

If the solution $\Psi$ of the Dirac equation (\ref{5-1}) is of the form
\begin{equation}
    \Psi={S_-}^{-1}\Phi,
\end{equation}
where
\begin{equation}
    \Phi=e^{-i\left(\omega t+(k+\frac{1}{2})\varphi\right)}\begin{pmatrix}
X_{-}(r)Y_-(\theta)\\
X_{+}(r)Y_+(\theta)\\
X_{+}(r)Y_-(\theta)\\
X_{-}(r)Y_+(\theta)\\
    \end{pmatrix},
\end{equation}
$k\in\mathbb{Z}$, and $S_-$ is the following diagonal matrix
\begin{equation}
    S_-=\Delta_-(r)^{\frac{1}{4}}\begin{pmatrix}
 (r+ia\cos\theta)^{\frac{1}{2}}& 0&0 & 0\\
0& (r+ia\cos\theta)^{\frac{1}{2}}&0 & 0\\
0& 0&(r-ia\cos\theta)^{\frac{1}{2}} & 0\\
        0& 0&0 & (r-ia\cos\theta)^{\frac{1}{2}}\\
    \end{pmatrix},
\end{equation}
then by the method of separating variables, the radial equations in the extreme Kerr-Newman-AdS spacetime when $\Delta_-(r)>0$ are as follows (c.f. \cite{wyhzx}):
\begin{equation}\label{588}
    \begin{split}
&\frac{dX_+}{dr}-\frac{i}{\Delta_-(r)}\left(\omega(r^2+a^2)+Qr+\big(k+\frac{1}{2}\big)E_-a\right)X_+-\frac{i\lambda r+\eta_+}{\sqrt{\Delta_-(r)}}X_-=0,\\
        &\frac{dX_-}{dr}+\frac{i}{\Delta_-(r)}\left(\omega(r^2+a^2)+Qr+\big(k+\frac{1}{2}\big)E_-a\right)X_-+\frac{i\lambda r-\eta_+}{\sqrt{\Delta_-(r)}}X_+=0,
    \end{split}
\end{equation}
where $\eta_+\in\R$. Moreover,
\begin{equation}
        \frac{d}{dr}\left(\left|X_+\right|^2-\left|X_-\right|^2 \right)=0.
    \end{equation}

Since $\Delta_-(r)$ has a positive double root $r=r_e$, there exists a quadratic irreducible polynomial $B_{-}(r)>0$, satisfying
\begin{equation}
    \Delta_-(r)=(r-r_e)^2B_{-}(r).
\end{equation}
Let
\begin{equation}
\begin{split}
\tau_{-}&\triangleq\omega(r_e^2+a^2)+\left(k+\frac{1}{2}\right)E_-a+r_eQ,\\
    \mu_{-}&\triangleq 2r_e\omega+Q,
\end{split}    
\end{equation}
then we have
\begin{equation}\label{4-8}
\begin{split}
\omega\left((x+r_e)^2+a^2\right)+Q(x+r_e)+\big(k+\frac{1}{2}\big)E_-a&=\omega x^2+\tau_{-}+2\omega r_e x+Qx \\
&=\tau_{-}+\mu_{-}x+\omega x^2.
\end{split}
\end{equation}
We consider the exterior region outside the event horizon, i.e. $r\in(r_e,+\infty)$. For convenience, we define the variable $x:=r-r_e$ and the function
\begin{equation}
    F_{-}(x):=X(x+r_e),\;x\in (0,+\infty).
\end{equation}
Let
\begin{equation}
    \begin{split}
\alpha_2&=\frac{\left(\omega(r^2+a^2)+Qr+\big(k+\frac{1}{2}\big)E_-a\right)}{\Delta_-(r)},\\
\beta_2&=\frac{\lambda r}{\sqrt{\Delta_-(r)}},\\
        \gamma_2&=\frac{\eta_+}{\sqrt{\Delta_-(r)}},
    \end{split}
\end{equation}
then according to  (\ref{588}) and (\ref{4-8}), it follows that
\begin{equation}
\begin{split}
    \partial_x{F_{-}}_1(x)&=\partial_rX_1\Big|_{x+r_e} \\
&=i\alpha_2(x+r_e)X_1(x+r_e)+\left(i\beta_2(x+r_e)+\gamma_2(x+r_e)\right)X_2(x+r_e)\\
&=i\frac{(\tau_{-}+\mu_{-} x+\omega x^2)}{x^2B_{-}(x+r_e)}F_1(x)+\left(\frac{i\lambda (x+r_e)}{\sqrt{x^2B_{-}(x+r_e)}}+\frac{\eta_+}{\sqrt{x^2B_{-}(x+r_e)}}\right)F_2(x),
    \end{split}
\end{equation}
i.e.
\begin{equation}
\begin{split}
    \partial_x{F_{-}}_1(x)&=\left(\frac{i\tau_{-}}{x^2B_{-}(x+r_e)}+\frac{i\mu_{-}}{xB_{-}(x+r_e)}+\frac{i\omega}{B_{-}(x+r_e)}\right){F_{-}}_1(x)\\
&\qquad+\left(\frac{\eta_++i\lambda r_e}{x\sqrt{B_{-}(x+r_e)}}+\frac{i\lambda}{\sqrt{B_{-}(x+r_e)}}\right){F_{-}}_2(x).
    \end{split}
\end{equation}
In a similar way,
\begin{equation}
\begin{split}
\partial_x{F_{-}}_2(x)&=\partial_rX_2\Big|_{x+r_e}\\
&=\left(-i\beta_2(x+r_e)+\gamma_2(x+r_e)\right)X_1(x+r_e)-i\alpha_2(x+r_e)X_2(x+r_e)\\
&=\left(\frac{-i\lambda (x+r_e)}{\sqrt{x^2B_{-}(x+r_e)}}+\frac{\eta_+}{\sqrt{x^2B_{-}(x+r_e)}}\right){F_{-}}_1(x)-i\frac{(\tau_{-}+\mu_{-} x+\omega x^2)}{x^2B_{-}(x+r_e)}{F_{-}}_2(x)\\
&=\left(\frac{\eta_+-i\lambda r_e}{x\sqrt{B_{-}(x+r_e)}}-\frac{i\lambda}{\sqrt{B_{-}(x+r_e)}}\right){F_{-}}_1(x)\\
&\qquad+\left(\frac{-i\tau_{-}}{x^2B_{-}(x+r_e)}+\frac{-i\mu_{-}}{xB_{-}(x+r_e)}+\frac{-i\omega}{B_{-}(x+r_e)}\right){F_{-}}_2(x).
    \end{split}
\end{equation}
After rewriting it in the matrix form, for $x\in(0,+\infty)$, $F_{-}(x)$ satisfies the following equation
\begin{equation}\label{4-13}
    \partial_xF_{-}=\begin{pmatrix}
\frac{i\tau_{-}}{x^2B_{-}(x+r_e)}+\frac{i\mu_{-}}{xB_{-}(x+r_e)}+\frac{i\omega}{B_{-}(x+r_e)} & \frac{\eta_++i\lambda r_e}{x\sqrt{B_{-}(x+r_e)}}+\frac{i\lambda}{\sqrt{B_{-}(x+r_e)}} \\
 \frac{\eta_+-i\lambda r_e}{x\sqrt{B_{-}(x+r_e)}}-\frac{i\lambda}{\sqrt{B_{-}(x+r_e)}} &  \frac{-i\tau_{-}}{x^2B_{-}(x+r_e)}+\frac{-i\mu_{-}}{xB_{-}(x+r_e)}+\frac{-i\omega}{B_{-}(x+r_e)}    
    \end{pmatrix}F_{-}.
\end{equation}

Now we define the function $W_{-}(x)$ as
\begin{equation}\label{521}
    W_{-}(x):=T\cdot F_{-}(x),\;x\in(0,+\infty),
\end{equation}
where $T$ is the unitary matrix
\begin{equation}
    T=\begin{pmatrix}
\frac{-1}{\sqrt{2}}& \frac{-1}{\sqrt{2}}\\
        \frac{-i}{\sqrt{2}}& \frac{i}{\sqrt{2}}
    \end{pmatrix},
\end{equation}
then for any $x\in(0,+\infty)$ we have
\begin{equation}
    \left|W_{-}(x)\right|=\left|F_{-}(x)\right|.
\end{equation}
Since $F_{-}$ satisfies the equation (\ref{4-13}), it follows that
\begin{equation}
\begin{split}
    \partial_x{W_{-}}_1(x)&=-\frac{1}{\sqrt{2}}\left(\partial_x{F_{-}}_1+\partial_x{F_{-}}_2\right)\\
&=-\frac{1}{\sqrt{2}}\Bigg[\left(\frac{i\tau_{-}}{x^2B_{-}}+\frac{i\mu_{-}}{xB_{-}}+\frac{i\omega}{B_{-}}\right){F_{-}}_1+\left(\frac{\eta_++i\lambda r_e}{x\sqrt{B_{-}}}+\frac{i\lambda}{\sqrt{B_{-}}}\right){F_{-}}_2\\
&\qquad+\left(\frac{\eta_+-i\lambda r_e}{x\sqrt{B_{-}}}-\frac{i\lambda}{\sqrt{B_{-}}}\right){F_{-}}_1-\left(\frac{i\tau_{-}}{x^2B_{-}}+\frac{i\mu_{-}}{xB_{-}}+\frac{i\omega}{B_{-}}\right){F_{-}}_2\Bigg].
    \end{split}
\end{equation}
Therefore, substituting (\ref{521}), we obtain that
\begin{equation}
\begin{split}
    \partial_x{W_{-}}_1(x)&=\frac{\eta_+}{x\sqrt{B_{-}}}\cdot \frac{-1}{\sqrt{2}}\left({F_{-}}_1+{F_{-}}_2\right)-\frac{i}{\sqrt{2}}\left(\frac{\tau_{-}}{x^2B_{-}}+\frac{\mu_{-}}{xB_{-}}+\frac{\omega}{B_{-}}\right)\left({F_{-}}_1-{F_{-}}_2\right)\\
&\qquad -\frac{i}{\sqrt{2}}\left(\frac{\lambda r_e}{x\sqrt{B_{-}}}+\frac{\lambda}{\sqrt{B_{-}}}\right)\left({F_{-}}_2-{F_{-}}_1\right)\\
&=\frac{\eta_+}{x\sqrt{B_{-}}}{W_{-}}_1(x)+\left(\frac{\tau_{-}}{x^2B_{-}}+\frac{\mu_{-}}{xB_{-}}+\frac{\omega}{B_{-}}-\frac{\lambda r_e}{x\sqrt{B_{-}}}-\frac{\lambda}{\sqrt{B_{-}}}\right){W_{-}}_2(x).
\end{split}
\end{equation}
Similarly,
\begin{equation}
\begin{split}
    \partial_x{W_{-}}_2(x)&=\frac{-i}{\sqrt{2}}\left(\partial_x{F_{-}}_1-\partial_x{F_{-}}_2\right)\\
&=\frac{-i}{\sqrt{2}}\Bigg[\left(\frac{i\tau_{-}}{x^2B_{-}}+\frac{i\mu_{-}}{xB_{-}}+\frac{i\omega}{B_{-}}\right){F_{-}}_1+\left(\frac{\eta_++i\lambda r_e}{x\sqrt{B_{-}}}+\frac{i\lambda}{\sqrt{B_{-}}}\right){F_{-}}_2\\
&\qquad-\left(\frac{\eta_+-i\lambda r_e}{x\sqrt{B_{-}}}-\frac{i\lambda}{\sqrt{B_{-}}}\right){F_{-}}_1+\left(\frac{i\tau_{-}}{x^2B_{-}}+\frac{i\mu_{-}}{xB_{-}}+\frac{i\omega}{B_{-}}\right){F_{-}}_2\Bigg]\\
&=\frac{-i}{\sqrt{2}}\cdot\frac{-\eta_+}{x\sqrt{B_{-}}}\left({F_{-}}_1-{F_{-}}_2\right)+\frac{1}{\sqrt{2}}\left(\frac{\tau_{-}}{x^2B_{-}}+\frac{\mu_{-}}{xB_{-}}+\frac{\omega}{B_{-}}\right)\left({F_{-}}_1+{F_{-}}_2\right)\\
&\qquad +\frac{1}{\sqrt{2}}\left(\frac{\lambda r_e}{x\sqrt{B_{-}}}+\frac{\lambda}{\sqrt{B_{-}}}\right)\left({F_{-}}_1+{F_{-}}_2\right)\\
&=\left(\frac{-\tau_{-}}{x^2B_{-}}+\frac{-\mu_{-}}{xB_{-}}+\frac{-\omega}{B_{-}}-\frac{\lambda r_e}{x\sqrt{B_{-}}}-\frac{\lambda}{\sqrt{B_{-}}}\right){W_{-}}_1(x)+\frac{-\eta_+}{x\sqrt{B_{-}}}{W_{-}}_2(x).
    \end{split}
\end{equation}
Thus, $W_{-}(x)$ satisfies the following equation
\begin{equation}\label{4-20}
    \partial_xW_{-}(x)=\begin{pmatrix}
\frac{\eta_+}{x\sqrt{B_{-}}} & \frac{\tau_{-}}{x^2B_{-}}+\frac{\mu_{-}}{xB_{-}}+\frac{\omega}{B_{-}}-\frac{\lambda r_e+\lambda x}{x\sqrt{B_{-}}}\\
   -\frac{\tau_{-}}{x^2B_{-}}-\frac{\mu_{-}}{xB_{-}}-\frac{\omega}{B_{-}}-\frac{\lambda r_e+\lambda x}{x\sqrt{B_{-}}} & \frac{-\eta_+}{x\sqrt{B_{-}}}
    \end{pmatrix}W_{-}(x),
\end{equation}
where $B_{-}\triangleq B_{-}(x+r_e)$.

For $x\in(0,1]$, we take a new variable substitution $z:=\frac{1}{x}$, $z\in[1,+\infty)$. Define the function
\begin{equation}
    V_{-}(z):=W_{-}\left(\frac{1}{z}\right),\;z\in[1,+\infty).
\end{equation}
Thus,
\begin{equation}
    \partial_zV_{-}(z)=-\frac{1}{z^2}W^{\prime}_{-}\Big|_{\frac{1}{z}}.
\end{equation}
Substituting (\ref{4-20}), we have
\begin{equation}\label{528}
    \partial_zV_-(z)=-\frac{1}{z^2}\begin{pmatrix}
\widetilde{E_-}_{11}(z) & \widetilde{E_-}_{12}(z) \\
        \widetilde{E_-}_{21}(z) & \widetilde{E_-}_{22}(z)
    \end{pmatrix}W_-\left(\frac{1}{z}\right),
\end{equation}
where
\begin{equation}
    \begin{split}
\widetilde{E_-}_{11}(z)&=\frac{\eta_+z}{\sqrt{B_-\left(\frac{1}{z}+r_e\right)}}, \\
\widetilde{E_-}_{12}(z)&=\frac{\tau_- z^2}{B_-\left(\frac{1}{z}+r_e\right)}+\frac{\mu_- z}{B_-\left(\frac{1}{z}+r_e\right)}+\frac{\omega}{B_-\left(\frac{1}{z}+r_e\right)}-\frac{\lambda r_e z}{\sqrt{B_-\left(\frac{1}{z}+r_e\right)}}\\
&\qquad\qquad\qquad\qquad\qquad\qquad\qquad\qquad\qquad-\frac{\lambda}{\sqrt{B_-\left(\frac{1}{z}+r_e\right)}}, \\
\widetilde{E_-}_{21}(z)&=-\frac{\tau_- z^2}{B_-\left(\frac{1}{z}+r_e\right)}-\frac{\mu_- z}{B_-\left(\frac{1}{z}+r_e\right)}-\frac{\omega}{B_-\left(\frac{1}{z}+r_e\right)}-\frac{\lambda r_e z}{\sqrt{B_-\left(\frac{1}{z}+r_e\right)}}\\
&\qquad\quad\quad\qquad\qquad\qquad\;\;\qquad\qquad\qquad\qquad-\frac{\lambda}{\sqrt{B_-\left(\frac{1}{z}+r_e\right)}}, \\
\widetilde{E_-}_{22}(z)&=-\frac{\eta_+z}{\sqrt{B_-\left(\frac{1}{z}+r_e\right)}}.
    \end{split}
\end{equation}
By sorting (\ref{528}), the function $V_-(z)$ shall satisfy the following equation
\begin{equation}\label{4-25}
    \partial_zV_{-}(z)=\begin{pmatrix}
{E_{-}}_{11}(z) & {E_{-}}_{12}(z) \\
        {E_{-}}_{21}(z) & {E_{-}}_{22}(z)
    \end{pmatrix}V_{-}(z),\;z\in [1,+\infty),
\end{equation}
where
\begin{equation}
    \begin{split}
{E_{-}}_{11}(z)&=\frac{-\eta_+}{z\sqrt{B_{-}\left(\frac{1}{z}+r_e\right)}}, \\
{E_{-}}_{12}(z)&=-\frac{\tau_{-}}{B_{-}\left(\frac{1}{z}+r_e\right)}-\frac{\mu_{-} }{zB_{-}\left(\frac{1}{z}+r_e\right)}-\frac{\omega}{z^2 B_{-}\left(\frac{1}{z}+r_e\right)}+\frac{\lambda r_e }{z\sqrt{B_{-}\left(\frac{1}{z}+r_e\right)}}\\
&\qquad\qquad\qquad\qquad\qquad\qquad\qquad\qquad\qquad\qquad +\frac{\lambda}{z^2\sqrt{B_{-}\left(\frac{1}{z}+r_e\right)}}, \\
{E_{-}}_{21}(z)&=\frac{\tau_{-}}{B_{-}\left(\frac{1}{z}+r_e\right)}+\frac{\mu_{-} }{zB_{-}\left(\frac{1}{z}+r_e\right)}+\frac{\omega}{z^2 B_{-}\left(\frac{1}{z}+r_e\right)}+\frac{\lambda r_e }{z\sqrt{B_{-}\left(\frac{1}{z}+r_e\right)}}\\
&\qquad\quad\;\qquad\qquad\qquad\qquad\qquad\qquad\qquad\qquad +\frac{\lambda}{z^2\sqrt{B_{-}\left(\frac{1}{z}+r_e\right)}},\\
{E_{-}}_{22}(z)&=\frac{\eta_+}{z\sqrt{B_{-}\left(\frac{1}{z}+r_e\right)}}.
\end{split}
\end{equation}

Let $M_{(r_e,+\infty)}$ be the time slice in the extreme Kerr-Newman-AdS spacetime satisfying $\left\{t=\text{constant}\right\}$ and $r>r_e$. With the above preparations, we can now prove the following necessary condition for $\omega$:

\begin{thm}\label{thm5.1}
Let $\Psi$ be the nontrivial solution of the Dirac equation
\begin{equation}
\big(\mathscr{D}+ie^{\alpha}A(e_{\alpha})+i\lambda\big)\Psi=0
\end{equation}
on the exterior region $r>r_e$ in the extreme Kerr-Newman-AdS spacetime and it is of the form
\begin{equation}\label{adsppp}
    \Psi={S_-}^{-1}\Phi,
\end{equation}
where
\begin{equation}
    \Phi=e^{-i\left(\omega t+(k+\frac{1}{2})\varphi\right)}\begin{pmatrix}
X_{-}(r)Y_-(\theta)\\
X_{+}(r)Y_+(\theta)\\
X_{+}(r)Y_-(\theta)\\
X_{-}(r)Y_+(\theta)\\
    \end{pmatrix},
\end{equation}
$k\in\mathbb{Z}$ and $S_-$ is the following diagonal matrix
\begin{equation}
    S_-=\Delta_-(r)^{\frac{1}{4}}\begin{pmatrix}
 (r+ia\cos\theta)^{\frac{1}{2}}& 0&0 & 0\\
0& (r+ia\cos\theta)^{\frac{1}{2}}&0 & 0\\
0& 0&(r-ia\cos\theta)^{\frac{1}{2}} & 0\\
        0& 0&0 & (r-ia\cos\theta)^{\frac{1}{2}}\\
    \end{pmatrix}.
\end{equation}
If there exists $p\in[1,+\infty)$ such that
\begin{equation}
    \Psi\in L^p\left(M_{(r_e,+\infty)}\right),
\end{equation}
then $\omega$ satisfies the following equality
\begin{equation}\label{4-33}
    \omega\left(r_e^2+a^2\right)+\left(k+\frac{1}{2}\right)E_-a+r_eQ=0.
\end{equation}
\end{thm}

\begin{proof}
Assume that
\begin{equation}
\tau_{-}=\omega\left(r_e^2+a^2\right)+\left(k+\frac{1}{2}\right)E_-a+r_eQ\neq 0.
\end{equation}
We rewrite the equation (\ref{4-25}) as follows
\begin{equation}
    \partial_zV_{-}(z)=\big(C_{-}+R_{-}(z)\big)V_{-}(z),\;z\in [1,+\infty),
\end{equation}
where
\begin{equation}
    C_{-}=\begin{pmatrix}
0 & -\frac{\tau}{B_{-}(r_e)}\\
        \frac{\tau}{B_{-}(r_e)} &0
    \end{pmatrix}
\end{equation}
is a constant matrix and the 4 components of the $2\times 2$ matrix $R_{-}(z)$ are
\begin{equation}\label{4-31}
    \begin{split}
{R_{-}}_{11}(z)&=\frac{-\eta_+}{z\sqrt{B_{-}\left(\frac{1}{z}+r_e\right)}}, \\
{R_{-}}_{12}(z)&=\left(\frac{\tau_{-}}{B_{-}(r_e)}-\frac{\tau_{-}}{B_{-}\left(\frac{1}{z}+r_e\right)}\right)-\frac{\mu_{-}}{zB_{-}\left(\frac{1}{z}+r_e\right)}-\frac{\omega}{z^2 B_{-}\left(\frac{1}{z}+r_e\right)}\\
&\qquad\qquad+\frac{\lambda r_e }{z\sqrt{B_{-}\left(\frac{1}{z}+r_e\right)}}+\frac{\lambda}{z^2\sqrt{B_{-}\left(\frac{1}{z}+r_e\right)}}, \\
{R_{-}}_{21}(z)&=\left(\frac{\tau_{-}}{B_{-}\left(\frac{1}{z}+r_e\right)}-\frac{\tau_{-}}{B_{-}(r_e)}\right)+\frac{\mu_{-}}{zB_{-}\left(\frac{1}{z}+r_e\right)}+\frac{\omega}{z^2 B_{-}\left(\frac{1}{z}+r_e\right)}\\
&\qquad\qquad+\frac{\lambda r_e }{z\sqrt{B_{-}\left(\frac{1}{z}+r_e\right)}}+\frac{\lambda}{z^2\sqrt{B_{-}\left(\frac{1}{z}+r_e\right)}}, \\
{R_{-}}_{22}(z)&=\frac{\eta_+}{z\sqrt{B_{-}\left(\frac{1}{z}+r_e\right)}}. 
    \end{split}
\end{equation}
Since $\tau_{-}\neq 0$, we have
\begin{equation}
    \det C_{-}=0-\frac{-\tau_{-}^2}{B_{-}(r_e)^2}>0.
\end{equation}
Moreover, according to the expressions (\ref{4-31}) of ${R_{-}}_{ij}(z)$,  it is no hard to see that:

\noindent (i) ${R^{\prime}_{-}}_{ij}(z)$ is integrable on $[1,+\infty)$上是可积的;

\noindent (ii) 
\begin{equation}
    {R_{-}}_{ij}(z)\longrightarrow 0,\;1\le i,j\le 2
\end{equation}
as $z\rightarrow +\infty$.

\noindent (iii) $\text{tr}\,(C_{-}+R_{-})=0$.

\noindent Since $\Psi$ is nontrivial, we have $V_{-}\neq0$ (otherwise $X=0$ and $\Psi=0$). Therefore, by Lemma \ref{lemma41} (or c.f. Lemma 3.1 in \cite{extremekerr}),  there exists a constant $\delta_{-}>0$ such that for all $z\in[1,+\infty)$, 
\begin{equation}
    |V_{-}(z)|\ge \delta_{-}>0,
\end{equation}
i.e. for any $x\in(0,1]$ we have
\begin{equation}
    |W_{-}(x)|\ge \delta_{-}>0.
\end{equation}
Since
\begin{equation}
    |W_{-}(x)|=|F_{-}(x)|,
\end{equation}
we have
\begin{equation}\label{4-43}
    |X(r)|\ge \delta_{-}>0
\end{equation}
for $r\in (r_e,r_e+1]$.

On the other hand, the integrability condition
\begin{equation}
    \Psi\in L^p\left(M_{(r_e,+\infty)}\right)
\end{equation}
implies that
\begin{equation}
    \int_{M_{(r_e,+\infty)}}\left(\frac{1}{\sqrt{U}}\Delta_-(r)^{-\frac{1}{2}}\right)^{\frac{p}{2}}|Y|^{p}|X|^{p}\sqrt{\frac{UV_-\sin^2\theta}{E_-^2\Delta_-(r)\Delta_-(\theta)}}dr\,d\theta\,d\varphi<\infty.
\end{equation}
Moreover, since there exists a constant $C>0$ such that $|Y|^2=|Y_+|^2+|Y_-|^2>C$ on $[\frac{\pi}{4},\frac{\pi}{2}]$ (otherwise $\Psi\equiv 0$), we have
\begin{equation}
    \int_{r_e}^{+\infty}\int_{\frac{\pi}{4}}^{\frac{\pi}{2}}\int_{0}^{2\pi}\left(\frac{1}{\sqrt{U}}\Delta_-(r)^{-\frac{1}{2}}\right)^{\frac{p}{2}}|X|^{p}\sqrt{\frac{UV_-\sin^2\theta}{E_-^2\Delta_-(r)\Delta_-(\theta)}}dr\,d\theta\,d\varphi<\infty.
\end{equation}
Combining with (\ref{4-43}), we can infer that there exists a constant $C_1>0$ such that
\begin{equation}
\int_{M_{(r_e,+\infty)}}|\Psi|^pdV>C_1\int_{r_e}^{r_e+1}\frac{1}{(r-r_e)^{\frac{p}{2}}}\cdot \frac{1}{r-r_e}dr=+\infty,
\end{equation}
which is a contraction! Hence we have
\begin{equation}
\tau_{-}=\omega\left(r_e^2+a^2\right)+\left(k+\frac{1}{2}\right)E_-a+r_eQ=0.
\end{equation}

\end{proof}

By the similar proceduce as in the proof of Corollary \ref{coro4.2} in Section 4, combining with the equality (\ref{4-33}) obtained in Theorem \ref{thm5.1}, we can further obtain the necessary conditions for the existence of nontrivial $L^p$ integrable time-periodic solutions of the Dirac equation (\ref{5-1}) in the extreme Kerr-Newman-AdS spacetime.

Let $\zeta_1$, $\zeta_2$ and $\zeta_3$ be the following constants
\begin{equation}
    \begin{split}
\zeta_1&:= B_{-}(r_e)>0,\\
\zeta_2&:= B^{\prime}_{-}(r_e),\\
\zeta_3&:= \frac{\omega a^2+\left(k+\frac{1}{2}\right)E_-a}{-r_e}+\omega r_e,
    \end{split}
\end{equation}
where
\begin{equation}
    (r-r_e)^2B_{-}(r)=\Delta_-(r).
\end{equation}

\begin{coro}
Let $\Psi$ be a nontrivial time-periodic solution of the Dirac equation (\ref{5-1}) taking (\ref{adsppp}). If there exists $p\in[1,+\infty)$ such that
\begin{equation}
    \Psi\in L^p\left(M_{(r_e,+\infty)}\right),
\end{equation}
then at least one of the following three conditions holds:

\noindent (i) $\left(\eta_+^2+\lambda^2r_e^2\right)\zeta_1-\zeta_3^2\ge 0$;

\noindent (ii) $\zeta_2\zeta_3-2\omega\zeta_1=0$;

\noindent (iii) $r_e\zeta_2\zeta_3+2\zeta_1\zeta_3=0$.

\noindent Moreover, if $\eta_+=0$, then at least one of the conditions (i) and (iii) holds; if $\lambda=0$, then at least one of the conditions (i) and (ii) holds.  In particular, if $\lambda=\eta_+=0$, then $Q=-2\omega r_e$.

\end{coro}

\mysection{Conclusion and future work}\ls

In this paper, we study the nonexistence of nontrivial time-periodic solutions of the Dirac equation in Kerr-Newman-(A)dS spacetime. For non-extreme Kerr-Newman-dS spacetime, we prove that there is no $L^p$ integrable Dirac particle for arbitrary $(\lambda,p)\in \R \times[2,+\infty)$. For the extreme Kerr-Newman-dS and extreme Kerr-Newman-AdS spacetime, we prove that if the Dirac equation has a nontrivial $L^p$ integrable time-periodic solution, then the energy eigenvalue $\omega$ and the parameters of the spacetime should satisfy the following equations
\begin{equation}\label{6-1}
    \begin{split}
\omega(r_-^2+a^2)+\left(k+\frac{1}{2}\right)E_{+}a+Qr_-&=0,\\
        \omega(r_e^2+a^2)+\left(k+\frac{1}{2}\right)E_{-}a+Qr_e&=0,
    \end{split}
\end{equation}
respectively. Furthermore, by (\ref{6-1}), we further show the necessary conditions for the existence of nontrivial $L^p$ integrable time-periodic solutions of the Dirac equation. Combining with the existing works, we list the following problems to be further studied:

\noindent(1): If there exists nontrivial $L^p$ integrable time-periodic solution of the Dirac equation in the exterior region of the non-extreme Kerr-Newman-dS spacetime for $1<p<2$ ?

\noindent(2): If there exists nontrivial normalizable time-periodic Dirac particle with mass less than or equal to $\frac{\kappa}{2}$ ?

\bigskip

\bibliographystyle{plain}
\bibliography{IOPEXPORTBIB}

\begin{thebibliography}{10}

\bibitem{akcay}
S~Akcay and R~A. Matzner.
\newblock The {K}err-de {S}itter universe.
\newblock {\em Classical and Quantum Gravity}, 28:085012, 2011.

\bibitem{bahouri}
H~Bahouri, J-Y Chemin, and R~Danchin.
\newblock {\em Fourier analysis and nonlinear partial differential equations}.
\newblock Springer Berlin, Heidelberg, 2011.

\bibitem{bcds}
F~Belgiorno and S~L. Cacciatori.
\newblock The absence of normalizable time-periodic solutions for the {D}irac
  equation in the {K}err-{N}ewman-d{S} black hole background.
\newblock {\em Journal of Physics A: Mathematical and Theoretical},
  42(13):135207, 2009.

\bibitem{bcads}
F~Belgiorno and S~L. Cacciatori.
\newblock The {D}irac equation in {K}err-{N}ewman-{A}d{S} black hole
  background.
\newblock {\em Journal of Mathematical Physics}, 51:033517, 2010.

\bibitem{bicht}
K~Bichteler.
\newblock Global existence of spin structures for gravitational fields.
\newblock {\em Journal of Mathematical Physics}, 9:813--815, 1968.

\bibitem{adscft}
R~G Cai, L~Li, L~F Li, and R~Q Yang.
\newblock Introduction to holographic superconductor models.
\newblock {\em Science China Physics, Mechanics \& Astronomy}, 58:1--46, 2015.

\bibitem{chand3}
S~Chandrasekhar.
\newblock The solution of {D}irac's equation in {K}err geometry.
\newblock {\em Proceedings of Royal Society of London Series A-Mathematical and
  Physics Sciences}, 349:571--575, 1976.

\bibitem{chand2}
S~Chandrasekhar.
\newblock {\em The {M}athematical theory of black hole: {R}evised reprint of
  the 1983 original}.
\newblock The {C}larendon {P}ress, Oxford {U}niversity {P}ress, New York, 1992.

\bibitem{faria}
E~de~Faria and W~de~Melo.
\newblock {\em Mathematical aspects of quantum field theory}.
\newblock Cambridge University Press, 2010.

\bibitem{reissner}
F~Finster, J~Smoller, and S~T. Yau.
\newblock Non-existence of time-periodic solutions of the {D}irac equation in a
  {R}eissner-{N}ordstr\"{o}m black hole background.
\newblock {\em Journal of Mathematical Physics}, 41(4):2173--2194, 2000.

\bibitem{kn}
F~Finster, J~Smoller, and S~T. Yau.
\newblock Non-existence of time-periodic solutions of the {D}irac equation in
  an axisymmetric black hole geometry.
\newblock {\em Communications on Pure and Applied Mathematics}, 53:902--929,
  2000.

\bibitem{frankel}
T~Frankel.
\newblock {\em The geometry of physics: An introduction}.
\newblock Cambridge University Press, 2012.

\bibitem{hall}
B~C. Hall.
\newblock {\em Quantum theory for mathematicians}.
\newblock Springer New York, 2013.

\bibitem{hassani}
S~Hassani.
\newblock {\em Mathematical physics: {A} Modern introduction to its
  foundations}.
\newblock Springer Cham, 2013.

\bibitem{helgason}
S~Helgason.
\newblock {\em Differential geometry and symmetric spaces}.
\newblock Academic Press, New York, 1962.

\bibitem{hijazi}
O~Hijazi.
\newblock {\em Spectral properties of the {D}irac operator and geometrical
  structures}.
\newblock World Scientific, Colombia, 2001.

\bibitem{14}
N~Kamran and R.~G. McLenaghan.
\newblock Separation of variables and symmetry operators for the neutrino and
  {D}irac equations in the space-times admitting a two-parameter abelian
  orthogonally transitive isometry group and a pair of shearfree geodesic null
  congruences.
\newblock {\em Journal of Mathematical Physics}, 25(4):1019--1027, 1984.

\bibitem{kobayashi}
S~Kobayashi and K~Nomizu.
\newblock {\em Foundations of differential geometry, {V}olume 1}.
\newblock Interscience, New York, 1963.

\bibitem{lawson}
H~B. Lawson and M-L. Michlsohn.
\newblock {\em Spin geometry}.
\newblock Princeton University Press, Princeton, NJ, 1989.

\bibitem{lcbls}
C.B. Liang and B.~Zhou.
\newblock {\em Differential geometry and general relativity, Volume 1}.
\newblock Springer Singapore, 2023.

\bibitem{jose}
J~Nat\'{a}rio.
\newblock {\em An introduction to mathematical relativity}.
\newblock Springer Cham, 2021.

\bibitem{neill}
B~O'Neill.
\newblock {\em The geometry of {K}err black holes}.
\newblock A K Peters, Ltd, Wellesley, MA, 1995.

\bibitem{page}
D~Page.
\newblock Dirac equation around a charged, rotating black hole.
\newblock {\em Physical Review D}, 14(6):1509--1510, 1976.

\bibitem{extremekerr}
H~Schmid.
\newblock Bound state solutions of the {D}irac equation in the extreme {K}err
  geometry.
\newblock {\em Mathematische Nachrichten}, 274-275:117--129, 2004.

\bibitem{teschl}
G~Teschl.
\newblock {\em Ordinary differential equations and dynamic systems}.
\newblock American {M}athematical Society, Rhode Island, 2012.

\bibitem{loring}
W~Tu.
\newblock {\em Differential geometry: {C}onnections, curvature, and
  characteristic class}.
\newblock Springer Cham, 2017.

\bibitem{10}
R~M. Wald.
\newblock {\em General relativity}.
\newblock The University of Chicago Press, Chicago, 1984.

\bibitem{walter}
W~Walter.
\newblock {\em Ordinary differential equations}.
\newblock Springer New York, New York, 1998.

\bibitem{wyhzx}
Y~H. Wang and X~Zhang.
\newblock Nonexistence of time-periodic solutions of the {D}irac equation in
  non-extreme {K}err-{N}ewman-{A}d{S} spacetime.
\newblock {\em Science China Mathematics}, 61(1):73--82, 2018.

\bibitem{wein}
S~Weinberg.
\newblock {\em Gravitation and cosmology: {P}rinciples and applications of the
  general theory of relativity}.
\newblock Wiley, New York, 1972.

\bibitem{xz}
X~Zhang.
\newblock A new quasi-local mass and positivity.
\newblock {\em Acta Mathematica Sinica, English Series}, 24:881--890, 2008.

\end{thebibliography}

\end{CJK}
\end{document}